\documentclass[acmsmall,screen,dvipsnames]{acmart}
\pdfoutput=1

\makeatletter
\def\@acmplainindent{0pt}
\def\@acmdefinitionindent{0pt}
\def\@proofindent{\noindent}
\makeatother

\setcitestyle{nosort}

\setcopyright{rightsretained}
\acmDOI{10.1145/3649821}
\acmYear{2024}
\copyrightyear{2024}
\acmSubmissionID{oopslaa24main-p29-p}
\acmJournal{PACMPL}
\acmVolume{8}
\acmNumber{OOPSLA1}
\acmArticle{104}
\acmMonth{4}
\received{2023-10-12}
\received[accepted]{2024-02-24}

\AtBeginDocument{%
  \providecommand\BibTeX{{%
    \normalfont B\kern-0.5em{\scshape i\kern-0.25em b}\kern-0.8em\TeX}}}

\usepackage{graphicx}				
					
\usepackage{amsfonts}
\usepackage{amsmath}
\usepackage{amsthm}
\usepackage{algpseudocode}
\usepackage{algorithm}
\usepackage{pifont}
\usepackage[inline]{enumitem}
\usepackage{eqparbox}
\usepackage{hyperref}
\usepackage{listings}
\usepackage{mathpartir}
\usepackage{mathtools}
\usepackage{multicol}
\usepackage{multirow}
\usepackage{relsize}
\usepackage{scalerel}
\usepackage{stmaryrd}
\usepackage{thm-restate}
\usepackage{tikz}
\usepackage{wrapfig}
\usepackage{xfrac}
\usepackage{xspace}
\usepackage{xr}
\usepackage[all, cmtip]{xy}

\usetikzlibrary{decorations.pathreplacing,calligraphy}

\usepackage{cleveref}

\def\extended{0}
\usepackage{olmacros}

\begin{document}


\title{Outcome Separation Logic: Local Reasoning for Correctness and Incorrectness with Computational Effects}


\author{Noam Zilberstein}
\email{noamz@cs.cornell.edu}
\orcid{0000-0001-6388-063X}
\affiliation{%
  \institution{Cornell University}
  \country{USA}
}

\author{Angelina Saliling}
\email{ajs649@cornell.edu}
\orcid{0009-0000-6277-435X}
\affiliation{%
  \institution{Cornell University}
  \country{USA}
}

\author{Alexandra Silva}
\email{alexandra.silva@cornell.edu}
\orcid{0000-0001-5014-9784}
\affiliation{%
  \institution{Cornell University}
  \country{USA}
}


\begin{abstract}
Separation logic's compositionality and local reasoning properties have led to significant advances in scalable static analysis.
But program analysis has new challenges---many programs display \emph{computational effects} and, orthogonally, static analyzers must handle \emph{incorrectness} too. We present Outcome Separation Logic (OSL), a program logic that is sound for both correctness and incorrectness reasoning in programs with varying effects. OSL has a frame rule---just like separation logic---but uses different underlying assumptions that open up local reasoning to a larger class of properties than can be handled by any single existing logic.

Building on this foundational theory, we also define symbolic execution algorithms that use bi-abduction to derive specifications for programs with effects.
This involves a new \emph{tri-abduction} procedure to analyze programs whose execution branches due to effects such as nondeterministic or probabilistic choice.
This work furthers the compositionality promised by separation logic by opening up the possibility for greater reuse of analysis tools across two dimensions: bug-finding vs verification in programs with varying effects.
\end{abstract}




\maketitle

\section{Introduction}
\label{sec:intro}

Compositional reasoning using separation logic \cite{sl} and bi-abduction~\cite{biab} has helped scale static analysis to industrial software with hundreds of millions of lines of code, making it possible to analyze code changes without disrupting the fast-paced engineering culture that developers are accustomed to~\cite{infer,distefano2019scaling}.

While the ideal of fully automated program verification remains elusive, analysis tools can boost confidence in code correctness by ensuring that a program will not \emph{go wrong} in a variety of ways. In languages like C or C++, this includes ensuring that a program will not crash due to a segmentation fault or leak memory. However, a static analyzer \emph{failing} to prove the absence of bugs does not imply that the program is incorrect; it could be a false positive.

Many programs are, in fact, incorrect. Analysis tools capable of finding
bugs are thus in some cases \emph{more useful} than verification tools, since the reported errors lead directly to tangible code improvements~\cite{realbugs}. Motivated by the need to identify bugs, Incorrectness Logic \cite{il} and Incorrectness Separation Logic (ISL) \cite{isl,cisl} were recently introduced.

While ISL enjoys compositionality just like separation logic, the semantics of SL and ISL are incompatible---specifications and analysis tools cannot readily be shared between them. Further, the soundness of local reasoning in each separation logic variant relies on particular assumptions, meaning that no single program logic has a frame rule that can handle all of the following:




\heading{Computational Effects.}
The idea that program commands must be \emph{local actions} \cite{calcagno2007local} is central to the frame rule. In order to achieve locality, the semantics of memory allocation is often forced to be nondeterministic \cite{yang2002semantic}, but this approach is not suitable for local reasoning in alternative execution models such as probabilistic computation.

\heading{May and Must properties.}
Separation logic can only express properties that \emph{must} occur, whereas ISL can only express properties that \emph{may} occur. A logic for both correctness and incorrectness properties must handle both \emph{may} and \emph{must} properties.

\heading{Under-approximation}. Bug-finding static analyses operate more efficiently by only inspecting a subset of the program paths---the ones that lead to a bug. But many separation logics are \emph{fault avoiding}; the precondition must specify the resources needed to execute all traces, and therefore the entire program must be inspected. Fault avoidance is unsuitable in logics for incorrectness.

\smallskip

\noindent
In this paper, we show that local reasoning is sound under new assumptions, which do not force any particular evaluation model and are therefore compatible with both correctness and incorrectness reasoning. To that end, we introduce \emph{Outcome Separation Logic} (OSL), a single program logic that can handle all of the aforementioned scenarios.
Our contributions are as follows.

\heading{Outcome Separation Logic}. 
Our main contribution is Outcome Separation Logic (OSL), an extension of the recently introduced Outcome Logic (OL) \cite{outcome} for reasoning about pointer programs. While OL already supports correctness and incorrectness reasoning in programs with a variety of effects and much of the metatheory and inference rules carry over from it, the key novelty of OSL is that it is fundamentally based on separation-logic-style heap assertions, which allowed us to develop a frame rule for local reasoning and compositional symbolic execution.

This was no simple feat; while the soundness of the frame rule is intricate and nuanced even in partial correctness Hoare Logic \cite{yang2002semantic}, moving to a logic that supports a variety of assertions about termination, reachability, and probabilistic reasoning complicates the story significantly. OSL is the first program logic that supports local reasoning for both correctness and incorrectness, with the ability to also under-approximate program paths. On top of that, using an algebraic representation of choice, the soundness proof of our frame rule extends to other execution models too, such as probabilistic programs.


\heading{Symbolic execution algorithms}. As a proof of concept for how OSL enables the consolidation of correctness and incorrectness analysis, we present symbolic execution algorithms to analyze C-like pointer programs. The core algorithm finds all the reachable outcomes, and is therefore tailored for correctness reasoning.
We also define a \emph{single-path} variant---modeled after Incorrectness Logic based bug-finding algorithms (Pulse \cite{isl} and Pulse-X \cite{realbugs})---inferring specifications in which the postcondition is just one of the (possibly many) outcomes. While enjoying the scalability benefit of dropping paths, our algorithm can also soundly re-use procedure summaries generated by the correctness algorithm so as to not re-analyze the same procedure.

These algorithms are based on bi-abduction, which handles \emph{sequential} programs by reconciling the postcondition of one precomputed spec with the precondition of the next \cite{biab,biabjacm}. Since programs with effects are not purely sequential, but rather have branching that arises from, \eg nondeterministic or probabilistic choice, we also introduce \emph{tri-abduction}, a new form of inference for composing branches in an effectful program.
%

\smallskip
\noindent We begin in \Cref{sec:overview} by outlining the challenges of local reasoning in a highly expressive program logic.
Next, in \Cref{sec:sem,sec:osl}, we define Outcome Separation Logic (OSL), show three instantiations, and prove the soundness of the frame rule.
In \Cref{sec:symexec}, we define our symbolic execution algorithm and tri-abduction, which is inspired by bi-abduction but is used for branching rather than sequential composition.
In \Cref{sec:examples}, we examine two case studies to show the applicability of these algorithms and finally we conclude in \Cref{sec:related,sec:conclusion} by discussing outlooks and related work.



\section{Local Reasoning for Correctness and Incorrectness with Effects}
\label{sec:overview}

We begin by examining how the local reasoning principles of separation logic, along with bi-abductive inference, underlie scalable analysis techniques.
These analyses symbolically execute programs and report the result as Hoare Triples $\hoare PCQ$:
the postcondition $Q$ describes any result of running $C$ in a state satisfying the precondition $P$~\cite{hoarelogic}. Hoare triples are compositional; a specification for the sequence of two commands is constructed from specifications for each one.
\[
\inferrule{
\hoare P{C_1}Q
\\
\hoare Q{C_2}R
}
{\hoare P{C_1\fatsemi C_2}R}
{\textsc{Sequence}}
\]
The \textsc{Sequence} rule is a good starting point for building scalable program analyses, but it is not quite compositional enough. The postcondition of $C_1$ must exactly match the precondition of $C_2$, making it difficult to apply this rule, particularly if $C_1$ and $C_2$ are procedure calls for which we already have pre-computed summaries (in the form of Hoare Triples), none of which exactly match.
In response, separation logic offers a second form of (spatial) compositionality via the \textsc{Frame} rule, which adds information about unused program resources $F$ to the pre- and postcondition of a completed proof.
\[
\inferrule{\hoare PCQ}{\hoare{P*F}C{Q*F}}{\textsc{Frame}}
\]
But framing does not immediately answer how to sequentially compose triples. Given $\hoare{P_1}{C_1}{Q_1}$ and $\hoare{P_2}{C_2}{Q_2}$, it is unclear what---\emph{if anything}---can be added to make $Q_1$ match $P_2$. This is where bi-abduction comes in---a technique that finds a missing resource $M$ and leftover frame $F$ to make the entailment $Q_1\sep M \vDash P_2\sep F$ hold. The question of deciding these entailments for separation logic assertions is well-studied \cite{berdine2005decidable}, and yields a more usable sequence rule that stitches together two precomputed summaries without reexamining either program fragment.
\[
\inferrule{
  \hoare {P_1}{C_1}{Q_1}
  \\
  Q_1 * M \vDash P_2 * F
  \\
  \hoare {P_2}{C_2}{Q_2}
}
{\hoare{P_1*M}{C_1\fatsemi C_2}{Q_2*F}}
{\textsc{Bi-Abductive Sequence}}
\]
While bi-abduction has helped industrial strength static analyzers scale to massive codebases, current developments use disjoint algorithms for correctness vs incorrectness, and do not support computational effects such as probabilistic choice.
In the remainder of this section, we will examine why this is the case and explain how our logic allows for unified bi-abductive analysis algorithms.

\subsection{Interlude: Reasoning about Effects and Incorrectness}
\label{sec:interlude}

While identifying bugs in pure programs is already challenging, effects add more complexity.
This is demonstrated below; the program crashes because it attempts to dereference a null pointer.
\[
\hoare{\ok: x = \mathsf{null}}{[x] \leftarrow 1}{\er:x = \mathsf{null}}
\]
The specification is semantically straightforward; if $x$ is null then the program will crash.\footnote{
Caveat: nontermination is an effect, and the typical \emph{partial correctness} interpretation of SL is not suitable for incorrectness.
}
Now, rather than dereferencing a pointer that is \emph{known} to be invalid, suppose we dereference a pointer that \emph{might} be invalid, and---crucially---whether or not it is allocated comes down to nondeterminism. The following is one such scenario; $x$ is obtained using \textsf{malloc}, which either returns a valid pointer or \textsf{null}.
In Hoare Logic, the best we can do is specify this program using a disjunction.
\[
\hoare{\ok: \emp}{x\coloneqq \mathsf{malloc}()\fatsemi [x] \leftarrow 1}{(\ok:x\mapsto 1) \vee (\er:x = \mathsf{null})}
\]
While the above specification hints that the program has a bug, it is in fact inconclusive since the disjunctive postcondition does not guarantee that both outcomes are \emph{reachable} by actual program executions. Hoare Logic is fundamentally unable to characterize this bug, since the postcondition must describe \emph{all} possible end states of the program; we cannot express something that \emph{may} occur.

Two solutions for characterizing true bugs have been proposed. The first one is Incorrectness Logic (IL), which uses an alternative semantics to express that all states described by the postcondition are reachable from a state described by the precondition \cite{il}. Specifying the aforementioned bug is possible using Incorrectness Logic; the following triple stipulates that all the states described by the postcondition are reachable, including ones where the error occurs.
\[
\inc{\ok: \emp}{x\coloneqq \mathsf{malloc}()\fatsemi [x] \leftarrow 1}{(\ok:x\mapsto 1) \vee (\er:x = \mathsf{null})}
\]
A variant of IL has a frame rule \cite{isl} and can underlie bi-abductive symbolic execution algorithms \cite{realbugs}. However---just like separation logic---IL is specialized to nondeterministic programs and cannot be used for programs with other effects such as randomization. In addition, being inherently under-approximate, the semantics of IL cannot capture correctness properties, which must cover all the reachable outcomes. As such, different analyses and procedure summaries must be used for verification vs bug-finding.

In this paper, we take a different approach based on Outcome Logic (OL), which is compatible with both correctness and incorrectness while also supporting a variety of monadic effects \cite{outcome}. OL is similar to Hoare Logic, but the pre- and postconditions of triples describe \emph{collections} of states rather than individual ones. A new logical connective $\oplus$---the outcome conjunction---guarantees the reachability of multiple outcomes. For instance, the aforementioned bug can be characterized using the following Outcome Logic specification by using an outcome conjunction instead of a disjunction in the postcondition.
\[
\triple{\ok: \emp}{x\coloneqq \mathsf{malloc}()\fatsemi [x] \leftarrow 1}{(\ok:x\mapsto 1) \oplus (\er:x = \mathsf{null})}
\]
The above triple says that running the program in the empty heap will result in two reachable outcomes.
In this case, the program is nondeterministic and its semantics is accordingly characterized by a \emph{set} of program states $S$.
 The outcome conjunction tells us that 
there exist nonempty sets $S_1$ and $S_2$ with $S = S_1\cup S_2$ such that
$S_1\vDash (\ok:x\mapsto 1)$ and $S_2\vDash (\er:x = \mathsf{null})$. Since both outcomes are satisfied by nonempty sets, we know that they are both reachable by a real program trace.

For efficiency, specifying the bug above should not require recording information about the $\ok$ outcome. In incorrectness reasoning, it is desirable to \emph{drop outcomes} so as to only explore some of the program paths \cite{il,realbugs}. We achieve this in OL by replacing the extraneous outcome with $\top$---an assertion that is satisfied by any set of states, including the empty set.
\[
\triple{\ok: \emp}{x\coloneqq \mathsf{malloc}()\fatsemi [x] \leftarrow 1}{(\er:x = \mathsf{null}) \oplus\top}
\]
Since $\triple\top{C}\top$ is valid for any program $C$, this trick allows us to propagate the $\top$ forward through the program derivation without determining what would happen if the memory dereference had succeeded. The following triple is valid too, regardless of what comes next in the program.
\[
\triple{\ok: \emp}{x\coloneqq \mathsf{malloc}()\fatsemi [x] \leftarrow 1 \fatsemi \cdots\fatsemi \cdots\fatsemi \cdots}{(\er:x = \mathsf{null}) \oplus\top}
\]
Outcomes apply to more effects than just nondeterminism. For example, Outcome Logic can also be used to reason about probabilistic programs, where the (weighted) outcome conjunction additionally quantifies the likelihoods of outcomes. For example, the following program attempts to ping an IP address, which succeeds 99\% of the time, and fails with probability 1\% due to an unreliable network.
\[
\triple{\ok:\tru}{x \coloneqq \mathsf{ping}(192.0.2.1)}{(\ok:x=0) \oplus_{99\%} (\er:x=1)}
\]
Our goal in this paper is to extend local reasoning to Outcome Logic by augmenting it with a frame rule, and to use the resulting theory to build bi-abductive symbolic execution algorithms for both correctness (finding all reachable outcomes) and incorrectness (only exploring one program path at a time) in pointer programs with varying effects. We will next see the challenges behind designing a frame rule powerful enough to achieve that goal.

\subsection{Designing a More Powerful Frame Rule}
\label{sec:frameoverview}

In their initial development of the frame rule, \citet[\S1]{yang2002semantic} remarked:
\begin{quote}
``\emph{The first problem we tackle is the soundness of the Frame Rule. This turns out to be surprisingly delicate, and it is not difficult to find situations where the rule doesn’t work. So a careful treatment of soundness, appealing to the semantics of a specific language, is essential.}''
\end{quote}
Today, there are many frame rules that appeal to the semantics in their respective situations, relying on properties such as nondeterminism, partial correctness, or under-approximation for soundness.
Our challenge is to design a frame rule supporting \emph{all} of those features, which means that fundamental questions must be addressed head on, without appealing to a \emph{particular} semantic model. We follow the basic formula of standard separation logic \cite{yang2002semantic,calcagno2007local}, as it is semantically closest to Outcome Logic.\footnote{
We remark that other proof strategies exist including using heap monotonicity (ISL) and frame baking (higher-order separation logics), a comparison to these approaches is available in \Cref{sec:related}.}

\heading{Local Actions.} Framing is sound if all program actions are \emph{local} \cite{calcagno2007local}. Roughly speaking, $C$ is local if its behavior does not change as pointers are added to the heap.
%
Most actions are inherently local, \eg dereferencing a pointer depends only on a single heap cell.

But memory allocation is---or at least, some implementations of it are---inherently non-local; allocating a new address involves reaching into an unknown region of the heap. To see why this is problematic, consider a deterministic allocator that always returns the smallest available address. So, allocating a pointer in the empty heap is guaranteed to return the address 1. The following triple is therefore valid, stating that $x$ is the address 1, which points to some value.
\[
\hoare{\emp}{x\coloneqq\mathsf{alloc}()}{x = 1 \land x\mapsto -}
\]
However, applying the frame rule can easily invalidate this specification.
\[
\inferrule{
\hoare{\emp}{x\coloneqq\mathsf{alloc}()}{x = 1 \land x\mapsto -}
}
{\hoare{y\mapsto 2}{x\coloneqq\mathsf{alloc}()}{x = 1 \land x\mapsto - \sep y\mapsto 2}}
{\textsc{Frame}}
\]
Since $y$ can have address 1, the fresh pointer $x$ cannot also be equal to 1, so this application of the frame rule is clearly unsound.
Making memory allocation nondeterministic can solve the problem \cite{yang2002semantic}. If, in the above program, $x$ could be assigned \emph{any} fresh address, then the postcondition cannot say anything specific about \emph{which} address got allocated. We cannot conclude that $x=1$, but rather only that $x=1 \vee x=2 \vee \cdots$, which remains true in any larger heap.

Moving from Hoare Logic to Outcome Logic, this approach has two problems. First, Outcome Logic is parametric on an evaluation model, so the availability of nondeterminism is not a given; we need a strategy for allocation in the presence of \emph{any} computational effects. Second, even with nondeterminism, locality is complicated by the ability to reason about reachable states. 
The problematic interaction between framing and reachability is displayed in the following example, which explicitly states that $x=1$ is a reachable outcome of allocating $x$ in the empty heap. 
\[
\inferrule{
\triple{\ok:\emp}{x\coloneqq\mathsf{alloc}()}{(\ok:x=1) \oplus (\ok:x\neq 1)}
}
{\triple{\ok: y\mapsto2}{x\coloneqq\mathsf{alloc}()}{(\ok:x=1 \land y\mapsto 2) \oplus (\ok:x\neq 1 \land y\mapsto 2)}}
{\textsc{Frame}}
\]
This inference is invalid; the precondition does not preclude $y$ having address 1, in which case $x=1$ is no longer a reachable outcome. 
%
In OSL, we acknowledge that memory allocation is non-local and instead ensure that triples cannot specify the result of an allocation too finely. This is achieved by altering the semantics of a triple $\triple{\varphi}C{\psi}$ to require that if $\varphi$ holds in some initial state, then $\psi$ holds after running $C$ using \emph{any} allocation semantics. By universally quantifying the allocator, we ensure that any concrete semantics is captured, while retaining soundness of the frame rule. In particular, the premise in the above inference that $x=1$ is a reachable outcome is invalid according to our semantics because there are many allocation semantics in which it is false.

\heading{Fault Avoidance.} Typical formulations of separation logic are \emph{fault avoiding}, meaning that the precondition must imply that the program execution does not encounter a memory fault \cite{sl,yang2002semantic}. Unlike in Hoare Logic, the triple $\hoare\tru{C}\tru$ is not necessarily valid, which is crucial to the frame rule.
If the triple $\hoare{\tru}{[x] \leftarrow 1}{\tru}$ were valid, then the frame rule would give us $\hoare{x \mapsto 2 \sep\tru}{[x]\leftarrow 1}{x \mapsto 2 \sep \tru}$, which is clearly false. 



The problem with fault avoidance is that it involves inspecting the entire program, whereas OSL must be able to ignore some paths to more efficiently reason about incorrectness.
Fortunately,
OSL preconditions need not be safe.
 Following the previous example, $\triple{\ok:\tru}{[x]\leftarrow 1}{\ok:\tru}$ is not a valid OSL specification since $\ok:\tru$ is not a \emph{reachable} outcome of running the program. Rather, if the precondition of an OSL triple is unsafe, then the postcondition can only be $\top$, an assertion representing any \emph{collection} of outcomes, including divergence. So, $\triple{\ok:\tru}{[x]\leftarrow 1}{\top}$ is a valid triple and framing information into it is perfectly sound since the $\top$ will absorb \emph{any} outcome, including undefined behavior:
 $\triple{\ok:x\mapsto 2\sep \tru}{[x]\leftarrow 1}{\top}$.

OSL allows us to decide how much of the memory footprint to specify. In a correctness analysis covering all paths, the precondition must be safe for the entire program. If we instead want to reason about incorrectness and drop paths, then it must only be safe for the paths we explore.

\subsection{Symbolic Execution and Tri-Abduction}

OSL provides a logical foundation for symbolic execution algorithms that are capable of \emph{both} verification \emph{and} bug-finding.
Our approach takes inspiration from industrial strength bi-abductive analyzers (Abductor and Infer \cite{biab,infer}), but paying greater attention to effects, as the aforementioned tools may fail to find specifications for programs with control flow branching.

To see what goes wrong, let us examine a program that uses disjoint resources in the two nondeterministic branches: $([x]\leftarrow 1) + ([y] \leftarrow 2)$.
Using bi-abduction, we could conclude that $x\mapsto-$ is a valid precondition for the first branch whereas $y\mapsto-$ is valid for the second, but there is no immediate way to find a precondition valid for \emph{both} branches since the branches may contain overlapping resources involving pointers and inductive predicates about more complex data structures such as lists. As a result, the program must be re-evaluated with each candidate precondition to ensure that they are safe for all branches.

\citet[\S4.3]{biabjacm} acknowledged this issue, and suggested fixing it by re-running the abduction procedure until nothing more can be added to each precondition. Rather than using two passes (as Abductor already does), this requires a pass for each combination of nondeterministic choices, which is exponential in the worst case.
While a single-pass bi-abduction algorithm has been proposed \cite{sextl2023sound}, it still cannot handle all cases (see \Cref{rem:bi-tri} in \Cref{sec:triabduction}).

We take a different approach, acknowledging that branching is fundamentally different from sequential composition and requires a new type of inference, which we call tri-abduction. As its name suggests, tri-abduction infers three components (to bi-abduction's two). Given $P_1$ and $P_2$---preconditions for two branches---the goal is to find a single anti-frame $M$ and two leftover frames $F_1$ and $F_2$ such that $M \vDash P_1 \sep F_1$ and $M \vDash P_2\sep F_2$, in order to compose the summaries for two program branches, as demonstrated below.
\[
\inferrule{
\triple{P_1}{C_1}{Q_1}
\\
P_1\sep F_1 \Dashv M \vDash P_2\sep F_2
\\
\triple{P_2}{C_2}{Q_2}
}
{\triple{M}{C_1 + C_2}{(Q_1\sep F_1) \oplus (Q_2 \sep F_2)}}
{\textsc{Tri-Abductive Composition}}
\]
Tri-abduction does not replace bi-abduction; they work in complementary ways---bi-abduction is used for sequential composition whereas tri-abduction composes branches arising from effects.

In addition, we are interested in bug-finding algorithms, which---similar to Pulse and Pulse-X \cite{isl,realbugs}---do not traverse all the program paths.
We achieve this using a single-path version of the algorithm, producing summaries of the form $\triple{\ok:P}{C}{(\er:Q) \oplus\top}$, with only a single outcome specified and the remaining ones covered by $\top$.
The soundness of the single-path approach is motivated by the fact that $P\oplus Q \Rightarrow P\oplus\top$; extraneous outcomes can be converted to $\top$, ensuring that those  paths will not be explored.
Just like in Pulse and Pulse-X, this ability to \emph{drop outcomes} allows the analysis to retain less information for increased scalability.

\smallskip
\noindent
We formalize these concepts starting in \Cref{sec:sem,sec:osl}, where we define a program semantics and Outcome Separation Logic and prove the soundness of the frame rule. Symbolic execution and tri-abduction are defined in \Cref{sec:symexec}, and we examine case studies in \Cref{sec:examples}.

\section{Program Semantics}
\label{sec:sem}

We begin the technical development by defining the semantics for the underlying programming language of Outcome Separation Logic. 
All instances of OSL share the same program syntax, but this syntax is interpreted in different ways, corresponding to the choice mechanisms dictated by each instance's computational effects. The syntax of the language is given below.
\begin{align*}
\mathsf{Cmd} \ni C \Coloneqq&~ \skp \mid C_1 \fatsemi C_2 \mid C_1 + C_2 \mid \assume e \mid \whl eC \mid c \\
\mathsf{Act} \ni c \Coloneqq &~ x\coloneqq e \mid x\coloneqq\mathsf{alloc}() \mid \mathsf{free}(e) \mid [e_1]\leftarrow e_2 \mid x \leftarrow [e] \mid \mathsf{error}() \mid f(\vec e) \\
\Exp \ni e \Coloneqq&~ x \mid \kappa \mid e_1 \diamond e_2 \mid \lnot e
\qquad\qquad\qquad\qquad (x \in\mathsf{Var}, \kappa\in\mathsf{Val}, f\in \mathsf{Proc}, {\diamond}\in\mathsf{BinOp})
\end{align*}
Commands $C\in \mathsf{Cmd}$ are similar to those of \cites{gcl} Guarded Command Language (GCL), containing $\skp$, sequencing $(C_1\fatsemi C_2)$, choice ($C_1 +C_2$) and while loops. Given a test $e$ (\ie a Boolean-valued expression), conditionals are defined as syntactic sugar in the typical way: $\iftf e{C_1}{C_2} \triangleq (\assume e\fatsemi C_1) + (\assume{\lnot e}\fatsemi C_2)$.
But---unlike GCL---$\mathsf{assume}$ also accepts expressions that are interpreted over \emph{weights} drawn from a set that has additional algebraic properties described in \Cref{sec:algebra} (and can also encode the Booleans). For example, weights in randomized programs are probabilities $p \in [0,1]$, and we can accordingly define a probabilistic choice operator $C_1+_p C_2 \triangleq (\assume p \fatsemi C_1) + (\assume{1-p}\fatsemi C_2)$, which runs $C_1$ with probability $p$ and $C_2$ with probability $1-p$.

Atomic actions $c\in\mathsf{Act}$ can assign to variables ($x\coloneqq e$), allocate ($x\coloneqq\mathsf{alloc}()$) and deallocate ($\mathsf{free}(e)$) memory, write ($[e_1]\leftarrow e_2$) and read ($x \leftarrow [e]$) pointers, crash ($\mathsf{error}()$), and call procedures ($f(\vec e)$). Expressions can be variables $x\in\mathsf{Var}$, constants $\kappa\in\mathsf{Val}$ (\eg integers, Booleans, and \emph{weights}), a variety of binary operations $e_1 \diamond e_2$ where ${\diamond}\in\mathsf{BinOp} = \{ +, -, =, \le, \ldots\}$, or logical negation $\lnot e$. 

In the remainder of this section, we will formally define denotational semantics for the language above. This will first involve discussing the algebraic properties of the program weights, after which we can define a (monadic) execution model to interpret sequential composition.

\subsection{Algebraic Preliminaries}
\label{sec:algebra}

We first recall the definitions of some algebraic structures that will be used to instantiate the program semantics for different execution models. \emph{Monoids} model combining and scaling outcomes.

\begin{definition}[Monoid]
A monoid $\langle A, +, \zero\rangle$ consists of a carrier set $A$, an associative binary operator $+\colon A\times A\to A$, and an identity element $\zero\in A$ such that $a+\zero = \zero+a = a$ for all $a\in A$.
Additionally, a monoid is \emph{partial} if $+$ is partial ($+\colon A\times A\rightharpoonup A$) and it is \emph{commutative} if $a+b = b+a$.
\end{definition}

For example, $\langle [0,1], +, 0\rangle$ is a partial commutative monoid that is commonly used in probabilistic computation since probabilities come from the interval $[0,1]$ and addition is undefined if the sum is greater than 1. Scalar multiplication $\langle [0,1], \cdot, 1\rangle$ is another monoid with with same carrier set, but it is total rather than partial. Two monoids can be combined to form a semiring, as follows.

\begin{definition}[Semiring]
A semiring $\langle A, +, \cdot, \zero, \one\rangle$ consists of a carrier set $A$, along with an addition operator $+$, a multiplication operator $\cdot$ and two elements $\zero,\one\in A$ such that:
\begin{enumerate}
\item $\langle A, +, \zero\rangle$ is a commutative monoid.
\item $\langle A, \cdot, \one\rangle$ is a monoid (we sometimes omit $\cdot$ and write $a\cdot b$ as $ab$).
\item Multiplication distributes over addition: $a \cdot (b+c) = ab + ac$ and $(b+c)\cdot a = ba + ca$
\item $\zero$ is the annihilator of multiplication: $a\cdot\zero = \zero\cdot a = \zero$
\end{enumerate}
A semiring is \emph{partial} if $\langle A, +, \zero\rangle$ is instead a partial commutative monoid (PCM), but multiplication remains total. In the case of a partial semiring, distributive rules only apply if $b+c$ is defined.
\end{definition}


We now define {\em Outcome Algebras} that give the interpretation of choice and weighting. The carrier set $A$ is used to represent the weight of an outcome. In deterministic and nondeterministic evaluation models, this weight can be 0 or 1 (a Boolean), but in the probabilistic model, it can be any probability in $[0,1]$. The rules for combining these weights vary by execution model.

\begin{definition}[Outcome Algebra]\label{def:ocalg}
An \emph{Outcome Algebra} is a structure $\langle A, +, \cdot, \zero, \one\rangle$, which is a complete, Scott continuous, naturally ordered, partial semiring, and:
\begin{enumerate}
\item \textbf{\textsf{Ordering:}}
$\langle A, \le\rangle$ is a directed complete partial order (dcpo) and $\sup(A) = \one$.
\item \textbf{\textsf{Normalization:}}
If $\sum_{i\in I} a_i$ is defined, then there exist $(b_i)_{i\in I}$ such that $\sum_{i\in I} b_i = \one$ and $\forall i\in I. a_i = (\sum_{j\in I} a_j)\cdot b_i$.
\end{enumerate}
\Aref{app:semantics} defines Scott continuity and completeness. Property (2) guarantees that any weighting function can be normalized to have a cumulative weight of $\one$ without affecting the relative weights of any of the elements. This will be necessary later on in order to prove that weights can be tabulated to witness the relationship between two weighted collections. More details are available in \Aref{app:semantics}.
\end{definition}

\noindent Outcome Algebras can encode the following three interpretations of choice:

\begin{definition}[Deterministic Outcome Algebra]\label{def:detalg}
A deterministic program has \emph{at most} one outcome (zero if it diverges). To encode this, we use an Outcome Algebra $\langle \{0,1\}, +, \cdot, 0, 1\rangle$ where the elements $\{0,1\}$ are Booleans indicating whether or not an outcome has occurred. The sum operation is usual integer addition, but is undefined for $1+1$, since two outcomes cannot simultaneously occur in a deterministic setting, and $\cdot$ is typical integer multiplication. 
\end{definition}

\begin{definition}[Nondeterministic Outcome Algebra]\label{def:ndeval}
The Boolean semiring $\langle \{0,1\}, \lor,\land, 0,1\rangle$ represents nondeterminism.
Similar to \Cref{def:detalg}, the elements indicate whether an outcome has occurred, but now addition is logical disjunction so that multiple outcomes can occur. 
%
%
\end{definition}

\begin{definition}[Probabilistic Outcome Algebra]\label{def:probalg}
Let $\langle [0,1], +, \cdot, 0, 1\rangle$ be an outcome algebra where $+$ is real-valued addition (and undefined if $a+b > 1$) and $\cdot$ is real-valued multiplication. The carrier set $[0,1]$ indicates that each outcome has a probability of occurring.
\end{definition}

Now, using these different kinds of weights, we will interpret the result of running programs as weighted collections of end states. More precisely, it will be a map $\Sigma \to A$ from states $\sigma\in\Sigma$ to weights $a\in A$.
In the style of \citet{moggi91}, the language semantics is monadic in order to sequence effects. We now show how to construct a monad given any Outcome Algebra.

\begin{definition}[Outcome Monad]\label{def:ocmonad}
Given an Outcome Algebra $\mathcal A=\langle A, +, \cdot, \zero, \one \rangle$ (\Cref{def:ocalg}), let $\mathcal W_{\mathcal A}(S) = \{ m : S\to  A \mid \sum_{s\in \supp(m)} m(s) ~\text{is defined} \}$ be the set of countably supported \emph{weighting functions} on $\mathcal A$.
We define a monad $\langle \mathcal W_{\mathcal A}, \unit, \bind\rangle$, where the monad operations $\unit\colon X \to \mathcal W_{\mathcal A}(X)$ and $\bind \colon \mathcal W_{\mathcal A}(X) \times (X \to \mathcal W_{\mathcal A}(Y)) \to \mathcal W_{\mathcal A}(Y)$ are defined as follows:
\[
\unit(s)(t) = \left\{ \begin{array}{ll}
\one & \text{if}~s=t \\
\zero & \text{if}~s\neq t
\end{array}\right.
\qquad\qquad
\bind(m, f)(t) = \smashoperator{\sum_{s\in\supp(m)}} m(s)\cdot f(s)(t)
\]
We also let $\supp(m) = \{ s \mid m(s) \neq \zero \}$ and $|m| = \sum_{s\in\supp(m)}m(s)$.
This monad is similar to the \citet{giry} monad, which sequences computations on probability distributions. In our case, it is generalized to work over any partial semiring, rather than probabilities $[0,1]\subseteq\mathbb{R}$. It is fairly easy to see that $\mathcal W_{\mathcal A}$ obeys the monad laws, given the semiring laws. 
\end{definition}

So, a weighting function $m\in\mathcal{W}_{\mathcal A}(S)$ assigns a weight $a\in A$ to each program state $s\in S$. \Cref{def:detalg,def:ndeval,def:probalg} gave interpretations for $\mathcal A$ in which $\mathcal{W}_{\mathcal A}(S)$ encodes deterministic, nondeterministic, and probabilistic computation, respectively. In the (non)deterministic cases, $m(s)\in\{0,1\}$, indicating whether or not $s$ is present in the collection of outcomes $m$. Due to the interpretation of $+$ in \Cref{def:detalg}, the constraint that $\sum_{s\in \supp(m)} m(s)$ is defined guarantees that $m$ can contain at most one outcome, whereas in the nondeterministic case, $m$ can contain arbitrarily many. In the probabilistic case, $m(s)\in[0,1]$ and gives the probability of the outcome $s$ in the distribution $m$.

The semiring operations can be lifted to weighting functions. We will overload some notation to also refer to pointwise liftings as follows: $m_1+m_2 = \lambda s.(m_1(s)+m_2(s))$, $\zero = \lambda s.\zero$, and $a\cdot m = \lambda s.(a\cdot m(s))$.
When the nondeterministic algebra (\Cref{def:ndeval}) is lifted in this way, the result is isomorphic to the powerset monad with $m_1+m_2 \cong m_1 \cup m_2$ and $\zero\cong\emptyset$.

Now, in order to represent errors and undefined states in the language semantics, we will define a monad transformer \cite{transformers} based on the coproduct $S + E + 1$ where $S$ is the set of program states, $E$ is the set of errors, and we additionally include an undefined symbol. We define the following three injection functions, plus shorthand for the undefined element (where $1 = \{\star\}$):
\[
\inj_\ok \colon S \to S+E+1
\qquad
\inj_\er \colon E \to S+E+1
\qquad
\inj_\und \colon 1\to S+E+1
\qquad
\und = \inj_\und(\star)
\]
Borrowing the notation of Incorrectness Logic \cite{il}, we use $\ok$ and $\er$ to denote states in which the program terminated successfully or crashed, respectively.
We will also write $\inj_\epsilon$ to refer to one of the above injections, where $\epsilon\in\{\ok,\er\}$.

\begin{definition}[Error Monad Transformer]\label{def:ermonad}
Let $\langle \mathcal W_{\mathcal A}, \bind_{\mathcal W}, \unit_{\mathcal W}\rangle$ be the monad described in \Cref{def:ocmonad} and $E$ be a set of error states. We define a new monad $\langle \mathcal W_{\mathcal A}(-+E+1), \bind, \unit\rangle$ where the monad operations are defined as follows:
\[
\unit(s) = \unit_{\mathcal W}(\inj_\ok(s))
\qquad
\bind(m, f) = \bind_{\mathcal W}\left(m, \lambda \sigma.\left\{\begin{array}{ll}
f(s) & \text{if}~\sigma = \inj_\ok(s) \\
\unit_{\mathcal W}(\sigma) & \text{otherwise}
\end{array}\right. \right)
\]
\end{definition}

\subsection{Denotational Semantics}

\begin{figure}
%
\begin{flushleft}\fbox{
\textsf{Commands\quad $\de{-}_{\af} \colon \mathsf{Cmd} \to \mathcal S\times\mathcal H\to\mathcal W_{\mathcal A}(\st)$}
}\end{flushleft}
\begin{align*}
\de{\skp}_{\af}(s,h) &= \unit(s,h)
\\
\de{C_1\fatsemi C_2}_\af(s,h) &= \bind(\de{C_1}_\af(s,h), \de{C_2}_\af)
\\
\de{C_1 + C_2}_\af(s,h) &= \de{C_1}_\af(s,h) + \de{C_2}_\af(s,h)
\\
\de{\assume e}_\af(s,h) &= \de{e}(s)\cdot\unit(s, h) \quad\text{if}~\de{e}(s) \in A
\\
\de{\whl eC}_\af(s,h) &= \mathsf{lfp}(F_{\langle C, e, \af\rangle})(s,h)
\\
  \text{where}~ F_{\langle C, e, \af\rangle}(f)(s,h) &= \left\{\begin{array}{ll}
\bind(\de{C}_\af(s,h), f) & \text{if}~ \de{e}(s) = \one \\
\unit(s,h) & \text{if}~ \de{e}(s) = \zero
\end{array}\right.
\end{align*}
\begin{flushleft}
\fbox{
\textsf{Actions \quad $\de{-}_\af \colon \mathsf{Act} \to \mathcal S\times\mathcal H\to\mathcal W_{\mathcal A}(\st)$}
}
\end{flushleft}
\begin{align*}
\de{x\coloneqq e}_\af(s,h) &= \unit(s[x \mapsto \de{e}(s)], h) \\
\de{x \coloneqq \mathsf{alloc}()}_\af(s,h) &= \bind_{\mathcal W}(\af(s,h), \lambda (\ell, v). \unit(s[x\mapsto \ell], h[\ell\mapsto v])) \\
\de{\mathsf{free}(e)}_\af(s,h) &= \mathsf{update}(s,h, \de{e}(s), s, h[\de{e}(s) \mapsto\bot]) \\
\de{[e_1] \leftarrow e_2}_\af(s,h) &= \mathsf{update}(s,h,\de{e_1}(s), s, h[\de{e_1}(s) \mapsto \de{e_2}(s)]) \\
\de{x \leftarrow [e]}_\af(s,h) &= \mathsf{update}(s,h,\de{e}(s), s[x \mapsto h(\de{e}(s))], h) \\
\de{\mathsf{error}()}_\af(s,h) &= \mathsf{error}(s,h)\\
\de{f(\vec e)}_\af(s, h) &= \de{C}_\af(s[ \vec x \mapsto \de{\vec e}(s) ], h) \quad \text{where}~ (C,\vec x) = \pt(f)
\end{align*}
%
\begin{flushleft}
\fbox{
\textsf{Basic Operations}
}
\end{flushleft}
\[
\mathsf{error}(s,h) = \unit_{\mathcal W}(\inj_\er(s,h))
\quad
\mathsf{update}(s, h, \ell, s',h') = \left\{
\begin{array}{ll}
\unit(s', h') & \text{if}~ h(\ell) \in\mathsf{Val} \\
\mathsf{error}(s,h) & \text{if} ~ h(\ell) = \bot \vee \ell=\mathsf{null} \\
\unit_{\mathcal W}(\und) & \text{if} ~  \ell\notin \mathsf{dom}(h)
\end{array}\right.
\]
\caption{Denotational semantics of program commands, parametric on an outcome algebra $\mathcal A = \langle A, +, \cdot, \zero, \one \rangle$, an allocator function $\mathsf{alloc} : \mathcal S\times \mathcal H \to \mathcal W_{\mathcal A}(\mathsf{Addr})$, and a procedure table $\pt \colon \mathsf{Proc} \to \mathsf{Cmd}\times\vec{\mathsf{Var}}$.}
\label{fig:cmdsem}
\end{figure}

The semantics of commands $\de{-}_\af \colon \mathsf{Cmd}\to \mathcal S\times\mathcal H\to\mathcal W_{\mathcal A}(\st)$ is defined in \Cref{fig:cmdsem} and is parametric on an outcome algebra $\mathcal A =\langle A, +, \cdot, \zero, \one \rangle$ and an allocator function $\af \in \mathsf{Alloc}_{\mathcal A}$, described below.
The semantics of expressions $\de{-}\colon \mathsf{Exp}\to\mathcal S\to\mathsf{Val}$ is omitted, but is defined in the obvious way. The set of states is $\st = (\mathcal S\times \mathcal H) + (\mathcal S\times\mathcal H) + 1$, where $\mathcal S = \{ s : \mathsf{Var} \cup \mathsf{LVar} \to \mathsf{Val} \}$ are stores (assigning values to both program variables $x\in\mathsf{Var}$ and logical variables $X\in\mathsf{LVar}$) and $\mathcal H = \{ h : \mathsf{Addr} \rightharpoonup \mathsf{Val}\cup\{\bot\} \}$ are heaps. In the style of \citet{isl}, a heap is both a \emph{partial} mapping and also includes $\bot$ in the codomain, distinguishing between cases with no information about an address ($\ell \notin \mathsf{dom}(h)$) vs cases where a pointer is explicitly deallocated ($h(\ell) = \bot$). We additionally assume that $\{\mathsf{null}\}\cup A \subseteq\mathsf{Val}$, $\mathsf{Addr} \subseteq\mathsf{Val}$ and $\mathsf{null}\notin\mathsf{Addr}$. 

Allocators are functions mapping $(s,h)$ to a collection of addresses and values with cumulative weight of $\one$, and each address $\ell$ is not in the domain of $h$. Allocators can use the same effects supported by the language semantics (\ie nondeterminism or random choice). Formally, allocators come from the following set.
\[
\mathsf{Alloc}_{\mathcal A} = \left\{ f \colon \mathcal S \times \mathcal H \to \mathcal W_{\mathcal A}(\mathsf{Addr} \times\mathsf{Val})
\quad\middle|\quad
 \forall (s, h).
\begin{array}{l}
 |f(s,h)| = \one \quad\text{and} \\
\forall (\ell,v) \in \supp(f(s,h)).~ \ell\notin \mathsf{dom}(h)
 \end{array}
 \right\}
\]
A deterministic allocator $\mathsf{alloc\_det}(s,h) = \unit(\min(\mathsf{Addr}\setminus \dom(h)), 0)$ that always picks the first unused address is valid in all OSL instances. Note that while logical variables $X\in\mathsf{LVar}$ cannot appear in program expressions $e$, allocators \emph{can} depend on them. This allows us to model allocators that depend on hidden state, such as kernel configuration that is not visible in user space.

A global procedure table $\pt \colon \mathsf{Proc} \to \mathsf{Cmd}\times\vec{\mathsf{Var}}$ that returns a command and vector of variable names (the arguments) given a procedure name $f\in\mathsf{Proc}$ is used to interpret procedure calls. We assume that all procedures used in programs are defined and pass the correct number of arguments.

The monad operations give semantics to $\skp$ and $\fatsemi$, and
$C_1 + C_2$ is defined as a sum whose meaning depends on the Outcome Algebra. While loops are defined by a least fixed point, which is guaranteed to exist \Acite{thm:fp}.
Since the semiring plus is partial, this sum may be undefined; in \Aref{app:semantics} we discuss simple syntactic checks to ensure totality. The semantics of $\assume e$ weights the branch by the value of $e$, as long as it is a valid program weight. Boolean expressions evaluate to the semiring $\zero$ (false) or $\one$ (true).
%

We define two operations before giving the semantics of atomic actions: $\mathsf{error}(s,h)$ constructs an error state and $\mathsf{update}(s,h,\ell,s',h')$ returns $(s',h')$ if the address $\ell$ is allocated in $h$, it returns an error if $\ell$ is deallocated, and returns $\und$ if $\ell\notin\mathsf{dom}(h)$. Assignment is defined in the usual way by updating the program store; memory allocation uses $\mathsf{alloc}$ to obtain a fresh address and initial value (or collection thereof); deallocation, reads, and writes are implemented using \textsf{update} and errors use \textsf{error}. Procedure names are looked up in $\pt$ to obtain $C$ and $\vec x$ before running $C$ on a store updated by setting $\vec x$ to have the values of the inputs $\vec e$.
We will additionally occasionally use the Kleisli extension of the semantics $\dem Cm{\af} \colon \mathcal W_{\mathcal A}(\st) \to \mathcal W_{\mathcal A}(\st)$, which takes as input a collection of states instead of a single one. It is defined $\dem Cm{\af} = \bind(m, \de{C}_\af)$.

\section{Outcome Separation Logic}
\label{sec:osl}

We now proceed to formalize Outcome Separation Logic (OSL) and the frame rule. First, we define an assertion logic for the pre- and postconditions of outcome triples. These assertions are based on the outcome assertions of \citet{outcome}, using SL assertions as basic predicates.


\heading{Outcome Assertions.}
The syntax for OSL assertions is below and the semantics is in \Cref{fig:oasem}. Both are parametric on an Outcome Algebra $\langle A, +, \cdot, \zero,\one\rangle$.
\[
\varphi \Coloneqq \top \mid \varphi\vee\psi \mid \varphi \oplus \psi \mid \wg\varphi{a} \mid \epsilon:P \qquad\quad \left(a\in A, \epsilon \in \{ \ok,  \er\}, P\in \bb{2}^{\mathcal S\times\mathcal H}\right)
\]
Outcome assertions include familiar constructs such as $\top$ (always true) and disjunctions.
The outcome conjunction $\varphi \oplus \psi$ splits $m$ into two pieces, $m_1\vDash\varphi$ and $m_2\vDash\psi$, summing to $m$. Weighting $\wg\varphi{a}$ guarantees that if $m\vDash \varphi$, then $a\cdot m\vDash \wg\varphi{a}$. Combining these, we also define probabilistic choice: $\varphi \oplus_p \psi \triangleq \wg\varphi{p} \oplus \wg\psi{1-p}$, meaning
that $\varphi$ occurs with probability $p$ and $\psi$ occurs with probability $1-p$ for some probability $p\in [0,1]$.

\begin{figure}
\small
\[
\begin{array}{lll}
m\vDash \top &\multicolumn{2}{l}{\text{always}} \\
m\vDash \varphi\lor\psi & \text{iff}& m\vDash \varphi \quad\text{or}\quad m\vDash\psi \\
m\vDash \varphi \oplus \psi & \text{iff}& \exists m_1,m_2.\quad
m = m_1 +  m_2 \quad\text{and}\quad m_1 \vDash\varphi \quad\text{and}\quad m_2\vDash\psi \\
m\vDash \wg\varphi{a} & \text{iff} & a=\zero ~\text{and}~ m=\zero \quad\text{or}\quad \exists m'.\quad m = a \cdot m' \quad\text{and}\quad m'\vDash\varphi \\
m\vDash \epsilon:P & \text{iff} & |m| = \one \quad\text{and}\quad \forall\sigma\in \supp(m).\quad \exists(s,h).\quad \sigma = \inj_\epsilon(s,h) \quad\text{and}\quad (s,h)\in P
\end{array}
\]
\caption{Semantics of outcome assertions given an outcome algebra $\langle A, +, \cdot, \zero,\one\rangle$.}
\label{fig:oasem}
\end{figure}

Finally, basic assertions $(\ok:P)$ and $(\er:Q)$ require that $|m| = \one$---ensuring that the set of outcomes is nonempty in the (non)deterministic cases (\Cref{def:ndeval,def:detalg}) or that the assertion occurs with probability 1 (\Cref{def:probalg})---and that all states in $\supp(m)$ terminated successfully and satisfy $P$ or crashed and satisfy $Q$, respectively, where $P, Q\in \bb{2}^{\mathcal S \times \mathcal H}$ are semantic heap assertions. These heap assertions are defined as in standard separation logic, for example:
\begin{align*}
\emp &\triangleq \{ (s, \emptyset) \mid s\in\mathcal S) \} \\
P \sep Q &\triangleq \{ (s, h_1 \uplus h_2) \mid (s, h_1) \in P, (s, h_2) \in Q \} \\
e_1 \mapsto e_2 &\triangleq \{ (s, \{ \de{e_1}(s) \mapsto \de{e_2}(s) \}) \mid s\in \mathcal S \}
\end{align*}
We were motivated to pick this particular set of outcome assertions in light of our goal to define symbolic execution algorithms in the style of \citet{biab}, which compute procedure summaries of the form $\hoare{P}{f(\vec e)}{Q_1 \vee \cdots \vee Q_n}$ and the disjunctive post indicates a series of possible outcomes. In our case, we exchange those disjunctions for outcome conjunctions in the cases where the outcomes arise due to computational effects. We include $\top$ in order to drop outcomes using the assertion $\varphi\oplus\top$ (as discussed in \Cref{sec:overview}). Finally, we include disjunctions to express joins of outcomes that occur due to logical choice, and also to express partial correctness; whereas $\epsilon:P$ guarantees reachability, $(\epsilon:P)\vee \wg\top\zero$ also permits nontermination (no outcomes).


\heading{OSL Triples.}
The semantics of OSL triples requires that the specification holds when the program is run using \emph{any} allocator, ensuring that the postcondition cannot say anything specific about which addresses were obtained, as those addresses may change if the program is executed in a larger heap.

\begin{definition}[Outcome Separation Logic Triples]
For any outcome algebra $\mathcal A = \langle A, +, \cdot, \zero,\one\rangle$:
\[
\vDash\triple{\varphi}C{\psi}
\qquad\text{iff}\qquad
\forall m\in\mathcal{W}_{\mathcal A}(\st), \af\in\mathsf{Alloc}_{\mathcal A}.\quad
m\vDash\varphi \quad \implies \quad \dem Cm\af\vDash\psi
\]
\end{definition}
\noindent The inference rules for OSL carry over from standard Outcome Logic \cite[Fig. 3]{zilberstein2024relatively}, along with the small axioms of separation logic \cite{localreasoning}. It is also fairly straightforward to derive inference rules from our symbolic execution algorithm (\Cref{fig:sym_exec}).
Instead of repeating the rules here, we discuss the new local reasoning principle that OSL supports, namely the frame rule.

%
%

\subsection{The Frame Rule}
We now build the necessary foundations to introduce and prove the soundness of the frame rule. 
First, we define a new separating conjunction as a transformation on outcome assertions, using the symbol $\osep$ to distinguish it from the usual separating conjunction $\sep$ on symbolic heaps; $\osep$ is a binary operation taking an outcome assertion and an SL assertion (rather than two SL assertions like $\sep$).
The operation is defined below, where ${\bowtie} \Coloneqq \oplus \mid \vee$.
\[
\top\osep F \triangleq \top
\quad
(\varphi \bowtie \psi)\osep F \triangleq (\varphi\osep F) \bowtie (\psi\osep F)
\quad
\wg\varphi{a} \osep F \triangleq \wg{(\varphi\osep F)}a
\quad
(\epsilon:P)\osep F \triangleq \epsilon: P\sep F
\]
So, $\osep$ has no effect on $\top$, it distributes over $\vee$, $\oplus$, and $\wg{(-)}a$, and for basic assertions $\epsilon:P$, we simply join $P\sep F$ with the regular separating conjunction. We can now express the frame rule. 
\[
\inferrule{
  \triple\varphi{C}\psi
  \\
  \mathsf{mod}(C) \cap \mathsf{fv}(F) = \emptyset
}{
  \triple{\varphi\osep F}C{\psi\osep F}
}{\textsc{Frame}}
\]
This rule resembles the frame rule of other separation logics, with the same side condition that $F$ cannot mention any modified program variables. However, as we described in \Cref{sec:frameoverview}, OSL's expressivity goes beyond that of existing separation logics, meaning that this frame rule can be used for both may and must properties in nondeterministic programs, as well as quantitative properties in probabilistic programs. Further, all of these capabilities stem from a single soundness proof.


The key to the proof is the \emph{frame property} (\Aref{lem:frameprop}), which roughly states that for any $(s, h)\in\mathcal S\times\mathcal H$, allocator $\af$, and $h'$ such that $(s, h') \vDash F$, we can construct a new allocator $\af'$ such that adding $h'$ to each end state of $\de{C}_{\af'}(s,h)$ gives us $\de{C}_\af(s, h\uplus h')$ (modulo the $\und$ states), guaranteeing that if $\de{C}_{\af'}(s,h)\vDash \psi$, then $\de{C}_\af(s, h\uplus h') \vDash \psi\osep F$. We formalize this idea in the remainder of the section.

The formalization of the frame property relies on \emph{lifted relations}, which describe the relation between two weighted collections in terms of a relation on individual states. Before we define this formally, we define the notion of projections $\pi_1 \colon \mathcal W(X\times Y) \to \mathcal W(X)$ and $\pi_2 \colon \mathcal W(X\times Y) \to \mathcal W(Y)$, which extract component weighting functions from weighting functions over a cartesian product.
\[
\pi_1(m)(x) = \smashoperator{\sum_{y \in Y}} m(x,y)
\qquad\qquad
\pi_2(m)(y) = \smashoperator{\sum_{x \in X}} m(x,y)
\]
Given these projections, we say that $m_1 \in \mathcal W(X)$ and $m_2\in\mathcal W(Y)$ are related by $R$ iff there is some weighting function $m \in \mathcal W(X\times Y)$ whose projections are $m_1$ and $m_2$, and where each pair $(x,y) \in \supp(m)$ is related by $R$. This notion corresponds to the idea of lifting via spans \cite{kurz2016relation} and is defined formally below.

\begin{definition}[Relation Liftings]
Given a relation $R \subseteq X\times Y$, we define a lifting of $R$ on weighting functions $\overline R \subseteq \mathcal W_{\mathcal A}(X)\times \mathcal W_{\mathcal A}(Y)$ as follows:
\[
\overline R = \Big\{~ (m_1,m_2) ~\Big|~ \exists m\in\mathcal{W}_{\mathcal A}(R).\quad
m_1 = \pi_1(m)
\quad\text{and}\quad
m_2 = \pi_2(m)
~\Big\}
\]
\end{definition}

\heading{Semantics of the Outcome Separating Conjunction.}
We now prove semantic properties about $\osep$ that allow us to relate program configurations satisfying $\varphi$ to ones that satisfy $\varphi\osep F$.
First, we need to relate states satisfying $P$ to states satisfying $P\sep F$.
\[
\mathsf{frame}(F) =
\big\{ \left(\inj_\epsilon(s, h), \inj_\epsilon(s, h\uplus h')\right) \mid \inj_\epsilon(s,h) \in \st, (s,h') \in F\big\}
\cup \big\{ (\und,\und)\big\}
\]
Any state $\inj_\epsilon(s,h)$ is related to all states $\inj_\epsilon(s,h\uplus h')$ such that $(s,h')\in F$, which guarantees that if $(s,h)\in P$, then $(s, h\uplus h')\in P\sep F$. Undefined states are only related to themselves. Now, we can express the semantics of $\osep$ by lifting this relation.
\begin{restatable}{lemma}{osepforward}
\label{lem:osepforward}
If $m\vDash\varphi$ and $(m, m') \in \overline{\fr(F)}$, then $m'\vDash \varphi \osep F$
\end{restatable}


It is tempting to say that the converse should also hold, but that is not quite right. We took $\top\osep F$ to be equal to $\top$, therefore if $m\vDash\top \osep F$, then we cannot guarantee that all the states in $m$ contain information about $F$.
We therefore characterize the semantics only for the states that are not covered by $\top$, leaving the other states unconstrained.
\begin{restatable}{lemma}{osepback}
\label{lem:osepback}

If $m\vDash\varphi\osep F$, then there exist $m_1$, $m'_1$, and $m_2$ such that $(m_1, m'_1)\in\overline{\fr(F)}$, $m = m_1' + m_2$ and $m_1 + m'_2 \vDash \varphi$ for any $m_2'$ such that $|m_2'| \le |m_2|$.
\end{restatable}
In the lemma above, $m_1'$ represents the nontrivial portion of $m$ and $m_2$ is the portion of $m$ that is covered by $\top$. As such, $m'_1$ must be the result of framing $F$ into some $m_1$. Since $m_2$ is covered by $\top$, we can replace it with anything smaller than $m_2$---$\top$ can absorb at least $|m_2|$ worth of weight. These two lemmas provide a semantic basis to reason about what it means for $\varphi\osep F$ to hold relative to $\varphi$.

\begin{remark}[Asymmetry of the Separating Conjunction]
One could imagine a symmetric definition of $\osep$, defined semantically as follows: $m \vDash\varphi\osep\psi$ iff there exist $m_1\vDash\varphi$ and $m_2\vDash\psi$, such that $m_1$ and $m_2$ are obtained by \emph{marginalizing} $m$. More precisely, $m_1(\inj_\epsilon(s, h))$ is equal to the sum of $m(\inj_\epsilon(s, h\uplus h'))$ over all $h'$ such that $\inj_\epsilon(s, h') \in \supp(m_2)$ (and a similar formula holds for $m_2$).

This gives expected properties, for example $(\ok:x\mapsto 1 \sep y\mapsto 2) \oplus (\ok:x\mapsto 3 \sep y\mapsto 4)$ implies
$\left( (\ok:x\mapsto 1) \oplus (\ok:x\mapsto 3) \right) \osep \left( (\ok:y\mapsto 2) \oplus (\ok:y\mapsto 4) \right)$. However, this approach does not work smoothly with the frame rule. As a simple counterexample, consider the triple $\triple{\ok:x\not\mapsto}{[x]\leftarrow 1}{\er:x\not\mapsto}$, using the frame rule with $(\ok:y\mapsto 2)$ gives us the precondition $(\ok:{x\not\mapsto}\sep y\mapsto 2)$, but the postcondition is false since there is no way to combine an $\er$ assertion with an $\ok$ one. Other problems occur when some of the program paths diverge.

The asymmetric definition of $\osep$ fits our needs, as the spirit of the frame rule is to run a computation in a larger \emph{heap}---adding pointers, but not outcomes. This is in line with how framing is used by \citet{biab}, where framing distributes over disjunctions in the postconditions of Hoare Triples.
That is not to say that the symmetric $\osep$ is useless. In fact, it could be used in a concurrent variant of Outcome Logic to divide resources among two parallel branches in the style of Concurrent Separation Logic \cite{csl}, but that is out of scope for this paper.
\end{remark}

\heading{Replacement of Unsafe States.}
Undefined states arise from dereferencing pointers not in the domain of the heap. Those pointers may be in the domain of a larger heap, so previously undefined outcomes can become defined after framing in more pointers.
The soundness of the frame rule depends on this not affecting the truth of the postcondition.

Whereas standard separation logic is \emph{fault avoiding}---it  requires that all states satisfying the precondition will not encounter unknown pointers---we omit this requirement in OSL in order to efficiently reason about incorrectness by only exploring a subset of the program paths. For example, in the triple
$
\triple{\ok: x\not\mapsto}{\mathsf{free}(x) + C}{(\er:x\not\mapsto)\oplus\top}
$,
the fact that the left path leads to a memory error is enough to conclude that the program is incorrect; exploring the right path would be wasted effort. However, the right path may use other pointers not mentioned in the precondition, meaning that executing $C$ will lead to $\und$. This is fine, since the assertion $\top$ covers undefined states.

However, if we use the frame rule to add information about more pointers to the precondition, the result may no longer be $\und$. This is still valid---\emph{any} outcome from running $C$ trivially satisfies $\top$.
To formalize this, we use $\repl\subseteq \st\times(\st \cup \{\lightning\})$, which relates $\und$ to any state $\sigma\in\st$ or $\lightning$ (representing nontermination)\footnote{
Previously undefined states may diverge after adding more pointers, 
\eg running $\mathsf{free}(x)\fatsemi \whl{\tru}\skp$ in an empty heap leads to $\und$ whereas it will not terminate if $x$ is allocated.}.
All other states ($\ok$ and $\er$) are related only to themselves. In addition, since $\lightning$ is not a state, we define an operation $\prune(m)$ to remove it. The formal definitions of both $\repl$ and $\prune$ are given in \Aref{app:replacement}.
The replacement lemma guarantees that undefined states can be replaced without affecting the validity of an assertion.

\begin{restatable}[Replacement]{lemma}{replcor}
\label{cor:replacement}
If $m\vDash\varphi$ and $(m, m')\in \overline\repl$, then $\prune(m') \vDash\varphi$
\end{restatable}

\heading{Soundness.}
We now have all the ingredients needed to prove the soundness of the frame rule.

%
%
%
%

\begin{restatable}[The Frame Rule]{theorem}{framethm}\label{thm:frame}
If $\vDash\triple{\varphi}C{\psi}$ and $\mathsf{mod}(C)\cap \mathsf{fv}(F) = \emptyset$, then $\vDash\triple{\varphi\osep F}C{\psi\osep F}$.
\end{restatable}

\noindent We briefly sketch the proof here, the full version is in \Aref{app:framerule}. Suppose that $m\vDash \varphi\osep F$ and take any $\af\in\mathsf{Alloc}$.
We now enumerate all the defined states of $m$ and (using \Cref{lem:osepback}) we know that each state has the form $\sigma_n = \inj_{\epsilon_n}(s_n, h_n\uplus h'_n)$ where $(s_n, h'_n)\in F$. We pick a fresh logical variable $X$ and construct $m'$ by modifying each $\sigma_n \in \supp(m)$ to be $\inj_{\epsilon_n}(s_n[ X \mapsto n], h_n)$. That is, we augment the variable store such that $X = n$ and remove the $h'_n$ portion of the heap. Note that $m'\vDash\varphi$.
We construct a new allocator as follows:
\[
\af'(s, h) = \af(s[ X \mapsto s_{n}(X)], h \uplus h'_{n})
\qquad\text{where}~
n = s(X)
\]
So, $\af'(s,h)$ uses the value $s(X) = n$ to select the appropriate $h'_n$ to add to $h$, guaranteeing that $\de{C}_{\af'}\!(s_n, h_n)$ allocates the same addresses as $\de{C}_\af\!(s_n, h_n\uplus h'_n)$ for all $n$. So, $\dem Cm\af$ is related to $\dem C{m'}{\af'}$ via the two relations described previously.
By the premise $\vDash\triple\varphi{C}\psi$, we know that $\dem C{m'}{\af'}\vDash\psi$, and by
\Cref{lem:osepforward,cor:replacement}, we conclude that $\dem Cm\af\vDash\psi\osep F$.

\medskip
\noindent We have devised a frame rule for Outcome Separation Logic, which supports a rich variety of computational effects and properties about those effects. More concretely, OSL can be used to reason about both correctness and incorrectness across nondeterministic and probabilistic programs.
In the next sections, we will build on this result to create compositional symbolic execution algorithms. 

\section{Symbolic Execution}
\label{sec:symexec}

With a sound frame rule, we are now ready to design symbolic execution algorithms based on OSL. The core algorithm is similar to Abductor and Infer \cite{biab,infer}, but adapted to better handle program choices arising from computational effects. We also show how minor modifications to this algorithm make it suitable for other use cases such as bug-finding and partial correctness. First, we introduce a restricted assertion syntax to make OSL compatible with bi-abduction, then we define tri-abduction (\Cref{sec:triabduction}) and the main algorithm (\Cref{sec:alg}).

\heading{Symbolic Heaps.}
In the remainder of the paper, we will work with a subset of SL assertions known as symbolic heaps. Although they have limited expressivity---particularly for pure assertions---implications of symbolic heaps are decidable \cite{berdine2005decidable}, which is necessary for bi-abductive analysis algorithms. These same symbolic heaps are used by \citet{biab,biabjacm}. The syntax is shown below and the semantics is standard as defined by \citet[\S2]{berdine2005symbolic}.
\begin{align*}
P &\Coloneqq \exists \vec X.\Delta && \textsf{(Symbolic Heaps)}
\\
\Delta &\Coloneqq \Pi \land \Sigma && \textsf{(Quantifier-Free Symbolic Heaps)}
\\
\Pi &\Coloneqq \tru \mid \Pi_1 \land \Pi_2 \mid e_1 = e_2 \mid e_1\neq e_2 && \textsf{(Pure Assertions)}
\\
\Sigma & \Coloneqq \tru \mid \emp \mid \Sigma_1 \sep \Sigma_2 \mid e_1 \mapsto e_2 \mid \ls(e_1, e_2) && \textsf{(Spatial Assertions)} 
\end{align*}
A symbolic heap $P$ consists of existentially quantified logical variables, along with a pure part $\Pi$ and a spatial part $\Sigma$. Pure assertions are conjunctions of equalities and inequalities, whereas spatial assertions are heap assertions joined by \emph{separating} conjunctions. The separating conjunction requires that the heap can be split into disjoint components to satisfy the two assertions separately.
\[
(s,h)\vDash \Sigma_1 \sep \Sigma_2 \qquad\text{iff}\qquad \exists h_1,h_2.\quad h = h_1\uplus h_2 \quad\text{and}\quad (s,h_1)\vDash\Sigma_1 \quad\text{and}\quad (s,h_2)\vDash\Sigma_2
\]
The points-to predicate $e_1 \mapsto e_2$ specifies a singleton heap in which the address $e_1$ points to the value $e_2$. Negative heap assertions $e\not\mapsto$ are syntactic sugar for $e\mapsto\bot$. 
These assertions were introduced in Incorrectness Separation Logic to prove that a program crashes when dereferencing invalidated memory \cite{isl}. Finally, list segments $\ls(e_1,e_2)$ are inductive predicates stating that a sequence of pointers starts with $e_1$ and ends with $e_2$. Formally, it is the least solution of $\ls(e_1, e_2) \Leftrightarrow (e_1 = e_2 \land \emp) \vee (\exists X. e_1 \mapsto X \sep \ls(X,e_2))$.
We also overload $\sep$ and $\land$ as follows:
\[
(\exists \vec X.\Pi\land\Sigma) \sep(\exists \vec Y.\Pi'\land\Sigma') \triangleq \exists \vec X\vec Y.(\Pi\land\Pi')\land (\Sigma\sep\Sigma')
\qquad
P\land \Pi \triangleq P\sep(\Pi\land\emp)
\]

\subsection{Tri-Abduction}
\label{sec:triabduction}

We now address the matter of composing paths in programs whose execution branches due to nondeterminism or random sampling.
When symbolically executing such programs, we must unify the preconditions of the two branches.
For example, the following program chooses to execute $[x]\leftarrow 1$ or $[y]\leftarrow 2$, so we need a precondition that mentions both $x$ and $y$, and we need to know what leftover resources to add to the two resulting outcomes.
\[
\xymatrix@R=-.5em@C=-.25em{
&& \ob{\ok:x\mapsto X}\ar@/_.75pc/[lld] & [x]\leftarrow 1 & \ob{\ok:x\mapsto 1}\ar@/^.75pc/[rrd] \\
\ob{\ok: \mathord{\?}} &\qquad&& + & &\qquad& {\color{purple}\langle(\ok:x\mapsto 1\sep \mathord{\?})} & {\color{purple}\oplus} & {\color{purple}(\ok:y\mapsto 2 \sep \mathord{\?}) \rangle} \\
&&  \ob{\ok:y\mapsto Y}\ar@/^.75pc/[llu]  &[y]\leftarrow 2& \ob{\ok:y\mapsto 2 }\ar@/_1pc/[rrrru]
}
\]
Tri-abduction provides us the power to reconcile the preconditions of the two program branches. Given $P_1$ and $P_2$, the goal is to find the \emph{anti-frame} $M$ and two \emph{leftover} frames $F_1$ and $F_2$ that make $P_1 \sep F_1 \Dashv M \vDash P_2 \sep F_2$ hold. Using this, we can compose  program branches as follows:
\[
\inferrule{
\triple{P_1}{C_1}{Q_1}
\\
P_1\sep F_1 \Dashv M \vDash P_2\sep F_2
\\
\triple{P_2}{C_2}{Q_2}
}
{\triple{M}{C_1 + C_2}{(Q_1\sep F_1) \oplus (Q_2 \sep F_2)}}
{\textsc{Tri-Abductive Composition}}
\]
Tri-abduction would have also been useful in Abductor, which is unable to analyze the program above despite supporting nondeterminism. Abductor operates in two passes; first finding candidate preconditions for each trace, and then re-evaluating the program with each candidate in hopes that one is valid for the entire program \cite{biab}. Since the program above uses disjoint resources in the two branches, no candidate is valid for all traces. Using tri-abduction, we infer \emph{more} summaries and do so in a single pass. While a single-pass bi-abduction algorithm has more recently been proposed \cite{sextl2023sound}, it still cannot handle the case in the following remark.

\begin{remark}[Solving Tri-Abduction using Bi-Adbuction]\label{rem:bi-tri}
Our initial approach to tri-abduction was to simply use bi-abduction: given $P_1$ and $P_2$, bi-abduction can give us $M$ and $F$ such that $P_1\sep M\vDash P_2\sep F$. Using $P_1\sep M$ as the anti-frame, $P_1 \sep M \Dashv (P_1\sep M) \vDash P_2 \sep F$ is a tri-abduction solution.

However, this approach is inherently asymmetric, with the left branch being favored. While it would be possible to also bi-abduce in the opposite direction ($P_2\sep \mathord{\textbf{?}}\vDash P_1 \sep\mathord{\textbf{?}}$) for symmetry, this still precludes valid solutions, \eg there is no bi-abduction solution for $X\mapsto Y\sep\ls(Y,Z)$ and $\ls(X,Y)\sep Y\mapsto Z$ (in either direction), whereas tri-abduction finds the anti-frame $X\mapsto Y\sep Y\mapsto Z$.
Tri-abduction is a fundamentally different operation that is precisely designed for branching.
\end{remark}

Similar to \cite[Algorithm 3]{biabjacm}, tri-abduction is done in two stages.
First, we describe the \emph{abduction} stage, in which only the anti-frame $M$ is inferred. Next, we describe how abduction is used as a subroutine to tri-abduce all three parameters $M$, $F_1$, and $F_2$.



\heading{Abduction}.
The abductive inference step $\textsf{abduce-par}(P,Q)$ is performed as a proof search---similar to \citet[Algorithm 1]{biabjacm}---using the rules in \Cref{fig:triab-pf} to infer judgements of the form $\trij{P}M{Q}$, indicating that $M\vDash P$ and $M\vDash Q$. As such, $P$ and $Q$ are the inputs to the algorithm and $M$ is the output. We describe the algorithm briefly, the full definition is in \Aref{app:triab}.

The inference rules are applied in the order in which they are shown, with the rules at the top being preferred over the rules lower down. The inference rules ending with \textsc{-L} have symmetric versions that can also be applied (the full set of rules is shown in \Aref{fig:triab-pf-full}).

The premise of each inference rule becomes a recursive call, finding a solution to a smaller abduction query. Some of the rules have side conditions of the form of $R \not \vdash \fls$, which is checked using the proof system of \citet[\S4]{berdine2005symbolic}. Given that each recursive call describes a progressively smaller symbolic heap, the algorithm either eventually reaches a case with no explicit resources ($\emp$ or $\tru$), in which a base rule applies, or gets stuck and returns no solutions.
The inference rules are described below.

\begin{figure}
  \footnotesize
  \begin{flushleft}\fbox{\textsf{Base Cases}}\end{flushleft}
  \[
    \inferrule* [rightstyle={\footnotesize \sc},right=\rulename{Base-Emp}]{
      \Pi \wedge \Pi' \not \vdash \fls
    }{
      \trij{\Pi\land\emp}{\Pi \wedge \Pi'\land \emp}{ \Pi'\land\emp}
    }
\qquad
    \inferrule* [rightstyle={\footnotesize \sc},right=\rulename{Base-True-L}]{
      \Pi \wedge \Pi'\land\Sigma' \not \vdash \fls
    }{
      \trij{\Pi\land\tru}{\Pi \wedge \Pi' \land\Sigma'}{ \Pi'\land\Sigma'}
    }
\]
  \smallskip
\begin{flushleft}\fbox{\textsf{Quantifier Elimination}}\end{flushleft}
    \[
    \inferrule* [rightstyle={\footnotesize \sc},right=\rulename{Exists}]{
      \trij{\Delta}M{\Delta'} \\ \vec X \cap (\mathsf{fv}(\Delta')\setminus \vec Y) = \emptyset \\ \vec Y\cap(\mathsf{fv}(\Delta) \setminus\vec X) = \emptyset}
      {\trij{\exists\vec X.\Delta}{\exists\vec X\vec Y.M}{\exists\vec Y.\Delta'}}
  \]
\begin{flushleft}\fbox{\textsf{Resource Matching}}\end{flushleft}
  \[
    \inferrule*[rightstyle={\footnotesize \sc},right=\rulename{Ls-Start-L}]{
      \trij{\Delta \sep \ls(e_3,e_2)}{M}{\Delta'}
    }{
      \trij{\Delta \sep \ls(e_1,e_2)}{M \sep e_1 \mapsto e_3}{\Delta'\sep e_1 \mapsto e_3}
    }
  \qquad
    \inferrule* [rightstyle={\footnotesize \sc},right=\rulename{$\mapsto$-Match}]{
      \trij{\Delta \wedge e_2 = e_3}M{\Delta' \wedge e_2 = e_3}
    }{
      \trij{\Delta \sep e_1 \mapsto e_2}{M \sep e_1 \mapsto e_2}{\Delta'\sep e_1\mapsto e_3}
    }
    \]
  \smallskip
\[
    \inferrule* [rightstyle={\footnotesize \sc},right=\rulename{Ls-End-L}]{
      \trij{\Delta \sep \ls(e_3,e_2)}{M}{\Delta'}
    }{
      \trij{\Delta \sep \ls(e_1,e_2)}{M \sep \ls(e_1,e_3)}{\Delta'\sep\ls(e_1,e_3)}
    }
  \]
\begin{flushleft}\fbox{\textsf{Resource Adding}}\end{flushleft}
  \[
  \inferrule*[rightstyle={\footnotesize \sc},right=\rulename{Missing-L}]{
    \trij{\Delta}M{\Pi'\land(\Sigma'\sep\tru)} \\ \Pi'\land\Sigma' \sep B(e_1,e_2) \not\vdash\fls
  }{
    \trij{\Delta\sep B(e_1,e_2)}{M\sep B(e_1,e_2)}{\Pi'\land(\Sigma'\sep\tru)}
  }
\quad
    \inferrule* [rightstyle={\footnotesize \sc},right=\rulename{Emp-Ls-L}]{
     \trij{\Delta\land e_1=e_2}M{\Delta' \land e_1= e_2}
    }{
      \trij{\Delta\sep\ls(e_1,e_2)}{M}{\Delta'}
    }
  \]

\caption{Selected abduction proof rules. Similarly to \citet{biab}, $B(e_1,e_2)$ represents either $\ls(e_1,e_2)$ or $e_1 \mapsto e_2$.
Rules ending in \textsc{-L} have symmetric versions not shown here; see \Aref{fig:triab-pf-full} for the full proof system.
}
\label{fig:triab-pf}
\end{figure}

\heading{Base Rules.} 
The first step is to attempt to apply a base rule to complete the proof.
\ruleref{Base-Emp} applies when both branches describe empty heaps, as long as the pure assertions in each branch do not conflict. In \ruleref{Base-True-L}, we match against the case wherein one of the branches has an arbitrary spatial assertion $\Sigma'$ and the other contains the spatial assertion $\tru$, indicating that it can absorb more resources not explicitly mentioned, so we are able to move $\Sigma'$ into the anti-frame.

\heading{Quantifier Elimination}. 
The next step is to strip existentials from the inputs $P$ and $Q$ and add them back to the anti-frame $M$ obtained from the recursive call. This is achieved using the \ruleref{Exists} rule in \Cref{fig:triab-pf}.
In bi-abduction, existentials are not stripped from the assertion to the right of the entailment---doing so prevents the algorithm from finding solutions in some cases. For example, $\ls(e,e')\sep \textbf{?} \vDash \exists X.e\mapsto X \sep \textbf{?}$ has a solution ($e\neq e'$), but it does not have a solution with the existential removed since nothing can be added to $\ls(e,e')$ to force $e$ to point to a \emph{specific} $X$.
Tri-abduction produces a standalone anti-frame $M$, so we are not operating under such constraints, allowing us to strip existentials at an early step in order to simplify further analysis.

It is important to note that in the \ruleref{Exists} rule, quantified variables in one assertion cannot overlap with the free variables of the other. This ensures that no free variables in $P$ or $Q$ end up existentially quantified in the anti-frame $M$.
Without the side condition, the rule is unsound; suppose we want to tri-abduce $\exists X. X=Y$ with $X=1$, then \textsc{Exists} gives us the anti-frame $\exists X.X=Y \land X=1$, which is too weak since $\exists X.X=1\not\vDash X=1$.
In practice, our symbolic execution algorithm always generates fresh logical variables, so we will not have collisions with our usage of the \ruleref{Exists} rule.

%
%
%
%

\heading{Resource Matching.}
If a base rule does not apply, then we attempt to match resources from both branches, and then call the algorithm recursively on smaller symbolic heaps with some resources moved into the returned anti-frame.
\ruleref{Ls-Start-L} applies when both branches contain the same resource $e_1$; however, one includes $e_1$ as the head of a list segment and the other refers to $e_1$ using a points-to predicate. The points-to predicate must be the head of the list, so we move it into $M$ and recurse on the tail of the list.
The \ruleref{$\mapsto$-Match} rule applies when both branches use $e_1$ in a points-to predicate, therefore the values pointed to must be equal too.
\ruleref{Ls-End-L} applies when both branches have list segments starting at the same address, so one segment must be a prefix of the other.

As in \citet{biab}, we do not consider cases where pointers are aliased. For example, if the two branches are $x\mapsto 1$ and $y\mapsto 1$, then it is possible that $x=y$. Precluding this solution helps limit the number of options we consider. \citet[Example 3]{biab} remark that this loss of precision is not detrimental in practice.


\heading{Resource Adding.}
Adding resources that are only present on one side is the last resort, since it involves checking a potentially expensive side condition of the form $\Pi\land\Sigma \sep B(e_1,e_2) \not\vdash\fls$. The \ruleref{Missing-L} rule handles the case wherein one branch refers to resources not present in the other. This is different from the \ruleref{Base-True-L} rule, since it handles cases where \emph{both} branches refer to resources not explicitly present in the other. For example, \ruleref{Missing-L} can solve $\trij{x \mapsto X \sep \tru}{\textbf{?}}{y \mapsto Y \sep \tru}$ even though
the \textsc{Base-True} rules do not apply.
If one side of the judgement contains a list segment, but the other side does \emph{not} contain the spatial assertion $\tru$, then there is a possible solution where the list segment is empty.
\ruleref{Emp-Ls-L} handles such cases by forcing the list segment to be empty.

\medskip

\noindent As we mentioned at the beginning of the section, the tri-abduction algorithm follows a similar structure to that of bi-abduction \cite[Algorithm 3]{biabjacm}.
\[
\mathsf{triab}(P,Q) = \{ (M, F_1, F_2) \mid M \in \textsf{abduce-par}(P\sep\tru,Q\sep\tru), \; M\vdash P\sep F_1, \; M\vdash Q\sep F_2 \}
\]
We first abduce a set of anti-frames using \Aref{alg:abduce-par} such that $M \vDash P \sep \tru$ and $M \vDash Q\sep\tru$ for each $M$. Adding the spatial assertion $\tru$ absorbs extra resources; if $P$ and $Q$ have different memory footprints, then there is no $M$ such that $\trij PMQ$, but adding $\tru$ to both sides of the entailments allows $M$ to refer to resources present in only one branch.
%
Next, we use the frame inference procedure from \citet[\S5]{berdine2005symbolic} to find $F_1$ and $F_2$ such that $M \vDash P\sep F_1$ and $M \vDash Q\sep F_2$. Applying  frame inference is necessary because $M$ may mention resources present in $P$, but not $Q$ (and vice versa). The set of solutions
is valid according to the following correctness result, which follows from the soundness of the proof system in \citet[\S5]{berdine2005symbolic}.

\begin{restatable}[Tri-abduction]{theorem}{triab}\label{thm:abduce}
  If $(M, F_1, F_2) \in \mathsf{triab}(P,Q)$, then $M \vDash P\sep F_1$ and $M \vDash Q\sep F_2$
\end{restatable}

\subsection{Symbolic Execution Algorithm}
\label{sec:alg}

\newcommand{\seq}[1]{\mathsf{seq}(\eqmakebox[seqdef][c]{$#1$}, S, \vec x)}

The algorithm is presented as a symbolic execution, which computes an abstract semantics for a program---denoted $\dea{C}(T)$---represented as a set of pairs of pre- and postconditions. The precondition is a single symbolic heap $P$, whereas the postcondition is an outcome assertion $\varphi$. The parameter $T$ is a lookup table that gets updated with summaries for procedures as the analysis moves through the codebase. Intuitively, we think about specifications as starting in a single state and producing a collection of outcomes, as the execution may branch due to effects. Abductor operates in a similar fashion, but the postcondition is a disjunction rather than an outcome conjunction. The intended semantics is captured by the following soundness result.

\begin{restatable}[Symbolic Execution Soundness]{theorem}{soundness}\label{thm:soundness}

If ~ $(P, \varphi) \in \dea{C}(T)$, ~ then ~
$\vDash\triple{\ok:P}{C}\varphi$
\end{restatable}

The strategy for the analysis is to accumulate a set of outcomes while moving forward through the program. At each step, every outcome in the current summary must be sequenced with a summary for the next command using bi-abduction. This is achieved using the \textsf{seq} procedure, defined in \Cref{fig:seq}, which takes in an outcome assertion $\varphi$, a set of summaries for the next command $C$, and $\vec x$, the variables modified by $C$. It computes a set of missing anti-frames $M$ and postconditions $\psi$ such that $\triple{\varphi\osep M}C\psi$ is a valid specification for $C$.

\begin{figure}
\footnotesize\arraycolsep=0pt
\begin{minipage}{.5\linewidth}
\begin{align*}
\seq{\top} = ~& \{(\emp, \top)\}
\\
\seq{\varphi_1\bowtie\varphi_2} = ~&\\
&\hspace{-8em}
\begin{array}{ll}
\big\{ (M, \psi'_1 \bowtie \psi'_2) \mid~&
(M_1,\psi_1) \in \mathsf{seq}(\varphi_1,S, \vec x) \\
& (M_2,\psi_2) \in \mathsf{seq}(\varphi_2,S, \vec x) \\
& (M,\psi'_1,\psi'_2) \in\mathsf{triab}'(M_1,M_2,\psi_1,\psi_2, \vec x) \big\}
\end{array}
\\
\seq{\wg\varphi{a}} =~ & \{(M, \wg\psi{a}) \mid (M, \psi) \in \mathsf{seq}( \varphi, S, \vec x) \}
\\
\seq{\ok:P} = ~& \\
&\hspace{-8em}
\big\{(M, \psi') \mid
(Q,\psi) \in S,
(M,\psi') \in \mathsf{biab}'(P,Q, \psi, \vec x)
\big\}
\\
\seq{\er:Q} = ~&  \{(\emp, \er:Q)\}
\end{align*}
\end{minipage}%
\begin{minipage}{.5\linewidth}
\begin{align*}
&\mathsf{biab}'(\exists \vec Z.\Delta, Q, \psi, \vec x) = \\
&\quad
\begin{array}{l}
\big\{ (M', (\psi \osep \exists \vec Z\vec X.F[\vec X/\vec x])[\vec e/ \vec Y]) \\
\mid (M, F) \in \mathsf{biab}(\Delta,Q) \\
\quad (\vec e, \vec Y, M') = \mathsf{rename}(\Delta, M, Q, \{\psi\}, \vec X, \vec x) \big\}
\end{array}
\\\medskip
&\mathsf{triab}'(P_1, P_2, \psi_1, \psi_2, \vec x) = \\
&\quad
\begin{array}{l}
\big\{ (M', (\psi_1 \osep \exists \vec X.F_1[\vec X/\vec x])[\vec e/\vec Y], \\
\phantom{\big\{ (M',\,} (\psi_2 \osep \exists \vec X.F_2[\vec X/\vec x])[\vec e/\vec Y]) \\
~\mid (M, F_1, F_2) \in \mathsf{triab}(P_1, P_2) \\
\quad  (\vec e, \vec Y, M') = \mathsf{rename}(\emp, M, \{\psi_1, \psi_2\}, \emptyset, \vec x) \big\}
\end{array}
\end{align*}
\end{minipage}
%
%

\caption{Sequencing procedure. The vector $\vec X$ is assumed to be fresh and the same size as $\vec x$, and $\mathord{\bowtie}\in\{\vee,\oplus\}$.}
\label{fig:seq}
\end{figure}

Sequencing after $\top$ and $(\er:Q)$ has no effect, since $\top$ carries no information about the current branch, and $(\er:Q)$ means the program has crashed. Sequencing after $\wg\varphi{a}$ simply sequences $\varphi$ and then reapplies the weight. Sequential composition is implemented in the $(\ok:P)$ case, where bi-abduction is used to reconcile the current outcome with each summary for the next command. The $\mathsf{biab}'$ procedure is similar to \textsf{AbduceAndAdapt} from \citet[Fig. 4]{biabjacm}, in which a renaming step is applied to ensure that the anti-frame $M$ is not phrased in terms of any program variables, therefore meeting the side condition of the frame rule. The details of renaming and why it is required for correctness are explained in \Aref{app:symexec}.

Our algorithm differs from Abductor in the cases with multiple program branches. This is where we use tri-abduction to obtain a precondition that is guaranteed to be valid for all program paths, allowing us to analyze the program in a single pass (unlike Abductor, which must re-evaluate the program using each candidate precondition). After sequencing each of the outcomes in the precondition $\varphi_1$ and $\varphi_2$ with the next command, we use $\mathsf{triab}'$ to obtain the single renamed anti-frame $M$ that is safe for both branches.
The soundness property for \textsf{seq} is stated below.

\begin{restatable}[Seq]{lemma}{sequencing}
If $(M,\psi)\in\mathsf{seq}(\varphi,S,\vec x)$, $\vec x = \mathsf{mod}(C)$, and $\vDash\triple{\ok:P}C\vartheta$ for all $(P,\vartheta)\in S$, then $\vDash\triple{\varphi\osep M}C\psi$.
\end{restatable}

\heading{Symbolic Execution Algorithm.}
The core symbolic execution algorithm, shown in \Cref{fig:sym_exec}, computes a \emph{local} symbolic semantics which can be augmented using the frame rule to obtain summaries in larger heaps. For example, the semantics for $\skp$ is simply the triple $\triple{\ok:\emp}\skp{\ok:\emp}$, but this implies that running the program in any heap will yield that same heap in the end.
Executing $C_1\fatsemi C_2$ uses $\mathsf{seq}$; summaries for $C_1$ are sequenced with all the summaries for $C_2$. Choices $C_1+ C_2$ are analyzed by computing summaries for each path and reconciling them with tri-abduction.

We split $\assume e$ into two cases: if $e$ is a simple test (equality or inequality), then we generate two specifications with the precondition stating that $e$ is true and the postcondition being unchanged, or $e$ being false and the outcome being eliminated. This is similar to the \emph{``assume-as-assert''} behavior of Abductor and produces precise specifications without a priori knowledge of the logical conditions that will occur.
If $e$ is a weight literal $a$, then we weight the current outcome by it. Other expressions are not supported due to limitations of the bi-abduction solver of \citet{biab}, which we also use.
Similarly, \textsf{while} loops use a least fixed point to unroll the loop, and produce one summary for each possible number of iterations. We will see more options for analyzing loops later on.

The abstract semantics of atomic actions mostly follow the small axioms of \citet{localreasoning}, with failure cases inspired by Incorrectness Separation Logic \cite{isl}. Each memory operation has three specifications: one in which the pointer is allocated and the operation accordingly succeeds, and two failure cases where the pointer is not allocated or null. Procedure calls rely on pre-computed summaries in a lookup table $T$, which is a parameter to $\dea{C}$.

%

\begin{figure}
\footnotesize
\[\def\arraystretch{1.25}
\begin{array}{l|l}
\arraycolsep=0pt
C \in\mathsf{Cmd} & \dea{C}(T) \\
\hline\hline
\skp & \{(\emp, \ok:\emp)\}\\
C_1\fatsemi C_2 & \{ (P\sep M,\psi) \mid (P,\varphi)\in\dea{C_1}(T), (M, \psi) \in \mathsf{seq}(\varphi,\dea{C_2}(T), \mathsf{mod}(C_2)) \} \\
C_1+ C_2 &
\arraycolsep=0pt
\begin{array}{ll}
\{ (M, \psi'_1 \oplus \psi'_2) \mid \;\;&
 (M_1, \psi_1) \in \dea{C_1}(T),  (M_2, \psi_2)\in\dea{C_2}(T), \\
& (M, \psi'_1, \psi'_2)\in \mathsf{triab}'(M_1, M_2, \psi_1, \psi_2, \mathsf{mod}(C_1,C_2)) \}
\end{array}
\\
\assume {b} & \{ (b \land \emp, \ok: b\land\emp), ( \lnot b \land\emp, \wg\top\zero \} 
\\
\assume{a} & \{ (\emp, \wg{(\ok:\emp)}a \} 
\\
\whl bC &
\arraycolsep=0pt
\begin{array}{ll}
\mathsf{lfp} ~S. &
\{(\lnot b\land\emp, \ok:\lnot b\land \emp)\} \mathop{\cup} \\
& \{ (M_1\sep M_2 \land b,  \psi) \\
& \mid (M_1,\varphi) \in \mathsf{seq}(\ok:b\land\emp, \dea C(T), \mathsf{mod}(C)), (M_2, \psi) \in \mathsf{seq}(\varphi, S, \mathsf{mod}(C)) \}
\end{array}
\\\\
c\in\mathsf{Act} & \dea{c}(T)\\
\hline\hline
x\coloneqq e & \{ (x = X \land \emp, \ok: x = e[X/x] \land \emp) \}\\
x \coloneqq \mathsf{alloc}() & \{ (x=X\land\emp, \ok: \exists Y. x \mapsto Y) \}\\
\mathsf{free}(e) & \{ (e \mapsto X, \ok:e\not\mapsto), (e\not\mapsto, \er: e\not\mapsto), (\emp\land e = \mathsf{null}, \er:\emp\land e=\mathsf{null}) \}\\
{[e_1]\leftarrow e_2} & \{ (e_1 \mapsto X, \ok: e_1 \mapsto e_2), (e_1 \not\mapsto, \er: e_1\not\mapsto), (\emp\land e_1 = \mathsf{null}, \er:\emp\land e_1=\mathsf{null})\}\\
x \leftarrow [e] & \{ (x=X \!\land\! e\mapsto Y, \ok: x = Y \!\land\! e[X/x]\mapsto Y), \!(e\not\mapsto, \er\!:\! e\not\mapsto), \!(\emp\!\land\! e\!=\! \mathsf{null}, \!\er\!:\!\emp\!\land\! e\!=\!\mathsf{null})\}\\
\mathsf{error}() & \{ (\emp, \er:\emp) \}\\
f(\vec e) & \{ (P\land \vec x = \vec X , \varphi) \mid (P, \varphi) \in \mathsf{seq}(\ok:\vec x = \vec e[\vec X/\vec x]\land \emp, T(f(\vec x)), \mathsf{mod}(f)) \}) \}
\end{array}
\]
\caption{Symbolic execution of commands and actions, all logical variables $X,Y\in\mathsf{LVar}$ are assumed to be fresh, $a\in A$ is a program weight, and $b \Coloneqq e_1 = e_2 \mid e_1 \neq e_2$ is a simple test.}
\label{fig:sym_exec}
\end{figure}

\heading{Simulating Pulse}. As recounted by \citet{il,isl,realbugs}, the scalability of IL-based analyses stems from their ability to \emph{drop disjuncts}. Analyzers such as Abductor accumulate a disjunction of possible end states at each program point. For incorrectness, it is not necessary to remember \emph{all} of these possible states; \emph{any} of them leading to an error constitutes a bug.
Since IL allows \emph{strengthening} of postconditions, those disjuncts can be soundly dropped.

We take a slightly different view, which nonetheless enables us to drop \emph{paths} in the same way. We differentiate between program choices that result from logical conditions (\ie if and while statements) vs computational effects (\ie nondeterministic or probabilistic choice). In the former cases, we generate multiple summaries in order to precisely keep track of which initial states will result in which outcomes. In the latter case, we use an outcome conjunction rather than a disjunction to join the outcomes. While we cannot drop outcomes per se, we can replace them by $\top$, ensuring that they will not be explored any further according to the definition of \textsf{seq} (\Cref{fig:seq}). 
This can be implemented by altering the definition of choice as follows.
\[
\dea{C_1 + C_2}(T) = \{ (P, \varphi \oplus \top) \mid (P,\varphi) \in \dea{C_1}(T) \} \cup \{ (P, \varphi \oplus \top) \mid (P,\varphi) \in \dea{C_2}(T) \} 
\]
Given this modification, the algorithm remains sound with respect to the same semantics (\ie \Cref{thm:soundness}), but it no longer fits with the spirit of correctness reasoning, since some of the program outcomes are left unspecified. We will see in \Cref{sec:example_realloc} how it can be used for efficient bug-finding.

Now, we have a situation where each element of $\dea{C}(T)$ stands alone as a sound summary for the program $C$, and represents a particular trace. This corresponds to taking a particular logical branch in if statements, unfolding while loops for a fixed number of iterations, and specific (nondeterministic or probabilistic) choices.
So, eliminating elements from the set will only preclude possible summaries without affecting the correctness of the existing ones, meaning that we can drop paths by, \eg unrolling loops for a certain number of iterations and only taking one path when encountering a $+$.
\citet{realbugs} refer to this as depth and width of the analysis, respectively.

\heading{Loop invariants and partial correctness.}
An alternative to bounded unrolling for deterministic and non-deterministic programs (but not probabilistic ones) is to use loop invariants.
We achieve this by altering the rule for while loops to the following.
\[
\dea{\whl eC}(T) = \{ (I, (\ok:I\land\lnot e)\vee \wg\top\zero \mid (I\land e, \ok: I) \in \dea{C}(T) \}
\]
The truth of the invariant $I$ is preserved by the loop body, therefore it must remain true \emph{if} the loop exits. However, the invariant cannot guarantee that the loop terminates, meaning that we lose the typical reachability guarantees of outcome logic and must use a weaker \emph{partial correctness} specification.
The possibility of nontermination is expressed by the disjunction with $\wg\top\zero$---either $I \land \lnot e$ holds, or the program diverges and there are no outcomes. Since reachability cannot be guaranteed, loop invariants are not suitable for true positive bug-finding, in which errors must be witnessed by an actual trace.


Finding loop invariants is generally undecidable, however techniques from abstract interpretation \cite{cousot1977abstract} can be used to find invariants by framing the problem as a fixed point computation over a finite domain, thereby guaranteeing convergence. This is the approach taken in Abductor \cite{biabjacm}, which uses the same symbolic heaps, but without outcome conjunctions. In nondeterministic programs, we can convert outcome conjunctions into disjunctions since $(\ok:P)\oplus (\ok:Q) \Rightarrow (\ok:P\vee Q)$. We therefore speculate that we can use the same technique as Abductor, although we leave a complete exploration of this idea to future work.


\heading{Nondeterministic Allocation}. Memory bugs can arise in C from failing to check whether the address returned by \textsf{malloc} is non-null. This is often modeled using nondeterminism, wherein the semantics of \textsf{malloc} returns either a valid pointer or null, nondeterministically. Our language is generic over effects, so we do not have a nondeterministic \textsf{malloc} operation, but we can add $x\coloneqq\mathsf{malloc}()$ as syntactic sugar for $(x\coloneqq \mathsf{alloc}()) + (x\coloneqq\mathsf{null})$, and derive the following semantics:
\[
\dea{x\coloneqq\mathsf{malloc}()}(T) = \{ (x = X\land\emp, (\ok:x = \mathsf{null}\land\emp) \oplus (\ok:\exists Y.x\mapsto Y)) \}
\]

\heading{Reusing Summaries.}
Though partial correctness specifications are incompatible with bug-finding, and under-approximate specifications are incompatible with verification, there is still overlap in summaries that can be used for both. Many procedures in a given codebase do not include loops or branching, so their summaries are equally valid for both correctness and incorrectness, and also in programs with different interpretations of choice.
In other cases, when a procedure does have multiple outcomes, it is easy to convert a correctness specification into several individual incorrectness ones, since the following implication is sound. We will see this in action in \Cref{sec:example_realloc}.
\[
\triple{\ok:P}{C}{\psi_1 \oplus \psi_2}
\qquad\implies\qquad
\triple{\ok:P}{C}{\psi_1 \oplus \top}
\quad\text{and}\quad
\triple{\ok:P}{C}{\psi_2 \oplus \top}
\]
%

\section{Case Studies}
\label{sec:examples}

We will now demonstrate how the symbolic execution algorithms work by examining two case studies, which show the applicability in both nondeterministic and probabilistic execution models.

\subsection{Nondeterministic Vector Reallocation}
\label{sec:example_realloc}

Our first case study involves a common error in C++ when using the \texttt{std::vector} library in which a call to \texttt{push\_back} may reallocate the vector's underlying memory buffer, invalidating any pointers to that cell that existed before the call. This was also used as a motivating example for Incorrectness Separation Logic \cite{isl}. Following their lead, we model the vector as a single pointer and we treat reallocation as nondeterministic.
The program is shown below.
\[
\begin{array}{ll}
\begin{array}{l}
\textsf{main}():\\
\quad x \leftarrow [v]\fatsemi\\
\quad \texttt{push\_back}(v)\fatsemi\\
\quad {[x]}\leftarrow 1
\end{array}
&
\qquad\qquad
\begin{array}{l}
\textsf{push\_back}(v):\\
\hspace{1em}\begin{array}{lll}
\left(\begin{array}{l}
y \leftarrow [v]\fatsemi\\
\textsf{free}(y)\fatsemi\\
y \coloneqq \textsf{alloc}()\fatsemi\\
{[v]} \leftarrow y
\end{array}
\right)
&+&
\skp
\end{array}
\end{array}
\end{array}
\]
Before we analyze the \textsf{main} procedure, we must store $\dea{\mathsf{push\_back}(v)}(T)$ in the procedure table. Since \textsf{push\_back} is a common library function, it makes sense to compute summaries describing all the outcomes, which will be reusable for both correctness and incorrectness analyses. The first step is to analyze the two nondeterministic branches, which are both simple sequential programs.
\begin{align*}
\small\arraycolsep=0pt
\dea{\begin{array}{l}y \leftarrow [v]\fatsemi\\\textsf{free}(y)\fatsemi\\ y \coloneqq \textsf{alloc}()\fatsemi\\ {[v]} \leftarrow y\end{array}}\!\!\!\!(T) &= \left\{
\begin{array}{lll}
(v\mapsto A \sep A\mapsto -, ~\;&\ok: \exists B. v\mapsto B \sep B\mapsto - \sep A\not\mapsto &) \\
(v\mapsto A \sep A\not\mapsto, &\er:v\mapsto A \sep A\not\mapsto &) \\
(v\not\mapsto, &\er: v\not\mapsto&) \\
&\cdots
\end{array}\right\}
\\
\dea{\skp}\!(T) &= \large\{ (\emp, \ok:\emp)\large\}
\end{align*}
Now, we can compose the two program branches using tri-abduction. Choosing the first summary for the first branch, we get the following tri-abduction solution.
\[
v\mapsto A \sep A\mapsto - \sep [\emp] \Dashv [v\mapsto A \sep A\mapsto -] \vDash \emp \sep [v\mapsto A \sep A\mapsto -]
\]
So, by framing $\emp$ into the first branch and $v\mapsto A \sep A\mapsto -$ into the second branch, we get a summary for \textsf{push\_back} as a whole. This can similarly be done for the other summaries of the first branch, yielding the lookup table below.
\[{\arraycolsep=1pt\small
\begin{array}{ll}
T = \left\{\begin{array}{lllclll}
{\color{purple}\langle \ok:} & {\color{purple} v\mapsto A \sep A\mapsto -} & {\color{purple}\rangle} & \mathsf{push\_back}(v) & {\color{purple}\langle (\ok:\exists B.v\mapsto B \sep B\mapsto - \sep {A\not\mapsto}} & {\color{purple} ) \oplus (\ok:v\mapsto A \sep A\mapsto -} & {\color{purple} )\rangle}
\\
{\color{purple}\langle \ok:} & {\color{purple} v\mapsto A \sep {A\not\mapsto}} & {\color{purple}\rangle} & \mathsf{push\_back}(v) & {\color{purple}\langle (\er: v\mapsto A \sep {A\not\mapsto}} & {\color{purple} ) \oplus (\ok:v\mapsto A \sep {A\not\mapsto}} & {\color{purple} )\rangle}
\\
{\color{purple}\langle \ok:} & {\color{purple} {v\not\mapsto}} & {\color{purple}\rangle} & \mathsf{push\_back}(v) & {\color{purple}\langle (\er: {v\not\mapsto}} & {\color{purple} ) \oplus (\ok:{v\not\mapsto}} & {\color{purple})\rangle}
\\
&&& \cdots
\end{array}\right\}
\end{array}}
\]
The first summary tells us that \textsf{push\_back} may reallocate the underlying buffer, in which case the original pointer $A$ will become deallocated. The next two summaries describe ways in which \textsf{push\_back} itself can fail. We will focus on using the first summary to show how \textsf{main} will fail if the buffer gets reallocated.
We analyze \textsf{main} in an under-approximate fashion in order to look for bugs. The first step is to compute summaries for the first two commands of main. The load on the first line has three summaries according to \Cref{fig:sym_exec}, we select the first one in which $v$ is allocated.
\[\footnotesize
\triple{\ok: x = X \land v\mapsto Y}{x\leftarrow [v]}{\ok: x=Y \land v\mapsto Y} \in \dea{x\leftarrow[v]}(T)
\]
The procedure call on the second line requires us to look up summaries in $T$. We select the first one, but we will use an under-approximate version of it so as to explore only one of the paths
\[\footnotesize
\triple{\ok:v\mapsto A\sep A\mapsto-}{\mathsf{push\_back}(v)}{(\ok: \exists B.v\mapsto B \sep B \mapsto -\sep{A\not\mapsto}) \oplus \top}
\]
Now, we use \textsf{seq} to sequentially compose these summaries, which involves bi-abducing the postcondition of $x\leftarrow [v]$ with the precondition of $\mathsf{push\_back}(v)$.
\[
x=Y \land v\mapsto Y \sep [A=Y \sep x\mapsto -] \vDash v\mapsto A\sep A\mapsto- \sep [\emp]
\]
So, after renaming, we get the following summary for the composed program:
\[\footnotesize
\triple{\ok:v\mapsto x\sep x\mapsto-}{x\leftarrow[v] \fatsemi \mathsf{push\_back}(v)}{(\ok: \exists B.v\mapsto B \sep B \mapsto -\sep{x\not\mapsto}) \oplus \top}
\]
Now, observe that the postcondition above is only compatible with one of the summaries in \Cref{fig:sym_exec} for the last line of the program. Since $x$ is deallocated in the only specified outcome, the write into $x$ must fail.
Using bi-abduction again, we can construct the following description of the error.
\[\footnotesize
\triple{\ok:v\mapsto x\sep x\mapsto-}{x\leftarrow[v] \fatsemi \mathsf{push\_back}(v)\fatsemi [x] \leftarrow 1}{(\er: \exists B.v\mapsto B \sep B \mapsto -\sep{x\not\mapsto}) \oplus \top}
\]

\subsection{Consensus in Distributed Computing}


In this case study, we use OSL to lower bound reliability rates of a distributed system. In the basic consensus algorithm, shown below,
each of three processes broadcasts a value $v_i$ by storing it in a pointer $p_i$, and consensus is reached if any two of these processes broadcasted the same value. To model unreliability of the network, the broadcast procedure fails with probability 1\%. We would like to know how likely we are to reach consensus, given that two of the processes agree.
\[
\begin{array}{l}
\mathsf{main}(): \\
\quad p_1 \coloneqq \mathsf{alloc}() \fatsemi \mathsf{broadcast}(v_1, p_1) \fatsemi \\
\quad p_2 \coloneqq \mathsf{alloc}() \fatsemi \mathsf{broadcast}(v_2, p_2) \fatsemi \\
\quad p_3 \coloneqq \mathsf{alloc}() \fatsemi \mathsf{broadcast}(v_3, p_3) \fatsemi \\
\quad v \coloneqq \mathsf{alloc}() \fatsemi \\
\quad \mathsf{decide}(p_1,p_2,p_3, v)
\\\\
\mathsf{broadcast}(v, p): \\
\quad
\left([p] \leftarrow v\right)
\oplus_{0.99}
\mathsf{error}()
\end{array}
\qquad\qquad
\begin{array}{l}
\mathsf{decide}(p_1, p_2, p_3, v): \\
\quad x_1 \leftarrow [p_1] \fatsemi 
 x_2 \leftarrow [p_2] \fatsemi 
 x_3 \leftarrow [p_3] \fatsemi \\
\quad \mathsf{if}~x_1 = x_2~\mathsf{then} \\
\quad\quad [v] \leftarrow x_1 \\
\quad\mathsf{else\ if}~x_1 = x_3~\mathsf{then} \\
\quad\quad [v] \leftarrow x_1 \\
\quad\mathsf{else\ if}~x_2 = x_3~\mathsf{then} \\
\quad\quad [v] \leftarrow x_2 \\
\quad\mathsf{else}~\skp
\end{array}
\]
We begin by examining the summary table. The \textsf{broadcast} procedure has two outcomes corresponding to whether or not the communication went through. Though there are many summaries for \textsf{decide}, we show only the one in which the values sent on $p_1$ and $p_2$ are equal.
\[\small
T = \left\{\begin{array}{c}
\triple{\ok: p \mapsto - \land v = V}{\mathsf{broadcast}(v,p)}{(\ok: p\mapsto V \land v = V) \oplus_{0.99} (\er:p \mapsto - \land v = V)} \\
\triple{\ok: p_1 \mapsto V_1 \sep p_2 \mapsto V_2 \sep p_3\mapsto V_3 \sep v\mapsto -\land V_1 = V_2}{\mathsf{decide}(p_1,p_2,p_3,v)}{\ok: v\mapsto V_1\sep \cdots} \\
\cdots
\end{array}\right\}
\]
We again use the single-path algorithm to analyze \textsf{main}, but this time we are interested only in the successfully terminating cases. We get the following summaries for each of the first three lines.
\[\mathsmaller{
\triple{\ok: v_i = V_i}{p_i \coloneqq \mathsf{alloc}() \fatsemi \mathsf{broadcast}(v_i, p_i)}{(\ok: v_i = V_i \land p_i \mapsto V_i) \oplus_{0.99} \top}}
\]
These three summaries can be combined---along with the simplification that $(\varphi \oplus_a \top) \oplus_b \top$ implies $\varphi \oplus_{a\cdot b} \top$---to obtain the following
assertion just before the call to \textsf{decide}
\[
(\ok:v_1 =V_1 \land v_2=V_2\land v_3 = V_3 \land p_1\mapsto V_1 \sep p_2 \mapsto V_2 \sep p_3\mapsto V_3) \oplus_{0.99^3} \top
\]
Now, we can bi-abduce the first outcome above with the precondition for \textsf{decide} shown in $T$. This sends $V_1 = V_2$ backwards into the precondition, and we get an overall summary for \textsf{main} telling us that if $v_1 = v_2$, then the protocol will reach consensus ($v\mapsto V_1$) with probability \emph{at least} 97\%.
\[
\triple{\ok:v_1 = V_1 \land v_2 = V_2 \land v_3 = V_3 \land V_1 = V_2}{\mathsf{main}()}{(\ok: v \mapsto V_1 \sep \cdots) \oplus_{0.9703} \top}
\]

\section{Related Work}
\label{sec:related}

\noindent\textbf{\textit{Separation Logic and the Frame Rule.}} While many variants of separation logic exist, OSL is the only one that supports alternative effects as well as both may and must properties. The soundness of the frame rule in both partial and total correctness separation logic relies on nondeterminism and \emph{must} properties \cite{yang2002semantic,calcagno2007local}, so those logics are suitable only for correctness in nondeterministic languages. Some work has been done to drop the nondeterminism requirement, for example \citet{baktiev2006} proved that the frame rule is sound in a deterministic language if heap assertions are unaffected by address permutation, however this requires languages without address arithmetic.
Similarly, \citet{tatsuta2009completeness} created a deterministic separation logic, but the frame rule only applies to programs that do not allocate memory.

Incorrectness Separation Logic (ISL) has a frame rule that is compatible with may properties, but not must properties \cite{isl}. Unlike OSL, ISL has no restrictions on assertions about memory allocation; it is possible to prove that a particular address $\ell$ is returned. If, after applying the frame rule, $\ell$ is already in use, then the postcondition of the ISL triple becomes false, making the triple vacuous (whereas in regular separation logic and OSL, a false \emph{pre}condition makes triples vacuous). However, whereas ISL has slightly more power to express may properties, ISL cannot express must (\ie correctness) properties and is also specialized only to nondeterminism.

Exact Separation Logic (ESL) \cite{maksimovi_c2023exact} combines the semantics of SL and ISL to create a logic that is suitable for both correctness and incorrectness, while inheriting the limitations of both. That is, ESL can only express properties that are \emph{both} may and must properties, meaning that it cannot under-approximate by dropping paths or over-approximate using loop invariants.

Higher-order separation logics \cite{birkedal2007relational,birkedal2008simple} including Iris \cite{iris1} bake the frame rule into the definition of the logic's triples, making the proof of soundness of the frame rule trivial without any additional assumptions. More precisely, triples $\hoare PCQ$ are valid iff for any frame $F$, if $\sigma\vDash P\sep F$, then $\tau\vDash Q\sep F$ for all $\tau \in \de{C}(\sigma)$. As a tradeoff, frame baking introduces additional proof obligations whenever constructing a triple (\ie it must be shown that every inference rule is \emph{frame preserving}). We could have used this approach in OSL, but there is no free lunch; it would have involved moving the inductive cases of \Aref{lem:frameprop} into \Cref{thm:soundness}.

We chose not to do this,
as we preferred to keep the proof of the frame rule self contained. The choice is mostly orthogonal to the power of the logic, and somewhat philosophical: we proved that programs are local \emph{actions} whereas the frame baking approach shows that you can only \emph{express} local properties about commands. The exception to this is memory allocation, which is not a local action, so we baked the expressivity limitation into the triple semantics. We felt this was a good tradeoff, since we baked in a weaker property, which was just strong enough to complete the proof.

In draft papers published after our work on Outcome Separation Logic, Sufficient Incorrectness Logic (SIL) \cite{ascari2023sufficient} and Backwards Under-Approximate triples (BUA) \cite{raad2023compositional} provide frame rules for triples based on \emph{may} properties. In both cases, the soundness result is a bit weaker than our own (\Cref{thm:frame}). That is, whereas we only require in the premise that the triple $\triple{\varphi}{C}{\psi}$ is \emph{semantically} valid, \citet{ascari2023sufficient,raad2023compositional} require that it is \emph{derivable}. While this trick allows them to avoid baking the allocation restriction into the triple validity, it means that the soundness of the frame rule depends on the other inference rules in the proof system---a property that we wished to avoid. In addition, the inference rules for allocation in those logics are not tight, meaning that the logics are incomplete. Adding a tighter rule (with the ability to witness \emph{which} address was nondeterministically chosen) would make the frame rule unsound, meaning that completeness of those logics is incompatible with the frame rule.

\heading{Probabilistic Separation Logics.}
While some work has been done to incorporate probabilistic reasoning into separation logic, these logics differ from OSL in the scope and applicability of the frame rule, and have not been shown to be compatible with bi-abduction.
Quantitative Separation Logic \cite{qsl} and its concurrent counterpart \cite{fesefeldt2022towards} use weakest pre-expectation \cite{wpe} style predicate transformers to derive expected values in probabilistic pointer programs. They rely on demonic nondeterminism for allocation---the expected value is lower bounded over all possible allocated addresses---and the frame rule gives a lower bound, whereas OSL is used for propositional reasoning, with each outcome having a likelihood.

Polaris \cite{polaris} incorporates probabilistic reasoning into Iris \cite{iris1} and is also used to bound expected values using refinements inspired by probabilistic relational Hoare Logic \cite{barthe2015relational}. Polaris is limited to programs that terminate in finitely many steps and the program logic itself is only used to relate probabilistic programs to each other, whereas the quantitative reasoning about expected values must be done externally to the program logic.



Probabilistic Separation Logic (PSL) \cite{psl} and subsequent works \cite{bao2021bunched,bao2022separation,li2023lilac} use an alternative model of separation to characterize probabilistic independence and related probability theoretic properties.
Doing so provides a compositional way to reason about probabilistic programs, though this work is orthogonal to our own as it does not deal with heaps.

\heading{Pulse and Incorrectness Separation Logic.}
As recounted by \citet[\S5]{isl}, Pulse uses under-approximation in four ways in order to achieve scalability:
\begin{enumerate}
\item Pulse takes advantage of the IL semantics in order to explore only one path at a time when the program execution branches, and to unroll loops for a bounded number of iterations.
\item Pulse elects to not consider cases in which memory is re-allocated.
\item Pulse uses under-approximate specifications for some library functions.
\item Pulse's bi-abductive inference assumes that pointers are not aliased unless explicitly stated. 
\end{enumerate}
We have shown how (1) is achieved using OSL in the single path algorithm, (2) and (4) are standard assumptions in bi-abduction \cite{biab,biabjacm} (which we also use), and (3) is a corollary to (1), since the ability to drop paths opens the possibility for under-approximate procedure summaries.

Pulse does not support inductive predicates (\eg list segments), so it uses a simplified bi-abduction procedure capable of handling more types of pure assertions. This results in \emph{exact} bi-abduction solutions; the inferred $M$ and $F$ satisfy $P\sep M \Dashv \vDash Q\sep F$. As a result, Pulse does not use consequences to---since it is based on IL---strengthen the postcondition, meaning that the resulting specs can be interpreted as both ISL and OSL triples and that our algorithm from \Cref{sec:symexec} does accurately model Pulse. This finding was later confirmed by \citet{raad2023compositional}.


\heading{Unified Metatheory with Effects.}
Our algebraic program semantics is similar to Weighted Programming \cite{batz2022weighted}. Whereas we use an algebraic interpretation of choice to represent multiple types of (executable) program semantics, the goal of weighted programming is to specify mathematical models and find solutions to optimization problems via static analysis.
\citet{c_irstea2013branching,c_irstea2014coalgebraic} also used partial semirings to represent properties of program branching in coalgebraic logics. 

\citet{delaware2013meta,delaware2013modular} developed 3MT, a unified framework in the Coq proof assistant for mechanizing metatheoretic proofs---such as type soundness--about languages with monadic effects. Our motivations are similar, but with the goal of developing program logics.

\heading{Unifying correctness and incorrectness.}
In the time since \citet{il} introduced Incorrectness Logic, a considerable amount of research attention has turned to unifying the theories of correctness and incorrectness.
Exact Separation Logic is one such approach, which combines the semantics of standard separation logic and Incorrectness Separation Logic to support correctness and incorrectness within a single program logic \cite{maksimovi_c2023exact}. It is used within the Gillian symbolic execution system \cite{fragoso_santos2020gillian}. The idea behind Exact Separation Logic is to create an analysis that is as precise as possible, up until it reaches a point where it must make the choice to move into correctness or incorrectness mode. However, the exact nature means that it does not have the flexibility to, \eg reason in abstract domains.

Local Completeness Logic is a related attempt to combine abstract interpretation with true bug-finding \cite{Bruni2021ALF,brunijacm}. It uses a standard over-approximate abstract domain to analyze programs, combined with Incorrectness Logic to ensure that the over-approximation has not gotten too coarse so as to include false alarms. However, if a bug depends on nondeterminism (such as the malloc example in \Cref{sec:overview}), then the over-approximate abstraction will preclude the analysis from \emph{dropping paths} to only explore the trace where the error occurs, an ability that was identified by \citet{il} as crucial to large scale bug-finding.

Though the goals of Exact Separation Logic and Local Completeness Logic are similar to our own, we take a substantially different approach. Combining Hoare and Incorrectness Logic inevitably yields something that is \emph{less than the sum of its parts}; the resulting logic inherits all the limitations of both, compromising the ability to \emph{approximate} the program behavior by either abstracting the current state or dropping paths. Instead of combining two existing logics, Outcome Separation Logic represents fundamentally new ideas and supports all the reasoning principles needed for efficient correctness and incorrectness analysis.

Like Exact Separation Logic and Local Completeness Logic, Outcome Separation Logic guarantees reachability for true-bug finding, while also covering all of the possible end states for correctness verification. But, unlike the aforementioned logics, it also supports weakening via standard forward consequences, so abstraction of the current state is always possible.
Abstraction is crucial to the scalability of techniques such as \emph{shape analysis} \cite{rinetzky2001interprocedural,berdine2007shape}, in which the analyzer attempts to prove memory safety while only tracking the \emph{shapes} of data structures on the heap, not the data or length (like the list segment predicate from \Cref{sec:symexec}). Outcome Separation Logic makes shape analysis compatible with true bug finding for the first time. More generally, a key benefit of Outcome Separation Logic (over Exact Separation Logic and Local Completeness Logic) is that it can identify true bugs while retaining much less information, potentially leading to better scalability.

Hyper Hoare Logic \cite{hyperhoare} was designed to prove {\em hyperproperties} about programs by reasoning about multiple executions simultaneously. Examples of such properties include noninterference---guaranteeing that a program does not leak sensitive information.  The design of Hyper Hoare Logic takes a similar approach to standard Outcome Logic \cite{outcome,zilberstein2024relatively}, in which both correctness and incorrectness can be proven for a nondeterministic language. In fact, Hyper Hoare Logic has the same semantics as the nondeterministic instance of OL, though OL can also cover other effects such as probabilistic computation. Putting aside the semantic similarities between Outcome Logic and Hyper Hoare Logic, the goal of this particular paper was to augment OL with heaps, separation, {local reasoning} (via the frame rule), and bi-abduction---none of which are currently supported by Hyper Hoare Logic.

\section{Conclusion}
\label{sec:conclusion}

Infer---based on separation logic and bi-abduction---is capable of efficiently analyzing industrial scale codebases, substantiating the idea that compositionality translates to real-world scalability \cite{infer}. But the deployment and integration of Infer into real-world engineering workflows also surfaced that proving the \emph{absence} of bugs is somewhat of a red herring---software has bugs and sound logical theories are needed to find them \cite{realbugs}.

Incorrectness Logic has shown that it is not only \emph{possible} to formulate a theory for bug-finding, but it is in fact advantageous from a program analysis view; static analyzers can take certain liberties in searching for bugs that are not valid for correctness verification, such as dropping program paths for added efficiency. The downside is that the IL semantics is incompatible with correctness analysis, therefore separate implementations and procedure summaries must be used.

In OSL, we seek to get the best of both worlds. 
As \citet[\S6]{isl} put it, ``aiming for under-approximate results rather than exact ones gives additional flexibility to the analysis designer, just as aiming for over-approximate rather than exact results does for correctness tools.'' 
The fact that OSL supports over-approximation in the traditional sense as well as under-approximation in the sense of Pulse
invites the reuse of tools between the two, while still enabling specialized techniques when needed (\ie loop invariants for correctness, dropping paths for incorrectness). In addition, OSL extends bi-abduction to programs with effects such as randomization for the first time.

OSL is not a simple extension of separation logic; it is designed from the ground up with new assumptions, since the properties that make the standard SL frame rule sound (nondeterministic allocation, must properties, and fault avoidance) are not suitable for reasoning about incorrectness and effects. The addition of tri-abduction to our symbolic execution algorithms also means that we can analyze more programs with control flow branching compared to Abductor.
The power and flexibility of OSL makes it a strong foundation for analyzing pointer-programs with effects.

\section*{Acknowledgements}

This work was supported in part by an Amazon Research Award, a Royal Society Wolfson award, and ERC grant Autoprobe (grant agreement 101002697). 


\bibliographystyle{ACM-Reference-Format}
\bibliography{biabduction}

\ifx\extended\undefined\else
\newpage
\appendix
\allowdisplaybreaks

{\noindent \huge\bfseries\sffamily Appendix}

\section{Program Semantics}
\label{app:semantics}

We begin by providing definitions for natural orders, completeness and Scott continuity, which are mentioned in \Cref{def:ocalg}.

\begin{definition}[Natural Ordering]
The \emph{natural order} of a semiring $\langle A, +, \cdot, \zero, \one\rangle$ is defined to be $a\le b$ iff $\exists a'\in A.  a+a' = b$. A semiring is \emph{naturally ordered} if $\le$ is a partial order.
Note that $\le$ is reflexive and transitive by the semiring laws, so it remains only to show that it is anti-symmetric. That is, if $a\le b$ and $b\le a$, then $a=b$.
\end{definition}

\begin{definition}[Complete Partial Semiring]
A partial semiring $\langle A, +, \cdot, \zero, \one\rangle$ is complete if there is a sum operator $\sum_{i\in I}x_i$ such that the following properties hold:
\begin{enumerate}
\item If $I = \{i_1, \ldots, i_n\}$ is finite, then $\sum_{i\in I} x_i =  x_{i_1} + \cdots + x_{i_n}$
\item If $\sum_{i\in I} x_i$ is defined, then $b\cdot \sum_{i\in I} x_i = \sum_{i\in I} b\cdot x_i$ and $(\sum_{i\in I} x_i)\cdot b = \sum_{i\in I} x_i\cdot b$ for any $b\in A$
\item Let $(J_k)_{k\in K}$ be any family of nonempty disjoint subsets of $I$, so $I = \bigcup_{k\in K} J_k$ and $J_k \cap J_\ell = \emptyset$ if $k\neq \ell$. Then, $\sum_{k \in K}\sum_{j\in J_k} x_j = \sum_{i\in I}x_i$.
\end{enumerate}
This definition is adapted from \citet[Chapter 3]{golan2003semirings}.
\end{definition}

\begin{definition}[Scott Continuity]
Consider a semiring $\langle A, +, \cdot, \zero, \one\rangle$ with partial order $\le$.
A function (or partial function) $f\colon A\to A$ is Scott continuous if for any directed set $D\subseteq A$ (where all pairs of elements in $D$ have a supremum), $\sup_{a\in D}f(a) = f(\sup D)$.
A semiring is Scott continuous if $\sum$ and $\cdot$ are Scott continuous in both arguments~\cite{semirings}.
\end{definition}

\noindent In addition, we recall the definition of normalizable, which is stated as property (3) in \Cref{def:ocalg}.

\begin{definition}[Normalizable]\label{def:normalize}
A semiring $\langle A, +, \cdot, \zero, \one\rangle$ is normalizable if for any well defined sum $\sum_{i\in I}a_i$, there exists $(b_i)_{i\in I}$ such that $\sum_{i\in I}b_i = \one$ and $a_i = (\sum_{i\in I}a_i) \cdot b_i$ for every $i\in I$.
\end{definition}

\noindent Normalizability is needed to show that relations lifted by the $\mathcal W_{\mathcal A}$ functor have several properties, for example that lifting is well behaved with respect to sums and scalar multiplication (\Cref{lem:plusrel,lem:scalerel}). We also show that normalization implies the weaker row-column property below:

\begin{definition}[Row-Column Property]\label{def:rowcol}
A monoid $\langle A, +, \zero\rangle$ has the row-column property if for any two sequences of elements $(a_i)_{i\in I}$ and $(b_j)_{j\in J}$, if $\sum_{i\in I}a_i = \sum_{j\in J} b_j$, then there exist $(u_k)_{k\in I\times J}$ such that $\sum_{j\in J} u_{(i,j)} = a_i$ for all $i\in I$ and $\sum_{i\in I} u_{(i,j)} = b_j$ for all $j\in J$.
\end{definition}

\begin{lemma}\label{lem:normrowcol}
Let $\langle A, +, \cdot, \zero, \one\rangle$ be a normalizable semiring (\Cref{def:normalize}), then $\langle A, +, \zero\rangle$ has the row-column property (\Cref{def:rowcol}).
\end{lemma}
\begin{proof}
Take any sequences $(a_i)_{i\in I}$ and $(b_j)_{j\in J}$ such that $\sum_{i\in I}a_i = \sum_{j\in J} b_j$ and let $x = \sum_{i\in I}a_i = \sum_{j\in J} b_j$. Now, since the semiring is normalizable, there must be $(a'_i)_{i\in I}$ and $(b'_j)_{j\in J}$ such that $\sum_{i\in I}a'_i = \sum_{j\in J} b'_j = \one$ and $a_i = x\cdot a'_i$ for all $i\in I$ and $b_j = x\cdot b'_j$ for all $j\in J$. Let $u_{(i,j)} = x \cdot a'_i \cdot b'_j$. Now we have:
\[
\sum_{i\in I} u_{(i,j)} = \sum_{i\in I} x \cdot a'_i \cdot b'_j = x\cdot (\sum_{i\in I}a'_i)\cdot b'_j = x\cdot\one\cdot b'_j = x\cdot b'_j = b_j
\]
And
\[
\sum_{j\in J} u_{(i,j)} = \sum_{j\in J} x \cdot a'_i \cdot b'_j = x\cdot a'_i\cdot (\sum_{j\in J} b'_j) = x\cdot a'_i\cdot\one = x\cdot a'_i = a_i
\]
\end{proof}

In fact, it is known that if some monoid $A$ has the row-column property, then the functor $\mathcal F_A$ of finitely supported maps into $A$ preserves weak pullbacks~\cite{gumm2009copower,moss1999coalgebraic,klin2009structural}. It has also been shown that relations lifted by some functor preserve composition iff the functor preserves weak pullbacks~\cite{kurz2016relation}. Combining these two results, we get that relations lifted by $\mathcal F_A$ preserve composition if $A$ has the row-column property.

In our case, the maps have countable support rather than finite, so the aforementioned results do not immediately apply, however they do provide evidence that the row-column property is a reasonable requirement. However, we require the stronger normalization property since lifted relations also must also be well behaved with respect to scalar multiplication (\Cref{lem:scalerel}).

\begin{remark}
In \Cref{def:ocalg} we require that $\sup(A) = \one$, which limits the models to be contractive maps (the weight of the computation can only decrease as the program executes). This rules out, for example, real-valued multisets where the weights are elements of the semiring $\langle \mathbb R^\infty, +, \cdot, 0,1\rangle$. Still, there are more models than the ones we have presented including the tropical semiring $\langle \mathbb R_0^\infty, \min, +, \infty, 0 \rangle$, which can be used to encode optimization problems~\cite{batz2022weighted}.

The fact that $\sup(A)=\one$ is used in order to define $\repl$ as a lifted relation (\Cref{app:replacement}), and it is also used in the proof of \Cref{lem:sequencing}.
It may be possible to relax this constraint, however the more general version is not needed for the models we explore in this paper.
\end{remark}

\subsection{Proofs}


\begin{lemma}[Scaled Sums]\label{lem:scaledsums}
If $\sum_{i\in I} x_i$ is defined, then $\sum_{i\in I} x_i\cdot y_i$ is defined for any $(y_i)_{i\in I}$.
\end{lemma}
\begin{proof}
Since $\sup(A)$ exists, then $y_i\le\sup(A)$ for each $i\in I$. By the definition of $\le$, there exist $(y'_i)_{i\in I}$ such that $y_i + y'_i = \sup(A)$. We also know that $(\sum_{i\in I}x_i)\cdot\sup(A)$ exists, since $A$ is closed under multiplication. Now, we have:
\begin{align*}
(\smashoperator{\sum_{i\in I}} x_i)\cdot\sup(A)
&= \smashoperator{\sum_{i\in I}}x_i\cdot (y_i + y'_i) \\
\intertext{And by the semiring laws:}
&= \smashoperator{\sum_{i\in I}}x_i\cdot y_i + x_i\cdot y'_i \\
&= \smashoperator{\sum_{i\in I}} x_i\cdot y_i + \smashoperator{\sum_{i\in I}}x_i\cdot y'_i
\end{align*}
So, clearly the subexpression ${\sum_{i\in I}} x_i\cdot y_i$ is defined.
\end{proof}

\begin{lemma}[Totality of Bind]\label{lem:bindtotal}
The $\bind$ function defined in \Cref{def:ocmonad} is a total function (this is not immediate, since it uses partial addition).
\end{lemma}
\begin{proof}
First, we note that $\sum_{a\in\supp(m)}m(a)$ must be defined by the definition of $\mathcal W$. By \Cref{lem:scaledsums}, we know that $\sum_{a\in\supp(m)} m(a)\cdot f(a)(b)$ must be defined too, for any family $(f(a)(b))_{a\in \supp(m)}$. Now, since $\bind(m,f)(b) = \sum_{a\in\supp(m)}m(a)\cdot f(a)(b)$, then it must be a total function. The result is also countably supported, since $m$ and each $f(a)$ are countably supported and $\supp(\bind(m,f)) = \cup_{a\in\supp(m)} \supp(f(a))$. Finally $|\bind(m, f)|$ exists since:
\begin{align*}
|\bind(m,f)|
&= \smashoperator[l]{\sum_{b\in\supp(\bind(m,f))}} \smashoperator[r]{\sum_{a\in\supp(m)}} m(a)\cdot f(a)(b) \\
&=  \smashoperator{\sum_{a\in\supp(m)}} m(a)\cdot \quad\smashoperator{\sum_{b\in\supp(\bind(m,f))}} f(a)(b) \\
&=  \smashoperator{\sum_{a\in\supp(m)}} m(a)\cdot \smashoperator{\sum_{b\in\supp(f(a))}} f(a)(b) 
=  \smashoperator{\sum_{a\in\supp(m)}} m(a)\cdot |f(a)|
\end{align*}
And the sum on the last line exists by \Cref{lem:scaledsums}.

\end{proof}


\begin{theorem}[Fixed Point Existence]\label{thm:fp}
The function $F_{\langle C, e,\af\rangle}$ defined in \Cref{fig:cmdsem} has a least fixed point.
\end{theorem}

\begin{proof}
It will suffice to show that $F_{\langle C, e, \af\rangle}$ is Scott continuous, at which point, we can apply the Kleene fixed point theorem to conclude that the least fixed point exists. First, we define the pointwise order for $f_1,f_2 \colon \mathcal S\times\mathcal H \to \mathcal W(\st)$ as $f_1 \sqsubseteq f_2$ iff $f_1(s,h) \sqsubseteq f_2(s,h)$ for all $(s,h)$, where $f_1(s,h) \sqsubseteq f_2(s,h)$ iff there exists $m$ such that $f_1(s,h) + m = f_2(s,h)$.
Now, we will show that the monad bind is Scott continuous with respect to that order. Let $D$ be a directed set.
\begin{align*}
\sup_{f\in D}\bind(m, f)
&= \sup_{f\in D} \sum_{s\in\supp(m)} m(s)\cdot \left\{\begin{array}{ll}
f(a) & \text{if}~ s = \inj_\ok(a) \\
\unit_M(s) & \text{otherwise}
\end{array}\right.\\
\intertext{Now, by continuity of the semiring, suprema distribute over sums and products. It is relatively easy to see by induction that the supremum can move into every summand in the series.}
&= \smashoperator{\sum_{s\in\supp(m)}} m(s)\cdot \left\{\begin{array}{ll}
\sup_{f\in D} f(a) & \text{if}~ s = \inj_\ok(a) \\
\unit_M(s) & \text{otherwise}
\end{array}\right.\\
&= \smashoperator{\sum_{s\in\supp(m)}} m(s)\cdot \left\{\begin{array}{ll}
(\sup D)(a) & \text{if}~ s = \inj_\ok(a) \\
\unit_M(s) & \text{otherwise}
\end{array}\right.\\
&= \bind(m, \sup D)
\end{align*}
Finally, we show that $F_{\langle C, e,\af\rangle}$ is Scott continuous with respect to the order defined above.
\begin{align*}
\sup_{f \in D} F_{\langle C, e,\af\rangle}(f)
&= \lambda(s,h). \sup_{f\in D} F_{\langle C, e, \af\rangle}(f)(s,h) \\
&= \lambda(s,h).\sup_{f\in D} \left\{\begin{array}{ll}
\bind(\de{C}_\af(s,h), f) & \text{if}~ \de{e}(s) = \one \\
\unit(s,h) & \text{if}~ \de{e}(s) = \zero
\end{array}\right. \\
&= \lambda(s,h). \left\{\begin{array}{ll}
\sup_{f\in D}\bind(\de{C}_\af(s,h), f) & \text{if}~ \de{e}(s) = \one \\
\sup_{f\in D}\unit(s,h) & \text{if}~ \de{e}(s) = \zero
\end{array}\right. \\
&= \lambda(s,h). \left\{\begin{array}{ll}
\bind(\de{C}_\af(s,h), \sup D) & \text{if}~ \de{e}(s) = \one \\
\unit(s,h) & \text{if}~ \de{e}(s) = \zero
\end{array}\right. \\
&= F_{\langle C, e, \af\rangle}(\sup D)
\end{align*}
\end{proof}

Before proving that the semantics is total, we define the notion of compatible expressions.
\begin{definition}[Compatible Expressions]\label{def:compat}
Given some outcome algebra $\langle A, +, \cdot, \zero, \one\rangle$, two expressions $e$ and $e'$ are \emph{compatible} if $\de{e}(s) +\de{e'}(s)$ is defined for all $s\in\mathcal S$.
\end{definition}

Though compatibility is a semantic notion, there are straightforward syntactic checks to ensure that two expressions are compatible. For example, if $e$ is a test---a Boolean-valued expression where $\de{e}(s) \in \{\zero, \one\}$ for all $s\in\mathcal S$---then $e$ and $\lnot e$ are compatible. In the probabilistic semiring, $p$ and $1-p$ are also compatible for any $p\in[0,1]$.

\begin{theorem}[Totality of Program Semantics]
Given an outcome algebra $\langle A, +, \cdot, \zero, \one\rangle$, the semantics of a program $\de{C}_\af(s,h)$ is defined as long as the following conditions are met:
\begin{enumerate}
\item For each loop $\whl e{C'}$ appearing in $C$, the guard $e$ must be a test. More precisely, $\de{e}(s) \in \{\zero, \one\}$ for every $s\in\mathcal S$.
\item If the semiring addition is partial, then each use of $+$ in $C$ must be guarded by compatible expressions. That is, they must have the form $(\assume{e_1}\fatsemi C_1) + (\assume{e_2}\fatsemi C_2)$ where $e_1$ and $e_2$ are compatible according to \Cref{def:compat}.
\end{enumerate}
\end{theorem}
\begin{proof}
By induction on the structure of $C$. The cases for $\skp$, $\mathsf{assume}$, and all atoms are trivial.
\begin{itemize}
\item $C = C_1\fatsemi C_2$. By the induction hypothesis, we assume $\de{C_1}_\af$ and $\de{C_2}_\af$ are total, and by \Cref{lem:bindtotal} we know that bind is total, therefore $\bind(\de{C_1}_\af(s,h),\de{C_2}_\af)$ is total.
\item $C = C_1 + C_2$. If the semiring addition is total, then this case follows trivially from the induction hypothesis. If not, then the program must have the form $(\assume{e_1}\fatsemi C_1) + (\assume{e_2}\fatsemi C_2)$ where $e_1$ and $e_2$ are compatible, so we have:
\begin{align*}
\de{C}_\af(s,h)
&= \de{(\assume{e_1}\fatsemi C_1) + (\assume{e_2}\fatsemi C_2)}_\af(s,h) \\
&= \de{e_1}(s) \cdot \de{C_1}(s,h) + \de{e_2}(s)\cdot \de{C_2}(s, h)
\end{align*}
By the induction hypothesis, $\de{C_1}(s,h)$ and $\de{C_2}(s,h)$ are defined. Since $e_1$ and $e_2$ are compatible, $\de{e_1}(s) + \de{e_2}(s)$ is defined too. So, the claim follows by \Cref{lem:scaledsums}.

\item $C = \whl eC$. Follows from \Cref{thm:fp}.
\end{itemize}
\end{proof}

\section{Properties of Lifted Relations}

\begin{lemma}\label{lem:relsize}
If $(m_1, m_2)\in \overline R$, then $|m_1| = |m_2|$.
\end{lemma}
\begin{proof}
We know that there exists an $m$ such that:
\[
m_1 = \lambda x.\smashoperator{\sum_{y\in\supp(m_2)}} m(x,y)
\qquad\text{and}\qquad
m_2 = \lambda y.\smashoperator{\sum_{x\in\supp(m_1)}} m(x,y)
\]
Now, we have:
\begin{align*}
|m_1|
&= |\lambda x.\smashoperator{\sum_{y\in\supp(m_2)}} m(x,y)| \\
&= \sum_{x\in\supp(m_1)}\sum_{y\in\supp(m_2)} m(x,y) \\
&= \sum_{y\in\supp(m_2)} \sum_{x\in\supp(m_1)} m(x,y) \\
&= |\lambda y.\smashoperator{\sum_{x\in\supp(m_1)}} m(x,y)| = |m_2|
\end{align*}
\end{proof}

\begin{lemma}\label{lem:zerorel}
$(\zero, m)\in \overline R$ iff $m=\zero$.
\end{lemma}
\begin{proof}\;

\iffcases{
We know there must be some $m'$ such that $m = \lambda y.\sum_{x\in\supp(\zero)} m'(x, y)$, but since $\supp(\zero) =\emptyset$, then $m = \zero$.
}{
Let $m' = \zero$, so clearly $m'\in\mathcal W_{\mathcal A}R$ and we also have $\lambda x.\sum_{y\in\supp(\zero)}(\zero(x,y)) = \zero$ and $\lambda y.\sum_{x\in\supp(\zero)}(\zero(x,y)) = \zero$
}
\end{proof}

\begin{lemma}\label{lem:summagdef}
If $|m_1|+|m_2|$ is defined, then $m_1+m_2$ is defined.
\end{lemma}
\begin{proof}
Observe that:
\begin{align*}
|m_1|+|m_2| &= \smashoperator{\sum_{x\in\supp(m_1)}}m_1(x) + \smashoperator{\sum_{x\in\supp(m_2)}} m_2(x) \\
\intertext{By associativity:}
&= \smashoperator{\sum_{x\in\supp(m_1)\cup\supp(m_2)}} m_1(x) + m_2(x)
\end{align*}
Therefore $m_1(x) + m_2(x)$ is defined for all $x$, and $|m_1 + m_2|$ is defined as well, so $m_1+m_2$ is defined.
\end{proof}

\begin{lemma}\label{lem:plusrel}
$(m_1 + m_2, m') \in \overline R$ iff there exist $m'_1$ and $m'_2$ such that $m' = m'_1+m'_2$ and $(m_1, m'_1)\in \overline R$ and $(m_2, m'_2)\in \overline R$.
\end{lemma}
\begin{proof}\;

\iffcases{
We know that there is some $m$ such that $m_1 + m_2 = \lambda x.\sum_{y\in\supp(m')} m(x,y)$ and $m' = \lambda y.\sum_{x\in\supp(m_1+m_2)} m(x,y)$.
So, for each $x$, we have that $m_1(x) + m_2(x) = \sum_{y\in\supp(m')} m(x,y)$. Using \Cref{lem:normrowcol}, we know there must be $(x_k)_{k\in \{1,2\}\times \supp(m')}$ such that $m_i(x) = \sum_{y\in\supp(m')} x_{(i, y)}$ for $i\in \{1,2\}$ and $m(x,y) = x_{(1,y)} + x_{(2,y)}$ for each $y\in\supp(m')$.
Now let $m_1''(x,y) = x_{(1, y)}$ and $m''_2(x,y) = x_{(2, y)}$ and
let $m_1'= \lambda y.\sum_{x\in\supp(m_1)} m_1''(x,y)$ and $m_2' = \lambda y.\sum_{x\in\supp(m_2)} m_2''(x,y)$. First, we establish that $m'_1 + m'_2 = m'$:
\begin{align*}
m_1' + m_2' &= \lambda y. \smashoperator{\sum_{x\in\supp(m_1)}} m_1''(x,y) + \smashoperator{\sum_{x\in\supp(m_2)}} m_2''(x,y) 
= \lambda y. \smashoperator{\sum_{x\in\supp(m_1)}} x_{(1, y)} + \smashoperator{\sum_{x\in\supp(m_2)}} x_{(2,y)} \\
\intertext{If $x\notin\supp(m_1)$, then $x_{(1, y)} = \zero$ and similarly for $x_{(2,y)}$ and $\supp(m_2)$, so we can combine the sums}
&= \lambda y. \smashoperator{\sum_{x\in\supp(m_1+m_2)}} x_{(1,y)} + x_{(2,y)} 
= \lambda y. \smashoperator{\sum_{x\in\supp(m_1+m_2)}} m(x,y) = m'
\end{align*}
Now, observe that:
\begin{align*}
\lambda x.\smashoperator{\sum_{y \in \supp(m'_i)}} m_i''(x,y)
&= \lambda x.\smashoperator{\sum_{y \in \supp(m'_i)}} x_{(i, y)} \\
\intertext{Now, for any $y \notin \supp(m'_i)$, it must be the case that $m'_i(y) = \zero$ and so $\sum_{x\in\supp(m_i)} x_{(i, y)} = \zero$, so each $x_{(i, y)} = \zero$. This means we can expand the sum to be over $\supp(m'_1+m'_2) = \supp(m')$.}
&= \lambda x.\smashoperator{\sum_{y\in\supp(m')}} x_{(i, y)}  = m_i
\end{align*}
We also know that $m_i'= \lambda y.\sum_{x\in\supp(m_i)} m_i''(x,y)$ by definition, so $(m_i, m'_i) \in \overline R$.
}{
We know that $(m_1, m'_1)\in\overline R$ and $(m_2, m'_2)\in \overline R$, so there are $m''_1$ and $m''_2$ such that $m_i = \lambda x.\sum_{y\in\supp(m'_i)} m''_i(x,y)$ and $m'_i = \lambda y.\sum_{x\in\supp(m_i)} m''_i(x,y)$ for $i\in\{1,2\}$. Let $m'' = m''_1 + m''_2$ (we can conclude that this is defined using \Cref{lem:relsize,lem:summagdef}). Now we have the following:
\[
\lambda x.\smashoperator{\sum_{y\in\supp(m')}} m''(x, y)
= \lambda x.\smashoperator{\sum_{y\in\supp(m')}} m_1''(x, y) + m_2''(x,y) 
= \lambda x.\smashoperator{\sum_{y\in\supp(m'_1)}} m_1''(x, y) + \smashoperator{\sum_{y\in\supp(m'_2)}} m_2''(x,y) 
= m_1 + m_2
\]
\[
\lambda y.\smashoperator{\sum_{x\in\supp(m_1+m_2)}} m''(x,y) 
= \lambda y.\smashoperator{\sum_{x\in\supp(m_1+m_2)}} m_1''(x,y) + m_2''(x,y) 
= \lambda y.\smashoperator{\sum_{x\in\supp(m_1)}} m_1''(x,y) + \smashoperator{\sum_{x\in\supp(m_2)}} m_2''(x,y) 
= m'_1 + m'_2 = m'
\]
So, $(m_1+m_2, m')\in \overline R$.
}
\end{proof}

\begin{lemma}\label{lem:scalerel}
$(a\cdot m_1, m_2)\in \overline R$ iff there exists $m'_2$ such that $m_2 = a\cdot m'_2$ and $(m_1, m'_2)\in \overline R$
\end{lemma}
\begin{proof}\;
\iffcases{
By the definition of relation lifting, there is an $m$ such that $a\cdot m_1 = \lambda x.\sum_{y\in \supp(m_2)}m(x,y)$ and $m_2 = \lambda y.\sum_{x\in\supp(a\cdot m_1)} m(x,y)$. This means that for all $x$:
\[
a\cdot m_1(x) = \smashoperator{\sum_{y\in \supp(m_2)}}m(x,y)
\]
By \Cref{def:normalize}, we can obtain $(b_{(x,y)})_{y\in\supp(m_2)}$ such that $\sum_{y\in\supp(m_2)} b_{(x,y)} = \one$ and $m(x,y) = ({\sum_{z\in \supp(m_2)}}m(x,z))\cdot b_{(x,y)}$ for all $y\in\supp(m_2)$. Now, define $m''(x,y) = m_1(x)\cdot b_{(x,y)}$ and $m_2'(y) = \sum_{x\in\supp(m_1)}m''(x,y)$.
We now show that $m_2 = a\cdot m'_2$:
\begin{align*}
a\cdot m'_2
&=\lambda y.a\cdot \smashoperator{\sum_{x\in\supp(m_1)}}m''(x,y) \\
&= \lambda y.a\cdot \smashoperator{\sum_{x\in\supp(m_1)}} m_1(x)\cdot b_{(x,y)} \\
&= \lambda y.\smashoperator{\sum_{x\in\supp(m_1)}} a\cdot m_1(x)\cdot b_{(x,y)} \\
&= \lambda y. \smashoperator[l]{\sum_{x\in\supp(m_1)}} (\smashoperator[r]{\sum_{z\in\supp(m_2)}} m(x,z))\cdot b_{(x,y)} \\
&= \lambda y. \smashoperator{\sum_{x\in\supp(m_1)}} m(x,y) = m_2
\end{align*}
We also have:
\begin{align*}
\lambda x.\smashoperator{\sum_{y\in \supp(m'_2)}} m''(x,y)
&= \lambda x.\smashoperator{\sum_{y\in \supp(m'_2)}} m_1(x)\cdot b_{(x,y)} \\
\intertext{Since we already showed that $m_2 = a\cdot m'_2$, then it must be the case that $\supp(m_2) = \supp(m_2')$.}
&= \lambda x. m_1(x)\cdot \smashoperator{\sum_{y\in \supp(m_2)}} b_{(x,y)}
= \lambda x.m_1(x) \cdot\one = m_1
\end{align*}
And clearly $m_2' = \lambda y. \sum_{x\in\supp(m_1)}m''(x,y)$ by definition, so $(m_1, m_2') \in\overline R$.
}{
By the definition of relation lifting, there is some $m$ such that $m_1 = \lambda x.\sum_{y\in\supp(m'_2)} m(x,y)$ and $m_2' = \lambda y.\sum_{x\in\supp(m_1)} m(x,y)$.
Now, let $m' = a\cdot m$, so this clearly means that $a\cdot m_1 = \lambda x.\sum_{y\in\supp(a\cdot m'_2)} m'(x,y)$ and $a\cdot m'_2 = \lambda y.\sum_{x\in\supp(a\cdot m_1)} m'(x,y)$, so $(a\cdot m_1, a\cdot m'_2)\in\overline R$.
}
\end{proof}

\begin{lemma}\label{lem:unitrel}
If $(x,y) \in R$, then $(\unit(x), \unit(y))\in \overline R$
\end{lemma}
\begin{proof}
Let $m = \unit(x,y)$. We therefore have:
\[
\lambda x'.\smashoperator{\sum_{y'\in\supp(\unit(y))}} m(x',y') = \lambda x' .\unit(x,y)(x',y) = \unit(x)
\]
And similarly, $\lambda y'.\sum_{x'\in\supp(\unit(x))} m(x',y') = \unit(y)$, so $(\unit(x), \unit(y))\in \overline R$.
\end{proof}

\begin{lemma}\label{lem:relcomp}
For any relations $R\subseteq X\times Y$ and $S\subseteq Y\times Z$, if $(m_1, m_2)\in\overline{S\circ R}$, then $(m_1, m_2)\in \overline S \circ \overline R$.

\end{lemma}
\begin{proof}
Suppose $(m_1, m_2)\in \overline{S\circ R}$, so there is some $m\in\mathcal{W}_{\mathcal A}(S\circ R)$ such that $m_1 = \lambda x.\sum_{z \in \supp(m_2)} m(x,z)$ and $m_2 = \lambda z.\sum_{x\in\supp(m_1)} m(x, z)$. This means that for each $(x,z)\in\supp(m)$, there is some $y$ such that $(x,y)\in R$ and $(y,z)\in S$. Let $S' \subseteq S$ be some relation where we choose one $y$ for each $z$, so that  $|\{ y \mid (y, z)\in S' \}| = 1$ for each $z$ and $\{z \mid \exists y. (y, z)\in S'\} = \{z \mid \exists y.(y,z) \in S\}$. Now, let:
\[
m'(y,z) = \left\{\begin{array}{ll}
m_2(z) & \text{if}~ (y,z)\in S' \\
\zero & \text{if}~ (y,z)\notin S'
\end{array}\right.
\quad
m''(x,y) = \smashoperator{\sum_{z\in\supp(m_2) \mid (y,z)\in S'}} m(x,z)
\quad
m_3 = \lambda y.\smashoperator{\sum_{z\in\supp(m_2)}} m'(y,z)
\]
And now, we have the following:
\begin{align*}
\lambda x.\smashoperator{\sum_{y \in \supp(m_3)}} m''(x,y)
&= \lambda x.\smashoperator[l]{\sum_{y \in \supp(m_3)}} \smashoperator[r]{\sum_{z\in\supp(m_2) \mid (y,z)\in S'}} m(x,z) \\
\intertext{Since $S'$ relates each $z$ to exactly one $y\in\supp(m_3)$, this is equivalent to summing over all $z$}
&= \lambda x.\smashoperator{\sum_{z\in\supp(m_2)}} m(x,z)  = m_1
\\\\
\lambda y.\smashoperator{\sum_{x\in\supp(m_1)}} m''(x,y)
&= \lambda y.\smashoperator[l]{\sum_{x \in \supp(m_1)}} \smashoperator[r]{\sum_{z\in\supp(m_2) \mid (y,z)\in S'}} m(x,z) \\
&= \lambda y. \smashoperator[l]{\sum_{z\in\supp(m_2) \mid (y,z)\in S'}} \smashoperator[r]{\sum_{x \in \supp(m_1)}}m(x,z) \\
&= \lambda y. \smashoperator{\sum_{z\in\supp(m_2) \mid (y,z)\in S'}} m_2(z) 
= \lambda y. \smashoperator{\sum_{z\in\supp(m_2)}} m'(y,z) = m_3 
\end{align*}
So, $(m_1, m_3)\in \overline R$. Also, by definition we know that $m_3 = \lambda y.\sum_{z\in\supp(m_2)} m'(y,z)$ and $m_2 = \lambda z.{\sum_{y\in\supp(m_3)}} m'(y,z)$ since $m'(y,z)$ is nonzero for exactly one $y\in\supp(m_3)$ and is equal to $m_2(z)$ at that point. This means that $(m_3, m_2)\in \overline S$, therefore $(m_1, m_2)\in \overline S\circ \overline R$.

\end{proof}

\begin{lemma}\label{lem:suprel}
For any pair of directed chains $m_{(i,1)} \sqsubseteq m_{(i,2)} \sqsubseteq \cdots$ for $i\in \{1,2\}$,
if $(m_{(1,n)}, m_{(2,n)}) \in \overline R$ for all $n\in\mathbb N$, then $(\sup_{n\in \mathbb N} m_{(1,n)}, \sup_{n\in \mathbb N} m_{(2,n)}) \in \overline R$.
\end{lemma}
\begin{proof}
For any $n\in\mathbb{N}$, we know that there is an $m_n$ such that $m_{(1,n)} = \lambda x.\sum_{y\in\supp(m_{(2,n)})} m_n(x,y)$ and $m_{(2,n)} = \lambda y.\sum_{x\in\supp(m_{(1,n)})} m_n(x,y)$.
%
%
Since each $(m_{(i,n)})_{n\in\mathbb{N}}$ is a chain, $\sup_{n\in\mathbb{N}} m_{(i,n)}$ exists, and:
\begin{align*}
\sup_{n\in\mathbb{N}} m_{(1, n)}
&= \sup_{n\in\mathbb{N}} (\lambda x.\smashoperator{\sum_{y\in\supp(m_{(2,n)})}} m_n(x,y))\\
\intertext{Since we use a pointwise order for functions, the $\sup$ of a function is equal to the $\sup$ at each point. Additionally, since the semiring is continuous, the $\sup$ distributes over the sum.}
&= \lambda x.\smashoperator[l]{\sum_{y\in\supp(\sup_{n\in\mathbb{N}}m_{(2,n)})}} \sup_{n\in\mathbb{N}} m_n(x,y)
\end{align*}
And by a similar argument, $\sup_{n\in\mathbb{N}} m_{(2, n)} = \lambda y.\sum_{x\in\supp(\sup_{n\in\mathbb{N}}m_{(1,n)})} \sup_{n\in\mathbb{N}} m_n(x,y)$, so:
\[(\sup_{n\in \mathbb N} m_{(1,n)}, \sup_{n\in \mathbb N} m_{(2,n)}) \in \overline R\]
\end{proof}

\section{Outcome Separation Logic}

\begin{lemma}[Normalization]\label{lem:normalize}
For any $m\neq\zero$, there exists $m'$ such that $|m'| = \one$ and $m = |m|\cdot m'$.
\end{lemma}
\begin{proof}
By property (3) of \Cref{def:ocalg}, there must be $(b_s)_{s\in\supp(m)}$ such that $m(s) = (\sum_{t\in\supp(m)}m(t))\cdot b_s$ and $\sum_{s\in\supp(m)}b_s = \one$. Now, let $m'$ be defined as follows:
\[
m'(s) = \left\{
\begin{array}{ll}
b_s & \text{if}~s \in \supp(m) \\
\zero & \text{if}~s\notin\supp(m)
\end{array}\right.
\]
So, clearly $|m'| = \sum_{s\in\supp(m)}b_s = \one$. For every $s$, we also have $m(s) = (\sum_{t\in\supp(m)}m(t))\cdot b_s = |m|\cdot b_s = |m|\cdot m'(s)$, so $m = |m|\cdot m'$.
\end{proof}

\begin{lemma}[Splitting]\label{lem:splitting}
If $|m'| \le   |m_1| + |m_2|$, then there exist $m'_1$ and $m'_2$ such that $|m'_1| \le |m_1|$ and $|m_2'| \le |m_2|$ and $m' = m'_1 + m'_2$.
\end{lemma}
\begin{proof}
Since $|m_1| + |m_2| \ge |m'|$, then there is some $a$ such that:
\[
|m_1| + |m_2| = |m'| + a = a + \smashoperator{\sum_{s \in \supp(m')}} m'(s)
\]
So, by \Cref{lem:normrowcol}, there exists $(u_k)_{k \in \{1,2\}\times (1+\supp(m'))}$ such that for all $i\in\{1,2\}$ and $s\in\supp(m')$:
\[
|m_i| = \smashoperator{\sum_{s \in 1+\supp(m')}} u_{(i, s)}
\qquad\text{and}\qquad
a = u_{(1, \star)} + u_{(2, \star)}
\qquad\text{and}\qquad
m'(s) = u_{(1, s)} + u_{(2, s)}
\]
Now, let $m'_i = \lambda s. u_{(i, s)}$ for $i\in\{1, 2\}$. So, $m'_1 + m'_2 = \lambda s. u_{(1, s)} + u_{(2, s)} = m'$. Now, it just remains to show that $|m_i| \ge |m'_i|$:
\[
|m_i|
= \smashoperator{\sum_{s \in 1+\supp(m')}} u_{(i, s)}
= u_{(i, \star)} + \smashoperator{\sum_{s \in\supp(m')}} m'_i(s)
= u_{(i, \star)} + |m'_i|
\ge |m'_i|
\]

\end{proof}

\subsection{The Outcome Separating Conjunction}

\osepforward*
\begin{proof} By induction on the structure of $\varphi$.
\begin{itemize}
\item $\varphi = \top$. Since $\varphi\osep F = \top$, then clearly $m'\vDash\varphi\osep F$
\item $\varphi = \varphi_1 \vee \varphi_2$. We know $m\vDash\varphi_1$ or $m\vDash\varphi_2$. Without loss of generality, suppose that $m\vDash\varphi_1$. By the induction hypothesis, we know that $m'\vDash \varphi_1\osep F$. We can therefore weaken this to conclude that $m'\vDash(\varphi_1\vee\varphi_2)\osep F$. The case where $m\vDash\varphi_2$ is symmetrical.

\item $\varphi = \varphi_1 \oplus \varphi_2$. We know that $m_1\vDash\varphi_1$ and $m_2\vDash\varphi_2$ for some $m_1$ and $m_2$ such that $m = m_1 + m_2$. 
Now, since $(m_1 + m_2, m') \in \overline{\fr(F)}$, by \Cref{lem:plusrel} there must be $m'_1$ and $m'_2$ such that $(m_1, m'_1) \in \overline{\fr(F)}$ and $(m_2, m'_2) \in \overline{\fr(F)}$ and $m' = m'_1 + m'_2$.
By the induction hypothesis, $m'_1\vDash \varphi_1\osep F$ and $m'_2\vDash \varphi_2\osep F$, so $m'\vDash (\varphi_1 \oplus \varphi_2)\osep F$.

\item $\varphi = \wg{\varphi'}a$. We know that $a=\zero$ and $m=\zero$ or $m_1 \vDash \varphi'$ for some $m_1$ such that $m = a\cdot m_1$.
In the first case, by \Cref{lem:zerorel}, $m' = \zero$ and so clearly $m' \vDash \wg{\varphi'}\zero \osep F$.
In the second case, since $(a\cdot m_1, m') \in \overline{\fr(A)}$, by \Cref{lem:scalerel} there must be $m'_1$ such that $(m_1, m'_1) \in \overline{\fr(F)}$ and $m' = a\cdot m'_1$. By the induction hypothesis, $m'_1 \vDash \varphi' \osep F$, so therefore $m' \vDash \wg{\varphi'}a \osep F$.

\item $\varphi = \epsilon:P$. We know that $|m| = \one$ and every $\sigma\in\supp(m)$ has the form $\inj_\epsilon(s,h)$ such that $(s,h)\in P$.
Since $(m,m') \in \overline{\fr(F)}$, we know by \Cref{lem:relsize} that $|m'| = |m| = \one$.
Additionally, for every element in $\supp(m')$, there must be an element in $\supp(m)$ related by $\fr(F)$, so each element of $m'$ has the form $\inj_\epsilon(s, h\uplus h')$ such that $(s,h)\in P$ and $(s,h')\in F$, and so clearly $(s, h\uplus h')\in P\sep F$, and therefore also $m'\vDash (\epsilon:P)\osep F$.
\end{itemize}
\end{proof}

\osepback*
\begin{proof}
By induction on the structure of $\varphi$.
\begin{itemize}
\item $\varphi = \top$. Suppose $m\vDash \top\osep F$. Let $m_1 = m'_1 = \zero$ and $m_2 = m$. Clearly $(\zero,\zero)\in \overline{\fr(F)}$ and $m'_1 + m_2 = \zero + m = m$. Now, taking any $m_2'$, it is obvious that $\zero+m_2'\vDash\top$.


\item $\varphi = \varphi_1\vee\varphi_2$. We know $m\vDash\varphi_1\osep F$ or $m\vDash\varphi_2\osep F$. Without loss of generality, suppose that $m\vDash\varphi_1$. By the induction hypothesis, there are $m_1$, $m'_1$, and $m_2$ such that $(m_1, m'_1)\in\overline{\fr(F)}$ and $m = m_1' + m_2$ and $m_1 + m_2'\vDash \varphi_1$ for any $m_2'$ such that $|m_2'| \le |m_2|$.
Now, take any such $m_2'$, we know that $m_1 + m_2'\vDash \varphi_1$. We can weaken this to conclude that $m_1 + m_2'\vDash \varphi_1\vee\varphi_2$

%

\item $\varphi = \varphi_1\oplus\varphi_2$.
We know that there are $m_1$ and $m_2$ such that $m_1\vDash\varphi_1\osep F$ and $m_2\vDash\varphi_2\osep F$ and $m = m_1 + m_2$.
By the induction hypotheses, we get that there are $u_1, u_1', u_2, v_1, v_1'$, and $v_2$ such that $m_1 = u_1'+u_2$ and $m_2' = v_1'+v_2$ and $(u_1,u_1')\in\overline{\fr(F)}$ and $(v_1,v_1')\in\overline{\fr(F)}$ and $u_1+u_2'\vDash \varphi_1$ and $v_1+v_2'\vDash\varphi_2$ whenever $|u_2'|\le|u_2|$ and $|v_2'|\le|v_2|$.

Now, let $m_1 = u_1 +v_1$ and $m_1' = u_1' + v_1'$ and $m_2 = u_2 +v_2$. By \Cref{lem:plusrel}, we get that $(m_1, m_1')\in\overline{\fr(F)}$. We also have that:
\[
m_1' + m_2 = (u_1' +v_1') + (u_2 +v_2)
= (u_1'+u_2) + (v_1'+v_2)
= m_1 + m_2
= m
\]
Now, take any $m_2'$ such that $|m_2'| \le |m_2| =  |u_2| + |v_2|$. By \Cref{lem:splitting} we know that there are $u_2'$ and $v_2'$ such that $|u_2'| \le |u_2|$ and $|v_2'| \le |v_2|$ and $m_2' = u_2'+v_2'$. This means that $u_1 + u_2'\vDash\varphi_1$ and $v_1+v_2'\vDash\varphi_2$. We also have:
\[
(u_1+u_2') + (v_1+v_2') = (u_1 + v_1) + (u_2' + v_2') = m_1 + m_2'
\]
So, $m_1+m_2'\vDash\varphi_1\oplus\varphi_2$.

\item $\varphi = \wg{\varphi'}a$. We know that there is $m'$ such that $m' \vDash \varphi'$ and $m = a\cdot m'$. By the induction hypothesis, we get that there are $u_1$, $u'_1$, and $u_2$ such that $(u_1, u'_1)\in\overline{\fr(F)}$, $m' = u'_1 + u_2$, and $u_1 + u_2'\vDash \varphi'$ whenever $|u_2'| \le |u_2|$.

Now, let $m_1 = a \cdot u_1$, $m_1' = a\cdot u'_1$, and $m_2 = a\cdot u_2$. By \Cref{lem:scalerel}, we get that $(m_1, m_1') \in \overline{\fr(F)}$. We also have that:
\[
m'_1 + m_2 = a\cdot u_1' + a\cdot u_2 = a\cdot (u'_1 + u_2) = a\cdot m' = m
\]
Now, take any $m'_2$ such that $|m'_2| \le |m_2| = a\cdot |u_2|$.
If $m'_2 = \zero$, then $m_1 + m'_2 = m_1 = a\cdot u_1$, so $m_1 + m'_2 \vDash (\varphi')_a$. If not, then by \Cref{lem:normalize} there is $m''_2$ such that $|m''_2| =\one$ and $m'_2 = |m'_2|\cdot m_2''$.
By the definition of $\le$, there must be $a'$ such that $|m'_2| + a' = a\cdot |u_2|$ and by \Cref{def:normalize}, there must be a $b$ such that $|m'_2| = (|m'_2|+a')\cdot b = a\cdot |u_2|\cdot b$.
Now, let $u'_2 = |u_2| \cdot b \cdot m''_2$, so $|u'_2| = |u_2|\cdot b \cdot \one \le |u_2|$, and therefore $u_1 + u'_2 \vDash \varphi'$ and also $a\cdot(u_1 + u'_2) \vDash \wg{\varphi'}a$. Finally, we have that:
\[
a\cdot (u_1 + u'_2)
= a\cdot u_1 + a\cdot |u_2|\cdot b\cdot m''_2
= m_1 + |m'_2| \cdot m''_2
= m_1 + m'_2
\]
So, we get that that $m_1 + m'_2\vDash\wg{\varphi'}a$.

\item $\varphi = \epsilon:P$. We know that $m\vDash \epsilon:P\sep F$, so $|m| = \one$ and every element of $\supp(m)$ has the form $\inj_\epsilon(s,h\uplus h')$ such that $(s,h)\in P$ and $(s,h')\in F$. Now, let $m_2 = \zero$, $m'_1 = m$, so clearly $m = m'_1+m_2$.
To define $m_1$, first we fix a relation $S \subseteq \fr(F)$ such that for any $h''$ such that $(s, h'')\in P\sep F$, there is a unique $h$ and $h'$ such that $h'' = h\uplus h'$ and $(s,h)\in P$ and $(\inj_\epsilon(s,h), \inj_\epsilon(s, h\uplus h'))\in S$. Now, $m_1$ is defined as follows:
\[
m_1(\sigma) = \smashoperator{\sum_{\tau\mid (\sigma,\tau) \in S}} m(\tau)
\]
Now, we must show that $(m_1, m)\in\overline{\fr(F)}$. To do so, first let:
\[
m'(\sigma,\tau) = \left\{\begin{array}{ll}
m(\tau) & \text{if}~(\sigma, \tau)\in S \\
\zero & \text{otherwise}
\end{array}\right.
\]
Now, we have:
\begin{align*}
\lambda\sigma.\smashoperator{\sum_{\tau\in\supp(m)}} m'(\sigma,\tau)
&= \lambda\sigma.\smashoperator{\sum_{\tau\mid (\sigma,\tau)\in S}} m(\tau)
= m_1
\\
\lambda\tau.\smashoperator{\sum_{\sigma\in\supp(m_1)}} m'(\sigma, \tau)
&= \lambda\tau.\smashoperator{\sum_{\sigma\in\supp(m_1)}} \left\{\begin{array}{ll}
m(\tau) & \text{if}~(\sigma, \tau)\in S \\
\zero & \text{otherwise}
\end{array}\right. \\
\intertext{By the definition of $S$, there is exactly one $\sigma$ for each $\tau$ such that $(\sigma, \tau)\in S$, so:}
&= \lambda\tau.m(\tau) = m
\end{align*}
So, $(m_1, m)\in\overline{\fr(F)}$, and this also means that $|m_1| = |m| = \one$. Also, by construction, each element of $\supp(m_1)$ has the form $\inj_\epsilon(s,h)$ where $(s,h)\in P$, so $m_1\vDash \epsilon:P$. Since $m_2=\zero$, then it follows that $m_1+m_2'\vDash \epsilon:P$ for any $m_2'$ such that $|m_2'| \le |m_2| = \zero$.
\end{itemize}

\end{proof}

Now, for the proof of the frame property, it will be necessary to relate states based not only on whether they can be obtained by augmenting the heap, but we must also ensure that after augmenting the heap, it is possible to observe the same allocation behavior. To that end, we introduce a slightly modified $\fr'$ relation.

\[
\begin{array}{l}
\fr'(F, X, \af, \af') = \phantom{x}
\\\quad
\left\{
(\inj_\epsilon(s[X\mapsto n], h), \inj_\epsilon(s, h\uplus h''))
\middle|
\begin{array}{ll}
 (s, h'') \vDash F \\
 \forall (s_0,h_0). ~s_0(x) = n \Rightarrow \\
 \quad\af'(s_0, h_0) = \af(s_0[X\mapsto s(X)], h_0\uplus h'')
 \end{array}
 \right\} \cup \phantom{x} \\
 \quad\{ (\und, \und) \}
\end{array}
\]

\begin{lemma}\label{lem:osepforward2}
If $m\vDash\varphi$, $X\notin\mathsf{fv}(\varphi)$ and $(m, m') \in \overline{\fr'(F,X,\af,\af')}$, then $m'\vDash \varphi \osep F$
\end{lemma}
\begin{proof} Let $R = \fr'(F,X,\af,\af')$. The proof is by induction on the structure of $\varphi$.
\begin{itemize}
\item $\varphi = \top$. Since $\varphi\osep F = \top$, then clearly $m'\vDash\varphi\osep F$
\item $\varphi = \varphi_1 \vee \varphi_2$. We know $m\vDash\varphi_1$ or $m\vDash\varphi_2$. Without loss of generality, suppose that $m\vDash\varphi_1$. By the induction hypothesis, we know that $m'\vDash \varphi_1\osep F$. We can therefore weaken this to conclude that $m'\vDash(\varphi_1\vee\varphi_2)\osep F$. The case where $m\vDash\varphi_2$ is symmetrical.

\item $\varphi = \varphi_1 \oplus \varphi_2$. We know that $m_1\vDash\varphi_1$ and $m_2\vDash\varphi_2$ for some $m_1$ and $m_2$ such that $m = m_1 + m_2$. 
Now, since $(m_1 + m_2, m') \in \overline{R}$, by \Cref{lem:plusrel} there must be $m'_1$ and $m'_2$ such that $(m_1, m'_1) \in \overline{R}$ and $(m_2, m'_2) \in \overline{R}$ and $m' = m'_1 + m'_2$.
By the induction hypothesis, $m'_1\vDash \varphi_1\osep F$ and $m'_2\vDash \varphi_2\osep F$, so $m'\vDash (\varphi_1 \oplus \varphi_2)\osep F$.

\item $\varphi = \wg{\varphi'}a$. If $a=\zero$, then $m=\zero$ and by \Cref{lem:zerorel} $m'=\zero$, so $m'\vDash \wg{\varphi'}\zero \osep F$. Otherwise, we know that $m_1 \vDash \varphi'$ for some $m_1$ such that $m = a\cdot m_1$. Since $(a\cdot m_1, m')\in\overline R$, by \Cref{lem:scalerel} there must be an $m'_1$ such that $(m_1, m'_1) \in \overline R$ and $m' = a\cdot m'_1$. By the induction hypothesis, $m'_1\vDash \varphi'\osep F$, so $m'\vDash \wg{\varphi'}a\osep F$.

\item $\varphi = \epsilon:P$. We know that $|m| = \one$ and every $\sigma\in\supp(m)$ has the form $\inj_\epsilon(s[X\mapsto n],h)$ such that $(s[X\mapsto n],h)\in P$, and since $X\notin\mathsf{fv}(P)$, then it must also be that $(s, h) \in P$.
Since $(m,m') \in \overline{R}$, we know by \Cref{lem:relsize} that $|m'| = |m| = \one$.
Additionally, for every element in $\supp(m')$, there must be an element in $\supp(m)$ related by $R$, so each element of $m'$ has the form $\inj_\epsilon(s, h\uplus h')$ such that $(s,h')\in F$, and so clearly $(s, h\uplus h')\in P\sep F$, and therefore also $m'\vDash (\epsilon:P)\osep F$.
\end{itemize}
\end{proof}

\begin{lemma}\label{lem:frtofr}
For any $X\notin\mathsf{fv}(\varphi)$, and $\af\in\mathsf{Alloc}$, if $(m_1, m_2) \in \overline{\fr(F)}$ and $m_1+u\vDash\varphi$, then there exists an $m'_1$ and $\af'$ such that $(m'_1, m_2)\in \overline{\fr'(F,X,\af,\af')}$ and $m'_1+u\vDash\varphi$.
\end{lemma}
\begin{proof}
Since $(m_1, m_2) \in \overline{\fr(F)}$, there is some $m$ such that:
\[
m_1 = \lambda\sigma.\smashoperator{\sum_{\tau\in\supp(m_2)}}m(\sigma, \tau)
\quad
\text{and}
\quad
m_2 = \lambda\tau.\smashoperator{\sum_{\sigma\in\supp(m_1)}}m(\sigma, \tau)
\]
Since $m$ must by countably supported, we can enumerate the elements of $\supp(m)$, assigning each a unique natural number $n$. Let $f \colon \supp(m) \to \{ n\in\mathbb N \mid 1 \le n\le |\supp(m)| \}$ be such a bijection, and note that $\mathbb N \subseteq \mathbb Z \subseteq \mathsf{Val}$, so the codomain of $f$ contains valid program values. Now, let $m'$ be defined as follows:
\begin{align*}
m'(\inj_\epsilon(s,h), \tau) &= \left\{
\begin{array}{ll}
m(\inj_\epsilon(s', h), \tau) & \text{if}~ \exists s'. f^{-1}(s(X)) = (\inj_{\epsilon}(s', h), \tau) \land s = s'[X\mapsto s(X)] \\
\zero & \text{otherwise}
\end{array}
\right.\\
m'(\und, \tau) &= m(\und, \tau)
\end{align*}
Intuitively $m'$ is obtained by taking each $(\inj_\epsilon(s, h), \tau)$ in the support of $m$ and updating $s$ so that $X$ has value $f(\inj_\epsilon(s, h), \tau)$. Now, we define $\af'$ as follows:
\[
\af'(s,h) = \left\{
\begin{array}{ll}
\af(s[X\mapsto s'(x)],h\uplus h'') & \text{if} ~ \exists s', h', h''.~f^{-1}(s(X)) = (\inj_\epsilon(s', h'), \inj_\epsilon(s', h'\uplus h'')) \\
\af(s,h) & \text{otherwise}
\end{array}
\right.
\]
First, we argue that $\af' \in \mathsf{Alloc}$. It is obvious that $\af'(s,h)$ cannot return anything in the domain of $h$, since its definition relies on $\af$, run on an even larger heap. In the definition of $\af'$, there are two cases. In the first case, $|\af'(s,h)| = |\af(s[X\mapsto s'(x)],h\uplus h'')| = \one$ and in the second case, $|\af'(s,h)| = |\af(s,h)| = \one$. 

We now show that $\supp(m') \subseteq \fr'(F, X, \af,\af')$. Take any $(\sigma, \tau) \in \supp(m')$. If $\sigma = \und$, then it must be the case that $\tau = \und$ too, since $m'(\und, \tau) = m(\und, \tau)$ and $\supp(m)\subseteq \fr(F)$, which only relates $\und$ to itself, therefore $(\sigma, \tau) = (\und, \und) \in \fr'(F,X,\af,\af')$.

If instead $\sigma = \inj_\epsilon(s,h)$, then there must be some $s'$ such that $f^{-1}(s(X)) = (\inj_\epsilon(s', h), \tau)$ and $s = s'[X \mapsto s(X)]$. By the definition of $f$, that means that $(\inj_\epsilon(s', h), \tau) \in \supp(m)$, and so $\tau = \inj_\epsilon(s', h \uplus h'')$ for some $h''$ such that $(s', h'') \in F$. In addition, by the definition of $\af'$, for any $s_0$ and $h_0$ such that $s_0(X) = s(X)$, $\af'(s_0, h_0) = \af(s_0[X\mapsto s'(X)], h_0 \uplus h'')$, therefore:
\[
(\sigma, \tau) = (\inj_\epsilon(s'[X\mapsto s(X)],h), \inj_\epsilon(s', h \uplus h'')) \in \fr'(F,X,\af,\af')
\]
Now, let $m_1' = \lambda\sigma.\sum_{\tau\in\supp(m_2)} m'(\sigma, \tau)$ and observe that for all $\tau$:
\begin{align*}
\smashoperator{\sum_{\sigma\in\supp(m_1')}} m'(\sigma, \tau)
&= m'(\und, \tau) + \smashoperator{\sum_{\inj_\epsilon(s, h)\in\supp(m_1')}} m'(\inj_\epsilon(s, h), \tau) \\
&= m(\und, \tau) + \smashoperator[l]{\sum_{\inj_\epsilon(s, h)\in\supp(m_1')}} \smashoperator[r]{\sum_{v\in\mathsf{Val}}} m(\inj_\epsilon(s[X\mapsto v], h), \tau) \\
&= \smashoperator{\sum_{\sigma \in\supp(m_1)}} m(\sigma, \tau)
= m_2(\tau)
\end{align*}
So, we conclude that $(m'_1, m_2) \in \overline{\fr'(F,X,\af,\af')}$. Finally, for any $\inj_\epsilon(s,h)$:
\begin{align*}
m_1(\inj_\epsilon(s,h))
&= \smashoperator{\sum_{\tau\in\supp(m_2)}} m(\inj_\epsilon(s,h), \tau) \\
&= \smashoperator{\sum_{\tau\in\supp(m_2)}} m'(\inj_\epsilon(s[ X \mapsto f(\inj_\epsilon(s,h))], h), \tau)\\
&= m_1'(\inj_\epsilon(s[ X \mapsto f(\inj_\epsilon(s,h))], h))
\end{align*}
So $m_1$ and $m_1'$ differ only in the values of $X$, which does not appear in $\varphi$ by assumption, so $m'_1\vDash\varphi$.

\end{proof}

\subsection{Replacement of Unsafe States}
\label{app:replacement}


First, we provide the definitions of the $\repl$ relation and the $\prune$ operation. $\repl$ relates $\und$ to any other state and $\lightning$, indicating that after framing, an $\und$ state can become any other state, or can diverge ($\lightning$). All $\ok$ and $\er$ states are only related to themselves. The $\prune\colon \mathcal W(\st \cup \{\lightning\}) \to \mathcal W(\st)$ function removes $\lightning$ from some program configuration $m$.
\[
\repl = \{ (\und, \sigma) \mid \sigma\in \st \cup \{\lightning\} \} \cup \{ (\sigma,\sigma) \mid \sigma\in \st \}
\qquad\quad
\mathsf{prune}(m)(\sigma) = \left\{\begin{array}{ll}
m(\sigma) & \text{if}~ \sigma \neq \lightning \\
\zero & \text{if}~ \sigma = \lightning
\end{array}\right.
\]
Now, by lifting the $\repl$ relation, we can prove that replacing undefined states in some program configuration does not affect the validity of outcome assertions. Intuitively, this is true because $\und$ can only be satisfied by $\top$, which is also satisfied by anything else.

\replcor*
\begin{proof}
By induction on the assertion $\varphi$.
\begin{itemize}
\item $\varphi=\top$. $\prune(m')\vDash\varphi$ since everything satisfies $\top$.
\item $\varphi = \varphi_1\vee\varphi_2$. Without loss of generality, suppose $m\vDash\varphi_1$. By the induction hypothesis, $\prune(m')\vDash\varphi_1$. We can weaken this to conclude that $\prune(m')\vDash\varphi_1\vee\varphi_2$.

\item $\varphi = \varphi_1\oplus\varphi_2$. We know that $m_1\vDash \varphi_1$ and $m_2\vDash\varphi_2$ for some $m_1$ and $m_2$ such that $m_1 +m_2 = m$.
Since $(m_1 +m_2, m')\in \overline\repl$, there must be some $m'_1$ and $m'_2$ such that $(m_1,m'_1)\in\overline\repl$ and $(m_2,m_2')\in\overline\repl$ and $m' = m_1' +m'_2$. By the induction hypothesis, $\prune(m'_1)\vDash\varphi_1$ and $\prune(m'_2)\vDash\varphi_2$. It is easy to see that $\prune(m'_1) + \prune(m'_2) = \prune(m'_1 + m'_2) = \prune(m')$, therefore $\prune(m')\vDash \varphi_1 \oplus_a\varphi_2$.

\item $\varphi = \wg{\varphi'}a$. If $a=\zero$, then $m=\zero$, and by \Cref{lem:zerorel} $m'=\zero$ too. So $\prune(m') = \zero$ and $\prune(m')\vDash\wg{\varphi'}\zero$.
Otherwise, we know that there is an $m_1$ such that $m_1\vDash\varphi'$ and $m = a\cdot m_1$. Since $(a\cdot m_1, m') \in \overline\repl$, by \Cref{lem:scalerel} we know there must be $m'_1$ such that $m' = a\cdot m'_1$ and $(m_1, m'_1)\in \overline\repl$. By the induction hypothesis, $\prune(m'_1)\vDash \varphi'$, so---since $\prune(m') = \prune(a\cdot m'_1) = a\cdot \prune(m'_1)$---we get that $\prune(m')\vDash\wg{\varphi'}a\osep F$.

\item $\varphi = (\epsilon:P)$. By definition, $\und\notin\supp(m)$, so $m' = m$ (since $\repl$ only relates defined states to themselves), which satisfies $\epsilon:P$ by assumption.
\end{itemize}
\end{proof}


\subsection{The Frame Rule}
\label{app:framerule}

\begin{lemma}[Sequencing]\label{lem:sequencing}
For any $f, g\colon \mathcal S\times\mathcal H \to \mathcal W_{\mathcal A}(\st)$ and relation $R \subseteq \st\times \st \cup\{\lightning\}$, if $(m_1, m_2) \in \overline R$ and:
\[
\forall(\inj_\ok(s_1,h_1), \inj_\ok(s_2,h_2))\in R.\quad \exists m. \quad g(s_2, h_2) = \prune(m)\quad \text{and}\quad (f(s_1,h_1), m) \in \overline R
\]
Then, there exists $m_2'$ such that $\bind(\prune(m_2), g) = \prune(m_2')$ and $(\bind(m_1, f), m_2') \in \overline R$.
\end{lemma}
\begin{proof}
First, let $f',g' \colon \st \to \mathcal W_{\mathcal A} (\st)$ be defined as follows:
\[
f'(\sigma) = \left\{\begin{array}{ll}
f(s, h) & \text{if} ~ \sigma = \inj_\ok(s,h) \\
\unit_\mathcal{W}(\sigma) & \text{otherwise}
\end{array}\right.
\qquad
g'(\sigma) = \left\{\begin{array}{ll}
g(s, h) & \text{if} ~ \sigma = \inj_\ok(s,h) \\
\unit_\mathcal{W}(\sigma) & \text{otherwise}
\end{array}\right.
\]
Note that $\bind(m,f) = \bind_\mathcal W(m, f')$ and $\bind(m, g) = \bind_{\mathcal W}(m, g')$ for all $m$.
Now, we argue that if $(\sigma, \tau)\in R$, then there exists $m_{\sigma,\tau}$ such that $g'(\tau) = \prune(m_{\sigma,\tau})$ and $(f'(\sigma), m_{\sigma,\tau}) \in \overline R$.
We do so by case analysis.
\begin{itemize}
\item $\sigma = \inj_\ok(s,h)$ and $\tau=\inj_\ok(s',h')$.  By definition, $g'(\tau) = g(s',h')$ and $f'(\sigma) = f(s,h)$, so the claim holds by assumption.
\item $\sigma = \und$. In this case, $f'(\sigma) = \unit_\mathcal{W}(\und)$, which means that $f'(\sigma)$ is related to all configurations of size $\one$ according to $\overline R$. Now, since $\sup(A)=\one$, it must be that $|g'(\tau)| \le \one$ and so there is some $u$ such that $|g'(\tau)| + u = \one$. Now, let $m_{\sigma,\tau} = g'(\tau)+ u\cdot\unit_{\mathcal W}(\lightning)$, so clearly $g'(\tau) = \prune(m_{\sigma,\tau})$ and $(f'(\sigma), m_{\sigma,\tau})\in \overline R$.
%
\item In this final case, $\tau$ cannot have the form $\inj_\ok$, since only $\ok$ and $\und$ states are related to $\ok$ states according to $R$, and we have already handled both of those cases. This means that $\sigma$ must also not be an $\ok$ state, since $\ok$ states are only related to other $\ok$ states. Therefore, $f'(\sigma) = \unit_{\mathcal W}(\sigma)$ and $f'(\tau) = \unit_{\mathcal W}(\tau)$. By \Cref{lem:unitrel}, we conclude that $(\unit_{\mathcal W}(\sigma), \unit_{\mathcal W}(\tau)) \in \overline R$.
\end{itemize}
Given this, we know that for each $\sigma$ and $\tau$, there must be some $m_{\sigma,\tau}$ and $m'_{\sigma,\tau}$ such that $g'(\tau) = \prune(m_{\sigma,\tau})$ and:
\[
f'(\sigma) = \lambda \sigma'.\smashoperator{\sum_{\tau'\in\supp(m_{\sigma,\tau})}} m'_{\sigma,\tau}(\sigma', \tau')
\qquad
m_{\sigma,\tau} = \lambda \tau'.\smashoperator{\sum_{\sigma'\in\supp(f'(\sigma))}} m'_{\sigma,\tau}(\sigma', \tau')
\]
Since $(m_1, m_2)\in\overline R$, we know that there is an $m$ such that $m_1 = \lambda\sigma.{\sum_{\tau\in\supp(m_2)}} m(\sigma, \tau)$ and $m_2 = \lambda\tau.{\sum_{\sigma\in\supp(m_1)}} m(\sigma, \tau)$
Now let:
\[
m'(\sigma', \tau') = \smashoperator[l]{\sum_{\sigma\in\supp(m_1)}} \smashoperator[r]{\sum_{\tau\in\supp(m_2)}} m(\sigma,\tau)\cdot m'_{\sigma,\tau}(\sigma', \tau')
\qquad\qquad
m'_2(\tau') = \smashoperator{\sum_{\sigma' \in \supp(\bind(m_1, f))}} m'(\sigma', \tau')
\]
Now, we show that:
\begin{align*}
\lambda\sigma'.\smashoperator{\sum_{\tau'\in\supp(m'_2)}} m'(\sigma', \tau')
&= \lambda\sigma'.\smashoperator[l]{\sum_{\tau'\in\supp(m'_2)}} \sum_{\sigma\in\supp(m_1)}\smashoperator[r]{\sum_{\tau\in\supp(m_2)}} m(\sigma,\tau)\cdot m'_{\sigma,\tau}(\sigma', \tau')\\
&= \lambda\sigma'. \smashoperator[l]{\sum_{\sigma\in\supp(m_1)}} \smashoperator[r]{\sum_{\tau\in\supp(m_2)}} m(\sigma,\tau)\cdot \smashoperator{\sum_{\tau'\in\supp(m'_2)}}m'_{\sigma,\tau}(\sigma', \tau')\\
\intertext{
Now, $\supp(m'_2)$ is all those $\tau'$ such that $m'(\sigma', \tau')\neq\zero$ for some $\sigma'\in\supp(\bind(m_1, f))$, which also means that $m'_{\sigma,\tau}(\sigma',\tau')\neq\zero$.
Since the outer sum is over $\sigma\in\supp(m_1)$,  the terms we are summing over ($m'_{\sigma,\tau}(\sigma',\tau')$) will be $\zero$ unless $\sigma' \in \supp(f'(\sigma))$.
So, the last sum is equivalent to summing over $\tau'$ such that there is a $\sigma' \in \supp(f'(\sigma))$ such that $m'_{\sigma,\tau}(\sigma',\tau')\neq\zero$, which is exactly $\supp(m_{\sigma,\tau})$.
}
&= \lambda\sigma'. \smashoperator[l]{\sum_{\sigma\in\supp(m_1)}}\smashoperator[r]{\sum_{\tau\in\supp(m_2)}} m(\sigma,\tau)\cdot (\smashoperator{\sum_{\tau'\in\supp(m_{\sigma,\tau})}}m'_{\sigma,\tau}(\sigma', \tau'))\\
&= \lambda\sigma'. \smashoperator{\sum_{\sigma\in\supp(m_1)}}m_1(\sigma)\cdot f'(\sigma)(\sigma')
= \bind(m_1, f)
\end{align*}
So $(\bind(m_1,f), m'_2) \in \overline R$. It now just remains to prove that $\bind(m_2,g) = \prune(m'_2)$:
\begin{align*}
\prune(m'_2) &= \prune(\lambda \tau'.\smashoperator{\sum_{\sigma' \in \supp(\bind(m_1, f))}} m'(\sigma', \tau')) \\
&= \prune( \lambda \tau'. \smashoperator[l]{\sum_{\sigma' \in \supp(\bind(m_1, f))}}
\sum_{\tau\in\supp(m_2)} \smashoperator[r]{\sum_{\sigma\in\supp(m_1)}} m(\sigma,\tau)\cdot m'_{\sigma,\tau}(\sigma', \tau'))\\
&= \prune( \lambda \tau'. 
\smashoperator[l]{\sum_{\tau\in\supp(m_2)}} \smashoperator[r]{\sum_{\sigma\in\supp(m_1)}} m(\sigma,\tau)\cdot \;\; \smashoperator{\sum_{\sigma' \in \supp(\bind(m_1, f))}}m'_{\sigma,\tau}(\sigma', \tau')) \\
&= \prune( \lambda \tau'. 
\smashoperator[l]{\sum_{\tau\in\supp(m_2)}} \smashoperator[r]{\sum_{\sigma\in\supp(m_1)}} m(\sigma,\tau)\cdot \;\; \smashoperator{\sum_{\sigma' \in \supp(f'(\sigma))}}m'_{\sigma,\tau}(\sigma', \tau')) \\
&= \prune( \lambda \tau'. 
\smashoperator[l]{\sum_{\tau\in\supp(m_2)}} \smashoperator[r]{\sum_{\sigma\in\supp(m_1)}} m(\sigma,\tau)\cdot m_{\sigma,\tau}(\tau')) \\
&= \lambda \tau'. 
\smashoperator[l]{\sum_{\tau\in\supp(m_2)}} \smashoperator[r]{\sum_{\sigma\in\supp(m_1)}} m(\sigma,\tau)\cdot \prune(m_{\sigma,\tau}(\tau')) \\
&= \lambda \tau'. 
\smashoperator[l]{\sum_{\tau\in\supp(m_2)}} \smashoperator[r]{\sum_{\sigma\in\supp(m_1)}} m(\sigma,\tau)\cdot g'(\tau)(\tau') \\
&= \lambda \tau'. 
\smashoperator{\sum_{\tau\in\supp(m_2)}} m_2(\tau) \cdot g'(\tau)(\tau') 
= \bind(m_2, g)
\end{align*}

\end{proof}


\begin{lemma}\label{lem:update}
Let $R = \mathsf{Rep} \circ \fr'(F,X,\af,\af')$.
If $(\inj_\ok(s[X\mapsto n], h), \inj_\ok(s, h\uplus h')) \in R$ and $(\inj_\ok(s_1,h_1), \inj_\ok(s_2,h_2))\in R$ whenever $h(\de{e}(s)) \in\mathsf{Val}$, then
\[
(\mathsf{update}(s[X\mapsto n],h,\de{e}(s[X\mapsto n]),s_1, h_1), \mathsf{update}(s, h\uplus h', \de{e}(s), s_2, h_2)) \in \overline R
\]
\end{lemma}
\begin{proof}
Let $\ell = \de{e}(s) = \de{e}(s[X\mapsto n])$ (since $X$ cannot affect the program expression $e$). We complete the proof by case analysis:
\begin{itemize}
\item $h(\ell) \in \mathsf{Val}$. It must also be the case that $(h \uplus h')(\ell) \in \mathsf{Val}$, since $h'$ is disjoint from $h$.
So, it just remains to show that $(\unit(s_1, h_1), \unit(s_2, h_2)) \in \overline R$, which follows from \Cref{lem:unitrel}.

\item $h(\ell) = \bot$. By a similar argument to the previous case, it must be that $(h\uplus h')(\ell) = \bot$ too. So, we just need to show that $(\mathsf{error}(s[X\mapsto n],h), \mathsf{error}(s, h\uplus h')) \in \overline R$. We know that $(\inj_\er(s[X\mapsto n],h), \inj_\er(s[X\mapsto n], h\uplus h'))\in R$ since $R$ treats $\ok$ and $\er$ the same, so the claim follows by \Cref{lem:unitrel}.

\item $\ell\notin\mathsf{dom}(h)$. So, $\mathsf{update}(s[X\mapsto n],h,\de{e}(s),s_1,h_1) = \unit_{\mathcal W}(\und)$, and since $R$ relates $\und$ to all states, $(\mathsf{update}(s,h,\de{e}(s),s_1,h_1), m) \in\overline R$ for all $m$, which means that $\mathsf{update}(s, h\uplus h', \de{e}(s), s_2, h_2)$ is related trivially.
\end{itemize}
\end{proof}

\begin{restatable}[The Frame Property]{lemma}{frameprop}\label{lem:frameprop}
Let $R = \mathsf{Rep} \circ \fr'(F,X,\af,\af')$, so $R\subseteq \st\times(\st\cup\{\lightning\})$.
For any program $C$ such that $\mathsf{mod}(C) \cap \mathsf{fv}(F) = \emptyset$:
\[
\forall (\inj_\ok(s, h), \inj_\ok(s', h')) \in R. \quad
\exists m.\quad
\de{C}_\af(s',h') = \prune(m)
\quad\text{and}\quad
(\de{C}_{\af'}(s,h), m) \in \overline R
\]
\end{restatable}
\begin{proof} By induction on the structure of the program $C$.
\begin{itemize}
\item $C=\skp$. In this case, $\de{C}_{\af'}(s,h) = \unit(s,h)$ and $\de{C}_{\af}(s',h') = \unit(s',h')$, so the claim follows from \Cref{lem:unitrel}.

\item $C= C_1\fatsemi C_2$. By definition, we know that $\de{C}_{\af'}(s,h) = \bind(\de{C_1}_{\af'}(s,h), \de{C_2}_{\af'})$ and $\de{C}_\af(s',h') = \bind(\de{C_1}_\af(s',h'), \de{C_2}_\af)$.
By the induction hypothesis, we know there is some $m$ such that $\de{C_1}_\af(s',h') = \prune(m)$ and $(\de{C_1}_{\af'}(s,h), m) \in \overline R$.
Therefore by the induction hypothesis and \Cref{lem:sequencing}, we can conclude that there is some $m'$ such that $\bind(\de{C_1}_\af(s',h'), \de{C_2}_{\af}) = \prune(m')$ and $(\bind(\de{C_1}_{\af'}(s,h), \de{C_1}_{\af'}), m') \in \overline R$, which completes the proof.


\item $C = C_1 + C_2$. We know that $\de{C_1 + C_2}_{\af'}(s,h) = \de{C_1}_{\af'}(s,h) + \de{C_2}_{\af'}(s,h)$ and that $\de{C_1 + C_2}_{\af}(s',h') = \de{C_1}_{\af}(s',h') + \de{C_2}_{\af}(s',h')$.
By the induction hypothesis, for each $i\in\{1,2\}$ we get that there is some $m_i$ such that $\de{C_i}_{\af}(s',h') = \prune(m_i)$ and $(\de{C_i}_{\af'}(s,h), m_i) \in \overline R$. Using \Cref{lem:plusrel} we can conclude that $(\de{C_1+C_2}_{\af'}(s,h), m_1 + m_2) \in \overline R$. It now only remains to show that:
\begin{align*}
\prune(m_1 + m_2)
&= \prune(m_1) + \prune(m_2) \\
&= \de{C_1}_\af(s',h') + \de{C_2}_\af(s',h') \\
&= \de{C_1+C_2}_\af(s',h')
\end{align*}

\end{itemize}
We now move on to the cases involving state. 
Since $(\inj_\ok(s,h), \inj_\ok(s', h')) \in R$, then there must be some $n$ and $h''$ such that $s = s'[X \mapsto n]$, $h' = h \uplus h''$, $(s', h'') \vDash F$, and for any $s_0$ and $h_0$, if $s_0(X) = n$, then $\af'(s_0, h_0) = \af(s_0[X\mapsto s'(X)], h_0 \uplus h'')$.
Additionally, many of the cases have a single outcome, so by \Cref{lem:unitrel} it suffices to show that $(\sigma, \tau)\in R$ where $\de{C}_{\af'}(s,h) = \unit_{\mathcal{W}}(\sigma)$ and $\de{C}_{\af}(s',h') = \unit_{\mathcal{W}}(\tau)$ in those cases. 
\begin{itemize}


\item $C = \assume{e}$. Since $s$ and $s'$ only differ in the value of $X$, which cannot affect the value of $e$, then $\de{e}(s') = \de{e}(s)$. This means that both programs weight the computation by the same amount, let this weight be $a=\de{e}(s') = \de{e}(s)$. So, we get that $\de{\assume e}_{\af'}(s,h) = a\cdot \unit(s,h)$ and $\de{\assume e}_{\af}(s', h') = a\cdot\unit(s', h')$. Since $(\inj_\ok(s,h), \inj_\ok(s',h')) \in R$, then by \Cref{lem:unitrel} $(\unit(s,h), \unit(s',h')) \in \overline R$, so by \Cref{lem:scalerel} $(a\cdot \unit(s,h), a\cdot\unit(s', h')) \in \overline R$.

\item $C = \whl e{C'}$. First, we will show that there exists $m$ such that $F_{\langle C',e, \af\rangle}^n(\bot)(s',h') = \prune(m)$ and $(F_{\langle C',e,\af'\rangle}^n(\bot)(s,h), m) \in\overline R$ for any $(\inj_\ok(s,h), \inj_\ok(s',h'))\in R$. The proof is by induction on $n$. If $n=0$, then the claim holds using \Cref{lem:zerorel}:
\[
F^0_{\langle C',e,\af\rangle}(\bot)(s',h') =\bot(s',h') = \zero = \bot(s,h) = F^0_{\langle C',e, \af'\rangle}(\bot)(s,h) 
\]
Now suppose the claim holds for $n$, we will show that it also holds for $n+1$:
\begin{align*}
F_{\langle C',e, \af\rangle}^{n+1}(\bot)(s',h') &= \left\{\begin{array}{ll}
\bind(\de{C'}_\af(s',h'), F_{\langle C',e, \af\rangle}^n(\bot)) & \text{if}~ \de{e}(s') = \one \\
\unit(s',h') & \text{if}~\de{e}(s') = \zero
\end{array}\right.\\
\intertext{and}
F_{\langle C',e,\af'\rangle}^{n+1}(\bot)(s,h) &= \left\{\begin{array}{ll}
\bind(\de{C'}_{\af'}(s,h), F_{\langle C',e,\af'\rangle}^n(\bot)) & \text{if}~ \de{e}(s) = \one \\
\unit(s,h) & \text{if}~\de{e}(s) = \zero
\end{array}\right.
\end{align*}
Note that as we showed in the previous case for assume, $\de{e}(s) = \de{e}(s')$, so both executions will take the same path.
If $\de{e}(s') = \de{e}(s) = \one$, then the claim holds by \Cref{lem:sequencing} and the induction hypothesis. In the second case, it holds by \Cref{lem:unitrel}.

Now, by the definition of $\prune$, this also means that for any $n$, there exists $a_n$ such that:
\[
(F^n_{\langle C',e, \af'\rangle}(\bot)(s,h), F^n_{\langle C',e, \af\rangle}(\bot)(s',h') + a_n\cdot\unit(\lightning)) \in \overline R 
\]
Recall that $\lightning$ represents the nonterminating traces, and as we continue to unroll the loop, the weight of nontermination can only increase, so the $a_n$ must increase monotonically and therefore $F^n_{\langle C',e,\af\rangle}(\bot)(s',h') + a_n\cdot\unit(\bot)$ is a chain, so
by \Cref{lem:suprel} we know that:
\[
(\sup_{n\in\mathbb{N}}F^n_{\langle C',e\rangle}(\bot)(s,h), \sup_{n\in\mathbb N}F^n_{\langle C',e\rangle}(\bot)(s',h') + a_n\cdot\unit(\lightning)) \in \overline R
\]
And we can therefore conclude that there exists $m$ such that $\de{\whl e{C'}}(s',h') = \prune(m)$ and $(\de{\whl e{C'}}(s,h), m)\in \overline R$.

\item $C = (x\coloneqq e)$. We know the following:
\begin{align*}
\de{C}_{\af'}(s'[X\mapsto n],h)
&= \unit(s'[X\mapsto n, x \mapsto \de{e}(s'[X\mapsto n])], h) \\
&= \unit(s'[x \mapsto \de{e}(s')][X\mapsto n], h) \\
\de{C}_{\af}(s',h\uplus h'') &= \unit(s'[x \mapsto \de{e}(s')], h\uplus h'')
\end{align*}
Since $x\in\mathsf{mod}(C)$, then $x\notin\mathsf{fv}(F)$, so updating $x$ in $s$ will not affect the truth of $F$, therefore $(s'[x\mapsto\de{e}(s')], h\uplus h'')\vDash F$. In addition, since we did not modify the value of $X$, it is still true that $\af'(s_0, h_0) = \af(s_0[X\mapsto s'(X)], h_0 \uplus h'')$ for any $s_0, h_0$ with $s_0(X) = n$. So, we know that:
\[
(\inj_\ok(s'[x\mapsto \de{e}(s')][X\mapsto n], h), \inj_\ok(s[x\mapsto \de{e}(s')], h \uplus h''))\in\fr'(F, X, \af, \af')
\]
Putting this together along with the fact that $\repl$ is reflexive finishes the proof.

\item $C = (x \coloneqq\mathsf{alloc}())$. We know the following:
\begin{align*}
\de{C}_{\af'}(s'[X\mapsto n],h) &= \bind_{\mathcal W}(\af'(s'[X\mapsto n],h), \lambda(\ell,v).\unit(s'[X\mapsto n][x\mapsto\ell], h[\ell\mapsto v])) \\
&= \bind_{\mathcal W}(\af(s',h \uplus h''), \lambda(\ell, v).\unit(s'[x\mapsto\ell][X\mapsto n], h[\ell\mapsto v])) \\
\de{C}_{\af}(s',h\uplus h'') &= \bind_{\mathcal W}(\af(s', h\uplus h''), \lambda(\ell, v).\unit(s'[x\mapsto\ell], (h\uplus h'')[\ell\mapsto v])) \\
&= \bind_{\mathcal W}(\af(s', h\uplus h''), \lambda(\ell, v).\unit(s'[x\mapsto\ell], h[\ell\mapsto v] \uplus h''))
\end{align*}
So it is clear that $(\de{C}_{\af'}(s,h), \de{C}_{\af}(s,h\uplus h')) \in \overline R$ since the two weighting functions are identical except that $\de{C}_{\af}(s,h\uplus h'')$ has $h''$ added at every state and just as in the previous case, we did not update $X$, so the relationship between $\af$ and $\af'$ holds as well.

\item $C = \mathsf{free}(e)$. Using \Cref{lem:update}, we only need to show that if $h(\de{e}(s')) \in\mathsf{Val}$, then:
\[
(\inj_\ok(s'[X\mapsto n], h[\de{e}(s)\mapsto \bot]), \inj_\ok(s', (h\uplus h'')[\de{e}(s') \mapsto \bot])) \in \fr'(F,X,\af,\af') \subseteq R
\]
If $h(\de{e}(s')) \in\mathsf{Val}$, then $\de{e}(s') \in \mathsf{dom}(h)$, and since $h''$ is disjoint from $h$, then $\de{e}(s')\notin\mathsf{dom}(h'')$, so
$(h\uplus h'')[\de{e}(s') \mapsto \bot] = h[\de{e}(s') \mapsto \bot] \uplus h''$, and clearly:
\[
(\inj_\ok(s'[X\mapsto n], h[\de{e}(s') \mapsto \bot]), \inj_\ok(s', h[\de{e}(s') \mapsto \bot]) \uplus h'')) \in \fr'(F,X,\af,\af')
\]

\item $C = ([e_1]\leftarrow e_2)$. Using \Cref{lem:update}, we only need to show that if $h(\de{e_1}(s')) \in\mathsf{Val}$, then:
\begin{align*}
&(\inj_\ok(s'[X\mapsto n], h[\de{e_1}(s')\mapsto \de{e_2}(s')]), \inj_\ok(s', (h\uplus h'')[\de{e_1}(s') \mapsto \de{e_2}(s')])) \\
&\quad \in \fr'(F,X,\af,\af') \\
&\quad \subseteq R
\end{align*}
If $h(\de{e_1}(s')) \in\mathsf{Val}$, then $\de{e_1}(s') \in \mathsf{dom}(h)$, and since $h''$ is disjoint from $h$, then $\de{e_1}(s')\notin\mathsf{dom}(h'')$, so
$(h \uplus h'')[\de{e_1}(s') \mapsto \de{e_2}(s')] = h)[\de{e_1}(s') \mapsto \de{e_2}(s')] \uplus h''$. Now, clearly:
\[
(\inj_\ok(s'[X\mapsto n], h[\de{e_1}(s')\mapsto \de{e_2}(s')]), \inj_\ok(s', h[\de{e_1}(s') \mapsto \de{e_2}(s')] \uplus h'')) \in R
\]

\item $C = (x\leftarrow [e])$.
Using \Cref{lem:update}, we only need to show that if $h(\de{e}(s')) \in\mathsf{Val}$, then:
\[
(\inj_\ok(s[X\mapsto n][x\mapsto h(\de{e}(s')), h), \inj_\ok(s'[x \mapsto (h\uplus h'')(\de{e}(s')]), h\uplus h'')) \in R
\]
If $h(\de{e}(s')) \in\mathsf{Val}$, then $\de{e}(s') \in \mathsf{dom}(h)$, and since $h''$ is disjoint from $h$, then $\de{e}(s')\notin\mathsf{dom}(h'')$, so $s'[x \mapsto (h\uplus h'')(\de{e}(s'))] = s'[ x \mapsto h(\de{e}(s'))]$.

So, clearly $(\inj_\ok(s'[ x\mapsto h(\de{e}(s)][X\mapsto n]), h), \inj_\ok(s'[ x\mapsto h(\de{e}(s')) ]), h\uplus h'') \in R$.

\item $C = \mathsf{error}()$. By \Cref{lem:unitrel}, it suffices to show that $(\inj_\er(s'[X\mapsto n],h), \inj_\er(s', h\uplus h''))\in R$, which follows from the assumptions since $R$ treats $\ok$ and $\er$ states in the same way.

\item $C = f(\vec e)$. Let $C'$ be the body of $f$. By the same argument used in the assignment case, $(\inj_\ok(s'[X\mapsto n][\vec x \mapsto \de{\vec e}(s')], h),\inj_\ok(s'[\vec x\mapsto \de{\vec e}(s')]), h \uplus h'')) \in R$. So, the claim follows from the induction hypothesis.
\end{itemize}

\end{proof}

\framethm*
\begin{proof}
Let $R = \repl \circ \fr'(F,X,\af,\af')$.
Suppose $m\vDash \varphi\osep F$, take any $\af\in\mathsf{Alloc}$, and pick some $X\notin\mathsf{fv}(\varphi, \psi, F)$. Then by \Cref{lem:osepback,lem:frtofr} we know that there are $m_1$, $m'_1$, $m_2$, and $\af'$ such that $(m_1, m'_1) \in \overline{\fr'(F,X,\af,\af')}$ and $m = m_1' + m_2$ and $m_1 + m'_2 \vDash \varphi$ for any $m_2'$ such that $|m_2'| \le |m_2|$. So that means that $m_1 + |m_2| \cdot \unit(\und) \vDash\varphi$. We know that:
\[
\dem{C}{m_1 + |m_2| \cdot \unit(\und)}{\af'} = \dem{C}{m_1}{\af'} + |m_2|\cdot\unit(\und)
\]
So, therefore $\dem{C}{m_1}{\af'} + |m_2|\cdot\unit(\und) \vDash\psi$ since $\vDash\triple{\varphi}C{\psi}$.
Now, observe that $(m_1, m_1') \in \overline{\fr'(F,X,\af,\af')} \subseteq \overline R$ since $\repl$ is reflexive.
We also know that $(|m_2|\cdot \unit(\und), m_2) \in \overline R$ since $R$ permits $\und$ states to be remapped to anything. Therefore, using \Cref{lem:plusrel} we get that:
\[
(m_1 + |m_2|\cdot \unit(\und), m) = (m_1 + |m_2|\cdot \unit(\und), m_1'+m_2) \in \overline R
\]
So, using \Cref{lem:sequencing,lem:frameprop}, we know that there exists some $m'$ such that $\dem{C}m\af = \prune(m')$ and:
\[
(\dem{C}{m_1 + |m_2|\cdot\unit(\und)}{\af'}, m') \in \overline R \subseteq \overline\repl \circ \overline{\fr'(F, X, \af, \af')}
\]
Where the last step is by \Cref{lem:relcomp}.
All that remains now is to peel away the layers in the above expression. More concretely, we know that there is some $m''$ such that $(\dem{C}{m_1 + |m_2|\cdot\unit(\und)}{\af'}, m'')\in\overline{\fr'(F,X,\af,\af')}$ and $(m'', m')\in \overline\repl$. By \Cref{lem:osepforward2} $m''\vDash \psi\osep F$ and by \Cref{cor:replacement}, $\dem Cm\af\vDash \psi\osep F$.

\end{proof}

\section{Tri-Abduction}
\label{app:triab}

The full set of inference rules for the tri-abduction proof system is shown in \Cref{fig:triab-pf-full}. The abduction algorithm is given in \Cref{alg:abduce-par}.

\begin{algorithm}[t]
\caption{abduce-par($P$,$Q$)}
\label{alg:abduce-par}
\begin{algorithmic}[1]
  \If{either \textsc{Base-Emp}, \textsc{Base-True-L}, or \textsc{Base-True-R} apply}
    \State\Return anti-frame $M$ as indicated by that rule
  \Else
    \For{Each remaining row in \Cref{fig:triab-pf-full} from top to bottom}
      \State $\text{result} = \emptyset$
      \For{Each inference rule in the row of the form below}
        \[
        \inferrule{\trij{P'}{M'}{Q'} \\ R}{\trij PMQ}
        \]
        \If{The input parameters match $P$ and $Q$ in the inference rule and $R$ is true}
            \State $\text{result} = \text{result} \cup \{ M \mid M' \in \textsf{abduce-par}(P', Q') \}$
        \EndIf
      \EndFor
      \If{$\text{result} \neq\emptyset$}
        \State\Return\text{result}
      \EndIf
    \EndFor
  \State\Return$\emptyset$
  \EndIf
\end{algorithmic}
\end{algorithm}

\begin{figure}
  \footnotesize
  \begin{flushleft}\fbox{\textsf{Base Cases}}\end{flushleft}
  \[
    \inferrule* [rightstyle={\footnotesize \sc},right=Base-Emp]{
      \Pi \wedge \Pi' \not \vdash \fls
    }{
      \trij{\Pi\land\emp}{\Pi \wedge \Pi'\land \emp}{ \Pi'\land\emp}
    }
  \]
  \smallskip
  \[
    \inferrule* [rightstyle={\footnotesize \sc},right=Base-True-L]{
      \Pi \wedge \Pi'\land\Sigma' \not \vdash \fls
    }{
      \trij{\Pi\land\tru}{\Pi \wedge \Pi' \land\Sigma'}{ \Pi'\land\Sigma'}
    }
  \quad
    \inferrule* [rightstyle={\footnotesize \sc},right=Base-True-R]{
      \Pi \wedge \Pi'\wedge\Sigma \not \vdash \fls
    }{
      \trij{\Pi\wedge\Sigma}{\Pi \wedge \Pi'\wedge\Sigma}{\Pi'\wedge\tru}
    }
\]
  \smallskip
\begin{flushleft}\fbox{\textsf{Quantifier Elimination}}\end{flushleft}
    \[
    \inferrule* [rightstyle={\footnotesize \sc},right=Exists]{
      \trij{\Delta}M{\Delta'} \\ \vec X \cap (\mathsf{fv}(\Delta')\setminus \vec Y) = \emptyset \\ \vec Y\cap(\mathsf{fv}(\Delta) \setminus\vec X) = \emptyset}
      {\trij{\exists\vec X.\Delta}{\exists\vec X\vec Y.M}{\exists\vec Y.\Delta'}}
  \]
\begin{flushleft}\fbox{\textsf{Resource Matching}}\end{flushleft}
  \[
    \inferrule*[rightstyle={\footnotesize \sc},right=Ls-Start-L]{
      \trij{\Delta \sep \ls(e_3,e_2)}{M}{\Delta'}
    }{
      \trij{\Delta \sep \ls(e_1,e_2)}{M \sep e_1 \mapsto e_3}{\Delta'\sep e_1 \mapsto e_3}
    }
  \quad
    \inferrule* [rightstyle={\footnotesize \sc},right=Ls-Start-R]{
      \trij{\Delta}{M}{\Delta'\sep \ls(e_3,e_2)}
    }{
      \trij{\Delta \sep e_1 \mapsto e_3}{M \sep e_1 \mapsto e_3}{\Delta'\sep \ls(e_1,e_2)}
    }
  \]
  \smallskip
  \[
    \inferrule* [rightstyle={\footnotesize \sc},right=Match]{
      \trij{\Delta \wedge e_2 = e_3}M{\Delta' \wedge e_2 = e_3}
    }{
      \trij{\Delta \sep e_1 \mapsto e_2}{M \sep e_1 \mapsto e_2}{\Delta'\sep e_1 \mapsto e_3}
    }
  \]
  \smallskip
  \[
    \inferrule* [rightstyle={\footnotesize \sc},right=Ls-End-L]{
      \trij{\Delta \sep \ls(e_3,e_2)}{M}{\Delta'}
    }{
      \trij{\Delta \sep \ls(e_1,e_2)}{M \sep \ls(e_1,e_3)}{\Delta'\sep\ls(e_1,e_3)}
    }
  \quad
    \inferrule* [rightstyle={\footnotesize \sc},right=Ls-End-R]{
      \trij{\Delta}{M}{\ls(e_2,e_3) \sep \Delta'}
    }{
      \trij{\Delta \sep \ls(e_1,e_2)}{M \sep \ls(e_1,e_2)}{\Delta'\sep\ls(e_1,e_3)}
    }
  \]
  \smallskip
\begin{flushleft}\fbox{\textsf{Resource Adding}}\end{flushleft}
  \[
  \inferrule*[rightstyle={\footnotesize \sc},right=Missing-L]{
    \trij{\Delta}M{\Pi'\land(\Sigma'\sep\tru)} \\ \Pi'\land\Sigma' \sep B(e_1,e_2) \not\vdash\fls
  }{
    \trij{\Delta\sep B(e_1,e_2)}{M\sep B(e_1,e_2)}{\Pi'\land(\Sigma'\sep\tru)}
  }
  \]
  \smallskip
  \[
  \inferrule*[rightstyle={\footnotesize \sc},right=Missing-R]{
    \trij{\Pi\land(\Sigma\sep\tru)}M{\Delta'} \\ \Pi\land\Sigma \sep B(e_1,e_2) \not\vdash\fls
  }{
    \trij{\Pi\land(\Sigma\sep\tru)}{M\sep B(e_1,e_2)}{\Delta'\sep B(e_1,e_2)}
  }
  \]
    \smallskip
    \[
    \inferrule* [rightstyle={\footnotesize \sc},right=Emp-Ls-L]{
     \trij{\Delta\land e_1=e_2}M{\Delta' \land e_1= e_2}
    }{
      \trij{\Delta\sep\ls(e_1,e_2)}{M}{\Delta'}
    }
  \quad
    \inferrule* [rightstyle={\footnotesize \sc},right=Emp-Ls-R]{
     \trij{\Delta\land e_1=e_2}M{\Delta' \land e_1= e_2}
    }{
      \trij{\Delta}{M}{\Delta'\sep\ls(e_1,e_2)}
    }
  \]

\caption{Tri-abduction proof rules. Similarly to \citet{biab}, in the above we use $B(e_1,e_2)$ to represent either $\ls(e_1,e_2)$ or $e_1 \mapsto e_2$.
}
\label{fig:triab-pf-full}
\end{figure}

\begin{lemma}\label{lem:trij}
If $\trij PMQ$ is derivable, then $M\vDash P$ and $M\vDash Q$
\end{lemma}

\begin{proof}

The proof is by induction on the derivation of $\trij PMQ$.

\begin{enumerate} 

\pfcase{Base-Emp} 
We need to show that $\Pi\land\Pi'\land\emp\vDash \Pi\land\emp$ and $\Pi\land\Pi'\land\emp\vDash \Pi'\land\emp$, both hold by the semantics of logical conjunction.

\pfcase{Base-True-L}
We need to show that $\Pi\land\Pi'\land\Sigma'\vDash \Pi\land\tru$ and $\Pi\land\Pi'\land\Sigma'\vDash \Pi'\land\Sigma'$, both hold by the semantics of logical conjunction.
  
\pfcase{Base-True-R} This case is symmetric to \caseref{Base-True-L}.

\pfcase{Exists}
Here, we know that $M \vDash \Delta$ and $M \vDash \Delta'$. We also know that $\vec{X}$ is not free in $\Delta'$ with $\vec{Y}$ removed and vice versa. Now let:
\[
\vec Z = \vec X \cap \vec Y
\qquad
\vec X' = \vec X \setminus \vec Z
\qquad
\vec Y' = \vec Y \setminus \vec Z
\]
That is, $\vec Z$ is the variables occurring in both $\vec X$ and $\vec Y$, $\vec X'$ is the variables only occurring in $\vec X$ and $\vec Y'$ is the variables only occurring in $\vec Y$. This means that $\vec X'$, $\vec Y'$, and $\vec Z$ are disjoint.

We first show that $\exists \vec{X} \vec{Y}. M \vDash \exists \vec{X}. \Delta$. Suppose $(s,h) \vDash \exists \vec{X} \vec{Y}. M $. We know that $(s', h)\vDash M$ where $s' = s[\vec X' \mapsto \vec v_1][\vec Y' \mapsto \vec v_2][\vec Z \mapsto v_3]$ for some $\vec v_1, \vec v_2$ and $\vec v_3$.
Given that we know $M \vDash \Delta$, we now have  $(s',h) \vDash \Delta$ and therefore $(s[Y' \mapsto \vec v_2],h) \vDash \exists \vec{X}. \Delta$ (since $\vec X = \vec X' \cup \vec Z$). Now, given that $\vec Y \cap (\mathsf{fv}(\Delta)\setminus \vec X) = \emptyset$, we know that none of the variables in $\vec Y'$ are free in $\Delta$, and so we can remove them from the state to conclude that $(s,h)\vDash\exists \vec X.\Delta$.

It can also be shown that $\exists \vec X\vec Y.M \vDash \exists \vec Y.\Delta'$ by a symmetric argument.


\pfcase{Ls-Start-L} Here, we know $M \vDash \Delta\sep \ls(e_3,e_2)$ and $M \vDash \Delta'$. We now want to show $M \sep e_1 \mapsto e_3 \vDash \Delta \sep \ls(e_1,e_2)$ and $M \sep e_1\mapsto e_3 \vDash \Delta' \sep e_1\mapsto e_3$.

Suppose that $(s,h)\vDash M \sep e_1 \mapsto e_3$, and so $(s,h_1)\vDash M$ and $(s,h_2)\vDash e_1\mapsto e_3$ for some $h_1$ and $h_2$ such that $h=h_1\uplus h_2$. 

Since $M\vDash \Delta\sep \ls(e_3, e_2)$, we get that $(s,h_1)\vDash \Delta\sep\ls(e_3, e_2)$ and recombining, we get that $(s,h)\vDash \Delta\sep \ls(e_3, e_2)\sep e_1\mapsto e_3$. Now, $ \ls(e_3, e_2)\sep e_1\mapsto e_3 \vDash \exists X.e_1\mapsto X\sep \ls(X, e_2)$ and $\exists X.e_1\mapsto X\sep \ls(X, e_2)\vDash \ls(e_1, e_2)$, so we get that $(s,h)\vDash \Delta\sep \ls(e_1,e_2)$ and therefore $M\sep e_1\mapsto e_3\vDash \Delta\sep \ls(e_1,e_2)$.


We also know that $M \vDash \Delta'$ which means that because $(s,h_1) \vDash M$, $(s,h_1) \vDash \Delta'$. This means that $(s,h_1 \uplus h_2) \vDash \Delta'  \sep  e_1 \mapsto e_3$ and $h = h_1 \uplus h_2$, so we now have $(s, h) \vDash \Delta'  \sep  e_1 \mapsto e_3$, therefore $M\sep e_1\mapsto e_3 \vDash \Delta' \sep e_1\mapsto e_3$

\pfcase{Ls-Start-R} This case is symmetric to \caseref{Ls-Start-L}.

\pfcase{Match}
Here, we know that $M \vDash \Delta \land e_2 = e_3$ and $M \vDash \Delta' \land e_2 = e_3$. We want to show that $M \sep e_1 \mapsto e_2\vDash\Delta \sep e_1 \mapsto e_2$ and $M \sep e_1 \mapsto e_2\vDash\Delta' \sep e_1 \mapsto e_3$. 

Let us first show $M \sep e_1 \mapsto e_2 \vDash \Delta \sep e_1 \mapsto e_2$. Because we know that $M \vDash \Delta$ and $e_1 \mapsto e_2 \vDash e_1 \mapsto e_2$, it follows that $M \sep e_1 \mapsto e_2 \vDash \Delta \sep e_1 \mapsto e_2$.

Let us now show $M \sep e_1 \mapsto e_2 \vDash \Delta' \sep e_1 \mapsto e_3$. We know $M \vDash \Delta' \land e_2 = e_3$, now suppose that $(s,h) \vDash M \sep e_1 \mapsto e_2$, so $(s,h)\vDash \Delta'\land e_2=e_3$ as well. This means that $e_2 = e_3$ in state $s$, and therefore it must also be the case that $(s,h)\vDash e_1\mapsto e_3$. Given what else we know, we conclude that $(s,h)\vDash \Delta'\sep e_1\mapsto e_3$.
  
%
%

\pfcase{Ls-End-L} 
Here, we know that $M \vDash \Delta \sep \ls(e_3,e_2)$ and $M \vDash \Delta'$. We want to show that $M \sep \ls(e_1,e_3) \vDash\Delta \sep \ls(e_1,e_2)$ and $M \sep \ls(e_1,e_3)\vDash\Delta'\sep\ls(e_1,e_3)$. 

Let us first show $M \sep \ls(e_1,e_3) \vDash \Delta \sep \ls(e_1,e_2)$. 
First, we know that $M\vDash\Delta\sep\ls(e_3,e_2)$, and so $M\sep \ls(e_1,e_3)\vDash \Delta\sep\ls(e_3,e_2)\sep\ls(e_1,e_3)$. Clearly, it is also the case that $\ls(e_3,e_2)\sep\ls(e_1,e_3)\vDash \ls(e_1, e_2)$, and so combining these facts we get $M\sep \ls(e_1,e_3)\vDash \Delta\sep\ls(e_1,e_2)$.


Let us now show $M \sep \ls(e_1,e_3) \vDash \Delta'\sep\ls(e_1,e_3)$. We know that $M \vDash \Delta'$, so clearly $M\sep\ls(e_1,e_3)\vDash\Delta'\sep\ls(e_1,e_3)$.

\pfcase{Ls-End-R}
This case is symmetric to \caseref{Ls-End-L}.

\pfcase{Missing-L}
Here, we know that $M \vDash \Delta$ and $M \vDash \Pi' \land (\Sigma' \sep \tru)$. This means we also know that $M \vDash \Pi'$ and $M \vDash \Sigma' \sep \tru$ by semantic definition. 

Let us first show that $M \sep B(e_1,e_2) \vDash \Delta \sep B(e_1,e_2)$. Suppose that $(s,h)\vDash M\sep B(e_1,e_2)$.
We know that $(s,h_1)\vDash M$ and $(s,h_2)\vDash B(e_1,e_2)$ for some $h_1$ and $h_2$ such that $h = h_1 \uplus h_2$. Since $M \vDash \Delta$, then $(s,h_1) \vDash \Delta$. This means that $(s,h) \vDash \Delta \sep B(e_1,e_2)$, therefore $M \sep B(e_1,e_2) \vDash \Delta \sep B(e_1,e_2)$.

Let us now show that $M \sep B(e_1,e_2) \vDash \Pi' \sep (\Sigma' \sep \tru)$. Here, we know that $M \vDash \Pi' \sep (\Sigma' \sep \tru)$ and trivially $B(e_1,e_2) \vDash \tru$. This means $M \sep B(e_1,e_2) \vDash \Pi' \sep (\Sigma' \sep \tru) \sep \tru$; therefore, $M \sep B(e_1,e_2) \vDash \Pi' \sep (\Sigma' \sep \tru)$.

\pfcase{Missing-R}
This case is symmetric to \caseref{Missing-L}.

\pfcase{Emp-Ls-L}
We know that $M \vDash \Delta \land e_1 = e_2$ and $M \vDash \Delta' \land e_1 = e_2$. This means that $M \vDash \Delta'$, so the right side of the tri-abductive judgement is taken care of. 

Let us now show $M \vDash \Delta \sep \ls(e_1,e_2)$. We first establish that $e_1=e_2\land\emp \vDash\ls(e_1, e_2)$ by definition, since $\ls(e_1,e_2) \iff (\emp \land e_1=e_2) \vee \exists X. e_1 \mapsto X \sep \ls(X,e_2)$. Now, we know that $M\vDash \Delta\land e_1=e_2$, which also means that $M\vDash (\Delta \sep \emp) \land e_1=e_2$, or equivalently, $M\vDash (\Delta \land e_1= e_2) \sep (\emp \land e_1=e_2)$. Using $e_1=e_2\land\emp \vDash\ls(e_1, e_2)$, we get $M\vDash (\Delta \land e_1= e_2) \sep \ls(e_1,e_2)$, and by weakening we get $M\vDash \Delta \sep \ls(e_1,e_2)$

\pfcase{Emp-Ls-R} This case is symmetric to \caseref{Emp-Ls-L}

  
%
%
%
%

\end{enumerate}

\end{proof}

\triab*

\begin{proof}
In our tri-abduction algorithm, we call $\textsf{abduce-par}$ on $P \sep \tru$ and $Q \sep \tru$, so we know based on \Cref{lem:trij} that if $\trij {P \sep \tru}M{Q \sep \tru}$ is derivable, then $M\vDash P \sep \tru$ and $M\vDash Q \sep \tru$ since \textsf{abduce-par} operates by applying the inference in \Cref{fig:triab-pf}. The procedure for finding $F_1$ and $F_2$ follows that of \citet[\S5]{berdine2005symbolic} and so $M \vDash P \sep F_1$ and $M \vDash Q \sep F_2$ by \citet[Theorem 7]{berdine2005symbolic}.
\end{proof}

\section{Symbolic Execution}
\label{app:symexec}

\subsection{Renaming}

We first define the renaming procedure in \Cref{alg:rename}, which is identical to that of \citet[Fig. 4]{biabjacm} except that we additionally require $\vec e$ to be disjoint from $\vec x$. Renaming produces a new anti-frame $M_0$ which is similar to $M$ except that it is guaranteed not to mention any program variables and so it trivially meets the side condition of the frame rule. It additionally provides a vector of expressions $\vec e$ to be substituted for the free variables in the postcondition $\vec Y$ so as to match $M_0$.
\begin{algorithm}[t]
\caption{$\mathsf{rename}(\Delta, M, Q, \mathcal Q, \vec X, \vec x)$}\label{alg:rename}
\begin{algorithmic}
\State Let $\vec Y$ be the free logical variables of $Q$ and all the assertions in $\mathcal Q$.
\State Pick $\vec e$ disjoint from $\vec Y$ and $\vec x$ such that $\Delta\sep M \vDash \vec e = \vec Y$.
\State Pick $M'$ disjoint from $\vec X$, $\vec Y$, and $\mathsf{Var}$ such that $\Delta\sep M' \vDash \Delta \sep M[\vec e/\vec Y]$.
\State\Return $(\vec e, \vec Y, M')$
\end{algorithmic}
\end{algorithm}
Now, we recall the definitions of the following two procedures:
\[
\begin{array}{l}
\mathsf{biab}'(\exists \vec Z.\Delta, Q, \psi, \vec x) = \\
\quad
\big\{ (M', (\psi \osep \exists \vec Z\vec X.F[\vec X/\vec x])[\vec e/ \vec Y]) \\
\quad\mid (M, F) \in \mathsf{biab}(\Delta,Q) \\
\quad\quad (\vec e, \vec Y, M') = \mathsf{rename}(\Delta, M, Q, \{\psi\}, \vec Z) \big\}
\end{array}
\qquad
\begin{array}{l}
\mathsf{triab}'(P_1, P_2, \psi_1, \psi_2, \vec x) = \\
\quad
\big\{ (M', (\psi_1 \osep \exists \vec X.F_1[\vec X/\vec x])[\vec e/\vec Y], \\
\quad\phantom{\big\{ (M',\,} (\psi_2 \osep \exists \vec X.F_2[\vec X/\vec x])[\vec e/\vec Y]) \\
\quad\quad~\mid (M, F_1, F_2) \in \mathsf{triab}(P_1, P_2) \\
\quad\quad  (\vec e, \vec Y, M') = \mathsf{rename}(\emp, M, \{\psi_1, \psi_2\}, \emptyset) \big\}
\end{array}
\]
The $\mathsf{biab}'$ procedure is similar to \textsf{AbduceAndAdapt} from \citet[Fig. 4]{biabjacm}. Since the bi-abduction procedure does not support existentially quantified assertions on the left hand side, the existentials must be stripped and then re-added later (as is also done in \citet[Algorithm 4]{biabjacm}). The renaming step ensures that the anti-frame $M'$ is safe to use with the frame rule. We capture the motivation behind $\mathsf{biab}'$ using the following correctness lemma, stating that $\mathsf{biab}'$ produces a suitable frame and anti-frame so as to adapt a specification $\vDash\triple{\ok:Q}C\psi$ to use a different precondition $P$.

\begin{lemma}\label{lem:biab'}
For all $(M, \psi') \in \mathsf{biab}'(P, Q, \psi, \vec x)$, if $\vDash\triple{\ok:Q}C{\psi}$ and $\vec x = \mathsf{mod}(C)$, then
\[
\vDash\triple{\ok:P\sep M}C{\psi'}
\]
\end{lemma}

\begin{proof}
By definition, any element of $\mathsf{biab}'(P, Q, \psi,\vec x)$ (where $P = \exists \vec Z.\Delta$) must have the form $(M', (\psi \osep \exists \vec Z\vec X.F[\vec X/\vec x])[\vec e/ \vec Y])$ where $(\vec e, \vec Y, M') = \mathsf{rename}(\Delta, M, Q, \{\psi\}, \vec X)$ and $(M, F) \in \mathsf{biab}(\Delta,Q)$.
By the definition of rename, we know that:
\begin{align}
\Delta\sep M' &\vDash \Delta \sep M[\vec e/\vec Y] \nonumber\\
\intertext{Since $M'$ is assumed to be disjoint from $\vec Z$, then $\exists \vec Z.M'$ iff $M'$, so we can existentially quantify both sides to obtain:}
P\sep M' &\vDash \exists\vec Z.\Delta\sep M[\vec e/\vec Y] \label{eq:entailment-lhs} \\
\intertext{In addition, $(M,F)\in \mathsf{biab}(\Delta, Q)$, so we also know that:}
\Delta\sep M&\vDash Q\sep F \nonumber
\intertext{In \Cref{fig:sym_exec}, we assumed all the logical variables used are fresh, so $\Delta$ must be disjoint from $\vec Y$ (the free variables of $Q$ and $\psi$), and therefore $\Delta[\vec e/\vec Y] = \Delta$, so substituting both sides, we get:}
\Delta\sep M[\vec e/\vec Y] & \vDash (Q\sep F)[\vec e/\vec Y] \nonumber \\
\intertext{We also weaken the right hand side by replacing $\vec x$ with fresh existentially quantified variables in $F$.}
\Delta\sep M[\vec e/\vec Y]&\vDash  (Q\sep \exists\vec X.F[\vec X/\vec x])[\vec e/\vec Y] \nonumber
\intertext{Now, we can existentially quantify both sides of the entailment. Since logical variables are fresh, $\vec Z$ is disjoint from $Q$.}
\exists\vec Z.\Delta\sep M[\vec e/\vec Y]&\vDash  (Q\sep \exists\vec Z\vec X.F[\vec X/\vec x])[\vec e/\vec Y] \label{eq:entailment-rhs} \\
\intertext{And finally, we combine \Cref{eq:entailment-lhs,eq:entailment-rhs} to get:}
P\sep M' &\vDash (Q\sep \exists\vec Z\vec X.F[\vec X/\vec x])[\vec e/\vec Y] \label{eq:entailment-final}
\end{align}
Now, given that $\vDash\triple{\ok:Q}C{\psi}$, we can use the frame rule to get:
\[
\vDash\triple{\ok:Q\sep \exists\vec Z\vec X.F[\vec X/\vec x]}C{\psi\osep \exists\vec Z\vec X.F[\vec X/\vec x]}
\]
This is clearly valid, since $\vec x = \mathsf{mod}(C)$ has been removed from the assertion that we are framing in, therefore satisfying $\mathsf{mod}(C) \cap \mathsf{fv}(\exists\vec Z\vec X.F[\vec X/\vec x]) = \emptyset$.
We can also substitute $\vec e$ for $\vec Y$ in the pre- and postconditions since we assumed that $\vec e$ is disjoint from the program variables $\vec x$, and therefore $\vec e$ remains constant after executing $C$.
\[
\vDash\triple{\ok:(Q\sep \exists\vec Z\vec X.F[\vec X/\vec x])[\vec e/\vec Y]}C{(\psi\osep \exists\vec Z\vec X.F[\vec X/\vec x])[\vec e/\vec Y]}
\]
Finally, using the rule of consequence with \Cref{eq:entailment-final}, we strengthen the precondition to get:
\[
\vDash\triple{\ok:P\sep M'}C{(\psi\osep \exists\vec Z\vec X.F[\vec X/\vec x])[\vec e/\vec Y]}
\]

\end{proof}

The $\mathsf{triab}'$ procedure is similar, but it is fundamentally based on tri-abduction and is accordingly used for parallel composition instead of sequential composition. We include two separate proofs corresponding to the two ways in which tri-abduction is using during symbolic execution. The first (\Cref{lem:triab'1}) pertains to merging the anti-frames obtained by continuing to evaluate a single program $C$ after the control flow has already branched whereas the second (\Cref{lem:triab'2}) deals with merging the preconditions from two different program program branches, $C_1$ and $C_2$.

\begin{lemma}\label{lem:triab'1}
If $(M, \psi'_1, \psi'_2) \in \mathsf{triab}'(M_1, M_2, \psi_1, \psi_2, \vec x)$ and $\vDash\triple{\varphi_1\osep M_1}{C}{\psi_1}$ and $\vDash\triple{\varphi_2\osep M_2}{C}{\psi_2}$ and $\vec x = \mathsf{mod}(C)$, then $\vDash\triple{\varphi_1\osep M}{C}{\psi'_1}$ and $\vDash\triple{\varphi_2\osep M}{C}{\psi'_2}$.
\end{lemma}

\begin{proof}
By definition, any element of $\mathsf{triab}'(M_1, M_2, \psi_1, \psi_2, \vec x)$ will have the form:
\[
\left(M', (\psi_1\osep \exists \vec X.F_1[\vec X/\vec x])[\vec e/\vec Y],(\psi_2\osep \exists \vec X.F_2[\vec X/\vec x])[\vec e/\vec Y]\right)
\]
Where $(\vec e, \vec Y, M') = \mathsf{rename}(\emp, M, \{\psi_1, \psi_2\}, \emptyset, \vec x)$ and $(M, F_1, F_2) \in \mathsf{triab}(P_1, P_2)$. From rename, we know that $M' \vDash M[\vec e/\vec Y]$ and from tri-abduction, we know that $M\vDash M_i\sep F_i$ for $i=1,2$, so $M'\vDash (M_i \sep F_i)[\vec e/\vec Y]$. We can weaken this by replacing $\vec x$ in $F_i$ with fresh existentially quantified variables to obtain $M'\vDash (M_i \sep \exists\vec X.F_i[\vec X/\vec x])[\vec e/\vec Y]$.
By assumption, we know that $\vDash\triple{\varphi_i\osep M_i}{C}{\psi_i}$ for $i=1,2$. So, using the frame rule, we get:
\[
\vDash\triple{\varphi_i\osep (M_i\sep \exists\vec X.F_i[\vec X/\vec x])}{C}{\psi_i\osep \exists\vec X.F_i[\vec X/\vec x]}
\]
This is valid since $\exists\vec X.F_i[\vec X/\vec x]$ is disjoint from $\vec x$ (the modified program variables) by construction.
We also assumed in \textsf{rename} that $\vec e$ is disjoint from $\vec x$, so we can substitute $\vec e$ for $\vec Y$ to get:
\[
\vDash\triple{(\varphi_i\osep (M_i\sep \exists\vec X.F_i[\vec X/\vec x]))[\vec e/\vec Y]}{C}{(\psi_i\osep \exists\vec X.F_i[\vec X/\vec x]i)[\vec e/\vec Y]}
\]
Note that $\varphi_i[\vec e/\vec Y] = \varphi_i$, since the logical variables $\vec Y$ are generated freshly, independent of $\varphi_i$, as was mentioned in \Cref{fig:sym_exec}. So, using the rule of consequence we get:
\[
\vDash\triple{\varphi_i\osep M}{C}{(\psi_i\osep \exists\vec X.F_i[\vec X/\vec x])[\vec e/\vec Y]}
\]
\end{proof}

\begin{lemma}\label{lem:triab'2}
If $(M, \psi'_1, \psi'_2) \in \mathsf{triab}'(P_1, P_2, \psi_1, \psi_2, \vec x)$ and $\vDash\triple{\ok:P_1}{C_1}{\psi_1}$ and $\vDash\triple{\ok:P_2}{C_2}{\psi_2}$ and $\vec x = \mathsf{mod}(C_1,C_2)$, then $\vDash\triple{\ok:M}{C_1}{\psi'_1}$ and $\vDash\triple{\ok:M}{C_2}{\psi'_2}$.
\end{lemma}

\begin{proof}

By definition, any element of $\mathsf{triab}'(P_1, P_2, \psi_1, \psi_2, \vec x)$ will have the form:
\[
\left(M', (\psi_1\osep \exists \vec X.F_1[\vec X/\vec x])[\vec e/\vec Y],(\psi_2\osep \exists \vec X.F_2[\vec X/\vec x])[\vec e/\vec Y]\right)
\]
Where $(\vec e, \vec Y, M') = \mathsf{rename}(\emp, M, \{\psi_1, \psi_2\}, \emptyset, \vec x)$ and $(M, F_1, F_2) \in \mathsf{triab}(P_1, P_2)$. From rename, we know that $M' \vDash M[\vec e/\vec Y]$ and from tri-abduction, we know that $M\vDash P_i\sep F_i$ for $i=1,2$, so $M'\vDash (P_i \sep F_i)[\vec e/\vec Y]$. We can weaken this by replacing $\vec x$ in $F_i$ with fresh existentially quantified variables to obtain $M'\vDash (P_i \sep \exists\vec X.F_i[\vec X/\vec x])[\vec e/\vec Y]$

By the frame rule, we know that $\vDash\triple{\ok:P_i\sep \exists\vec X.F_i[\vec X/\vec x]}{C_i}{\psi_i\osep\vec X.F_i[\vec X/\vec x]}$ since $\vec X.F_i[\vec X/\vec x]$ must be disjoint from the modified program variables $\vec x$.
By substituting into both the pre and postconditions, we get $\vDash\triple{(\ok:P_i\sep \vec X.F_i[\vec X/\vec x])[\vec e/\vec Y]}{C_i}{(\psi_i\osep \vec X.F_i[\vec X/\vec x])[\vec e/\vec Y]}$ (this is valid since \textsf{rename} guarantees that $\vec e$ is disjoint from $\vec x$). Finally, we complete the proof by applying the rule of consequence with $M'\vDash (P_i \sep \vec X.F_i[\vec X/\vec x])[\vec e/\vec Y]$ to obtain:
\[
\vDash\triple{\ok:M'}{C_i}{(\psi_i\osep \vec X.F_i[\vec X/\vec x])[\vec e/\vec Y]}
\]
\end{proof}

\subsection{Sequencing Proof}
\sequencing*
\begin{proof}
By induction on the structure of $\varphi$.
\begin{itemize}
\item $\varphi=\top$. We need to show that $\vDash\triple{\top\osep \emp}C{\top}$ holds. This triple is clearly valid since any triple with the postcondition $\top$ is trivially true.
\item $\varphi = \varphi_1\bowtie \varphi_2$ where $\mathord{\bowtie} \in \{ \vee, \oplus \}$. We need to show:
\[
\vDash\triple{(\varphi_1\bowtie\varphi_2)\osep M}{C}{\psi'_1 \bowtie \psi'_2}
\]
Where $(M,\psi'_1,\psi'_2) \in \mathsf{triab}'(M_1,M_2,\psi_1,\psi_2, \vec x)$ and $(M_i,\psi_i)\in \mathsf{seq}(\varphi_i,S,\vec x)$ for each $i\in\{1,2\}$.
By the induction hypothesis, we know that $\vDash\triple{\varphi_i\osep M_i}C{\psi_i}$, so by \Cref{lem:triab'1} we get that $\vDash\triple{\varphi_i\osep M}C{\psi'_i}$.
We now complete the proof separately for the two logical operators:
\begin{itemize}
\item $\varphi = \varphi_1 \vee \varphi_2$. Suppose that $m\vDash (\varphi_1\vee\varphi_2) \osep M$, so $m\vDash\varphi_i\osep M$ for some $i\in\{1,2\}$. Since $\vDash\triple{\varphi_i\osep M}{C}{\psi'_i}$, we know that $\dem{C}m\af\vDash \psi'_i$, and we can weaken this to conclude that $\dem{C}m\af\vDash \psi'_1\vee\psi'_2$.


\item $\varphi = \varphi_1 \oplus \varphi_2$. Suppose that $m\vDash (\varphi_1 \oplus \varphi_2)\osep M$, and so there are $m_1$ and $m_2$ such that $m = m_1 + m_2$ and $m_i\vDash \varphi_i\osep M$ for each $i$. Since $\vDash\triple{\varphi_i\osep M}{C}{\psi'_i}$, we know that $\dem{C}{m_i}\af\vDash \psi'_i$. We also know that $\dem{C}m\af = \dem{C}{m_1 + m_2}\af = \dem C{m_1}\af + \dem{C}{m_2}{\af}$, and so $\dem Cm\af\vDash \psi'_1\oplus \psi'_2$.

\end{itemize}

\item $\varphi = \wg{\varphi'}a$. We need to show that $\vDash\triple{\wg{\varphi'}a \osep M}C{\wg{\psi}a}$ where $(M, \psi) \in \mathsf{seq}(\varphi', S, \vec x)$. By the induction hypothesis, we know that $\vDash\triple{\varphi'\osep M}C{\psi}$. Now, suppose that $m\vDash \wg{\varphi'}a \osep M$. If $a=\zero$, then $m=\zero$ and therefore $\dem Cm\af\vDash \wg\psi\zero$. If not, then there is some $m'$ such that $m'\vDash \varphi' \osep M$ and $m = a\cdot m'$. So, $\dem C{m'}{\af}\vDash \psi'$. We also know that $\dem{C}m\af = \dem{C}{a\cdot m'}\af = a\cdot\dem{C}{m'}\af$, so by the semantics of $\wg{(-)}a$, $\dem Cm\af\vDash \wg\psi{a}$.

\item $\varphi = \ok:P$. We need to show that $\vDash\triple{\ok:P\sep M}C{\psi'}$ where $(M,\psi')\in\mathsf{biab}'(P,Q,\psi, \vec x)$ and $(Q,\psi) \in S$.
By assumption, we know that $\vDash\triple{\ok:Q}{C}{\psi}$. The remainder of the proof follows directly from \Cref{lem:biab'}.

\item $\varphi = \er:Q$. We need to show that $\vDash\triple{\er:Q}C{\er:Q}$. This trivially holds since any $m$ satisfying $\er:Q$ must consist only of $\inj_\er(s,h)$ states, and so $\dem Cm\af = m$.
\end{itemize}
\end{proof}

\subsection{Symbolic Execution Proofs}

\begin{lemma}\label{lem:soundness-loops}
Let:
\[\small
f(S) = \begin{array}{l}
\{ (\lnot e\land \emp, \ok:\lnot e \land \emp)\} ~\cup \\
 \{ (M_1\sep M_2 \land e, \psi) \mid (M_1,\varphi)\in\mathsf{seq}(\ok:e\land\emp, \dea{C}(T), \mathsf{mod}(C)), (M_2, \psi)\in\mathsf{seq}(\varphi, S, \mathsf{mod}(C)) \}
 \end{array}
\]
For any $n\in\mathbb{N}$ and $(P, \varphi) \in f^n(\emptyset)$, $\vDash\triple{\ok:P}{\whl eC}{\varphi}$.
\end{lemma}
\begin{proof}
By induction on $n$. Suppose $n=0$, then $f^0(\emptyset) = \emptyset$, so the claim vacuously holds. Now, suppose the claim holds for $n$, we will show it holds for $n+1$. First, observe that:
\begin{align*}
f^{n+1}(\emptyset) &= f(f^n(\emptyset)) \\
&= \{ (\lnot e\land \emp, \ok:\lnot e \land \emp)\} ~\cup \\
& \quad
\arraycolsep=0pt
\begin{array}{ll}
\{ (M_1\sep M_2 \land e, \psi) \mid \;\;& (M_1,\varphi)\in\mathsf{seq}(\ok:e\land\emp, \dea{C}(T), \mathsf{mod}(C)) \\
& (M_2, \psi)\in\mathsf{seq}(\varphi, f^n(\emptyset), \mathsf{mod}(C)) \}
\end{array}
\end{align*}
So any $(P,\varphi)\in f^{n+1}(\emptyset)$ comes from one of the two sets in the above union. Suppose it is in the first, so we need to show that $\vDash\triple{\ok:\lnot e\land\emp}{\whl eC}{\ok:\lnot e \land\emp}$. This is clearly true, since the loop does not execute in states where $\lnot e$ holds and therefore the whole command is equivalent to $\skp$.

Now suppose we are in the second case, so the element has the form $(M_1\sep M_2\land e, \psi)$ where $(M_1,\varphi) \in\mathsf{seq}(\ok:e\land\emp, \dea{C}(T), \mathsf{mod}(C))$ and $(M_2,\psi)\in\mathsf{seq}(\varphi, f^n(\emptyset), \mathsf{mod}(C))$. By \Cref{lem:sequencing}, we know that $\vDash\triple{\ok:M_1\land e}{C}{\varphi}$ and by \Cref{lem:sequencing} and the induction hypothesis, we get $\vDash\triple{\varphi\osep M_2}{\whl eC}{\psi}$. Using the frame rule, we also get $\vDash\triple{\ok:M_1\sep M_2\land e}{C}{\varphi\osep M_2}$, and so we can sequence the previous specifications to get $\vDash\triple{\ok:M_1\sep M_2\land e}{C\fatsemi \whl eC}{\psi}$. Now, since the precondition stipulates that $e$ is true, the loop must run for at least one iteration, so for any $m\vDash (\ok:M_1\sep M_2\land e)$, $\dem{C\fatsemi \whl eC}m{\af} = \dem{\whl eC}m{\af}$, and so $\vDash\triple{\ok:M_1\sep M_2\land e}{\whl eC}{\psi}$.
\end{proof}

\begin{lemma}\label{lem:lifting}
If for every $(s,h)\vDash P$ and $\af\in\mathsf{Alloc}$, there exists $s'$ and $t'$ such that $\de{C}_\af(s,h) = \unit_{\mathcal W}(\inj_\epsilon(s',t'))$ and $(s',h')\vDash Q$, then $\vDash\triple{\ok:P}C{\epsilon:Q}$
\end{lemma}
\begin{proof}
Suppose that $m\vDash\ok:P$. That means that $|m| = \one$ and all elements of $\supp(m)$ have the form $\inj_\ok(s,h)$ where $(s,h)\vDash P$. By assumption, we know that $\de{C}_\af(s,h) =\unit_{\mathcal W}(\inj_\epsilon(s',t'))$ such that $(s',t')\vDash Q$. This means that every element of $\dem{C}m{\af}$ must have the form $\inj_\epsilon(s',t')$ where $(s',t')\vDash Q$ and also $|\dem{C}m{\af}| = \one$ since each $(s,h)$ does not change the mass of the distribution, so $\dem{C}m{\af}\vDash \epsilon:Q$.
\end{proof}

\soundness*

\begin{proof}
By induction on the structure of the program $C$.
\begin{itemize}
\item $C = \skp$. We need to show that $\vDash\triple{\ok:\emp}\skp{\ok:\emp}$, which is trivially true.
\item $C = C_1\fatsemi C_2$. By definition, any element of $\dea{C_1\fatsemi C_2}(T)$ must have the form $(P\sep M, \psi)$ where $(P, \varphi) \in \dea{C_1}(T)$ and $(M, \psi)\in\mathsf{seq}(\varphi, \dea{C_2}(T), \mathsf{mod}(C_2))$.
By the induction hypothesis, we know that $\vDash\triple{\ok:P}{C_1}{\varphi}$ and by \Cref{lem:sequencing} we know that $\vDash\triple{\varphi\osep M}{C_2}{\psi}$. Using the frame rule, we get that $\vDash\triple{\ok:P\sep M}{C_1}{\varphi\osep M}$ (given the renaming step used in \textsf{seq}, $M$ contains no program variables, so it must obey the side condition of the frame rule). Finally, we can join the two specifications to conclude that $\vDash\triple{\ok:P\sep M}{C_1\fatsemi C_2}{\psi}$.


\item $C = C_1 + C_2$. Any element of $\dea{C_1+ C_2}(T)$ must have the form $(M, \psi'_1 \oplus \psi'_2)$ where $(M,\psi'_1,\psi'_2)\in\mathsf{triab}'(M_1, M_2, \psi_1, \psi_2, \mathsf{mod}(C_1, C_2))$ and $(M_1, \psi_1)\in\dea{C_1}(T)$ and $(M_2, \psi_2)\in\dea{C_2}(T)$. By the induction hypothesis, we know that $\vDash\triple{\ok:M_i}{C_i}{\psi_i}$ for $i=1,2$. By \Cref{lem:triab'2} we know that $\vDash\triple{\ok:M}{C_i}{\psi'_i}$. Now, we show that $\vDash\triple{M}{C_1+ C_2}{\psi'_1\oplus\psi'_2}$: suppose $m\vDash M$. By definition, $\dem{C_1+ C_2}m{\af} = \dem{C_1}m{\af} + \dem{C_2}m{\af}$. Now, using what we obtained from \Cref{lem:triab'2}, we know that since $m\vDash M$, $\dem{C_i}m{\af}\vDash \psi'_i$ for each $i\in\{1,2\}$. Combining these two, we get that $\dem{C_1}m{\af} + \dem{C_2}m{\af} \vDash \psi'_1 \oplus \psi'_2$.

\item $C = \assume b$. Any element of $\dea{\assume b}(T)$ must have the form $(b\land\emp, \ok:b\land \emp)$ or $(\lnot b\land \emp, \wg\top\zero)$.
In the first case, we need to show $\vDash\triple{\ok:b\land \emp}{\assume b}{\ok:b\land\emp}$. Suppose $m\vDash \ok:b\land\emp$, then it's easy to see that $\dem{\assume b}m\af = m$, so the triple is valid. In the second case, we must show $\vDash\triple{\ok:\lnot b\land \emp}{\assume b}{\wg\top\zero}$. Suppose $m\vDash \ok:\lnot b\land \emp$, so $\dem{\assume b}m\af = \zero$, and $\zero\vDash\wg\top\zero$.

\item $C = \assume a$. Any element of $\dea{\assume a}(T)$ must have the form $(\emp, \wg{(\ok:\emp)}a)$, so we need to show $\vDash\triple{\ok:\emp}{\assume a}{\wg{(\ok:\emp)}a}$. Suppose $m\vDash\ok:\emp$, so we know that $\dem{\assume a}m\af = a\cdot m$, therefore $\dem{\assume a}m\af\vDash \wg{(\ok:\emp)}a$.


\item $C = \whl eC$. By the Kleene fixed point theorem, $\dea{\whl eC}(T) = \bigcup_{n\in\mathbb{N}}f^n(\emptyset)$ where $f(S)$ is defined as in \Cref{lem:soundness-loops}. So, any $(P,\varphi)\in \dea{\whl eC}(T)$ must also be an element of $f^n(\emptyset)$ for some $n$. We complete the proof by applying \Cref{lem:soundness-loops}.

\end{itemize}
The remaining cases are for primitive instructions, most of which are \emph{pure}, meaning that each program state maps to a single outcome according to the program semantics. In these cases, it suffices to show that if $(P,\varphi)\in \dea{c}(T)$, then $\de{c}_\af(s,h)\vDash\varphi$ for all $(s,h)\vDash P$ by \Cref{lem:lifting}.
\begin{itemize}
\item $C = (x\coloneqq e)$. Suppose that $(s,h)\vDash \ok:x=X\land\emp$, so $s(x) = s(X)$ and $h = \emptyset$. Now, $\de{x\coloneqq e}_\af(s,h) = \unit(s[x \mapsto \de{e}(s)], h)$, so let $s' = s[x \mapsto \de{e}(s)]$.
Clearly, $\de{e}(s) = \de{e[X/x]}(s)$ since $s(x) = s(X)$. It must also be the case that $\de{e[X/x]}(s) = \de{e[X/x]}(s')$ since $s$ and $s'$ differ only in the values of $x$, and $x$ does not appear in $e[X/x]$. So, $s'(x) = \de{e}(s) = \de{e[X/x]}(s) = \de{e[X/x]}(s')$, and therefore $(s',h)\vDash \ok:x=e[X/x]\land\emp$. The remainder of the proof follows by \Cref{lem:lifting}.

\item $C = (x\coloneqq\mathsf{alloc}())$. We need to show that $\vDash\triple{\ok:\emp\land x=X}{x\coloneqq\mathsf{alloc}()}{\ok:\exists Y.x\mapsto Y}$. Suppose that $m\vDash\ok:\emp\land x=X$, so each state in $m$ has the form $\inj_\ok(s,h)$ where $(s,h)\vDash \emp\land x=X$, so $s(x) =s(X)$ and $h = \emptyset$. We know that $\de{x\coloneqq\mathsf{alloc}}_\af(s,h) = \bind_{\mathcal W}(\mathsf{alloc}(s,h), \lambda(\ell,v). \unit(s[x\mapsto \ell], h[\ell\mapsto v]))$ where $\ell \notin \mathsf{dom}(h)$. Let $s' = s[x\mapsto \ell]$ and $h' = h[\ell\mapsto v]$. Clearly, $h'(\de{x}(s')) = h'(\ell) = v$, so $(s',h')\vDash\exists Y.x\mapsto Y$. Since this is true for all end states, and since alloc does not alter the total mass of the distribution, then $\dem{x\coloneqq\mathsf{alloc}()}m{\af}\vDash \ok:\exists Y.x\mapsto Y$.

\item $C = \mathsf{free}(x)$. There are three cases since specifications for free can take on multiple forms. In the first case, we need to show that $\vDash\triple{\ok:e\mapsto X}{\mathsf{free}(e)}{\ok:e\not\mapsto}$. Suppose $(s,h)\vDash e\mapsto X$, so $h(\de{e}(s)) = s(X)$. We also know that $\de{\mathsf{free}(e)}_\af(s,h) = \unit(s, h[\de{e}(s) \mapsto \bot])$, and since $(s, h[\de{e}(s) \mapsto \bot]) \vDash e\not\mapsto$, the claim follows by \Cref{lem:lifting}.

In the other cases, we need to show that:
\[
\vDash\triple{\ok:e\not\mapsto}{\mathsf{free}(e)}{\er:e\not\mapsto}
\quad\text{and}\quad
\vDash\triple{\ok:e=\mathsf{null}}{\mathsf{free}(e)}{\er:e=\mathsf{null}}
\]
Suppose $(s,h)\vDash e\not\mapsto$ and so $h(\de{e}(s)) = \bot$. Clearly, $\de{\mathsf{free}(e)}_\af(s,h) = \mathsf{error}(s,h)$, so by \Cref{lem:lifting} the claim holds. The case where $e = \mathsf{null}$ is nearly identical.

\item $C = [e_1]\leftarrow e_2$. There are three cases, first we must show that $\vDash\triple{\ok:e_1\mapsto X}{[e_1]\leftarrow e_2}{\ok:e_1\mapsto e_2}$. Suppose $(s,h)\vDash e_1\mapsto X$, so $h(\de{e_1}(s)) = s(X)$ and therefore that memory address is allocated since $s(X)\in\mathsf{Val}$. This means that $\de{[e_1]\leftarrow e_2}_\af(s,h) = \unit(s, h[\de{e_1}(s) \mapsto \de{e_2}(s)])$ and clearly $(s, h[\de{e_1}(s) \mapsto \de{e_2}(s)])\vDash e_1\mapsto e_2$ by definition, so the claim holds by \Cref{lem:lifting}.

In the remaining case, we must show that $\vDash\triple{\ok:e_1\not\mapsto}{[e_1]\leftarrow e_2}{\er: e_1\not\mapsto}$ and $\vDash\triple{\ok:e_1=\mathsf{null}}{[e_1]\leftarrow e_2}{\er: e_1=\mathsf{null}}$. The proof is similar to the second case for $\mathsf{free}(e)$.

\item $C = x\leftarrow [e]$. There are three cases, first we must show that $\vDash\triple{\ok:x=X\land e\mapsto Y}{x\leftarrow[e]}{\ok:x=Y\land e[X/x]\mapsto Y}$. Suppose that $(s,h)\vDash x=X\land e\mapsto Y$, so $s(x) = s(X)$ and $h(\de{e}(s)) = s(Y)$. We also know that $\de{x\leftarrow [e]}_\af(s,h) = \unit(s[x\mapsto h(\de{e}(s))], h)$. Let $s' = s[x\mapsto h(\de{e}(s))]$, as we showed in the $x\coloneqq e$ case, $\de{e}(s) = \de{e[X/x]}(s) = \de{e[X/x]}(s')$. So, this means that $h(\de{e[X/x]}(s')) = h(\de{e}(s)) = s(Y)$ and $s'(x) = \de{e}(s) = s(Y)$, so clearly $(s',h)\vDash x=Y\land e[X/x]\mapsto Y$ and the claim follows from \Cref{lem:lifting}.

In the remaining case, we need to show that $\vDash\triple{\ok:e\not\mapsto}{x\leftarrow [e]}{\er: e\not\mapsto}$ and  $\vDash\triple{\ok:e=\mathsf{null}}{x\leftarrow [e]}{\er: e=\mathsf{null}}$. The proof is similar to the second case for $\mathsf{free}(e)$.

\item $C = \mathsf{error}()$. We need to show that $\vDash\triple{\ok:\emp}{\mathsf{error}()}{\er:\emp}$. Suppose $(s,h)\vDash\emp$. Clearly, $\de{\mathsf{error}()}_\af(s,h) = \mathsf{error}(s,h)$, and so the claim holds by \Cref{lem:lifting}.

\item $C = f(\vec e)$. Any element of $\dea{f(\vec e)}(T)$ must have the form $(P\land \vec x = \vec X)$ where $(P, \varphi) \in \mathsf{seq}(\ok:\vec x = \vec e[X/x], T(f(\vec x)), \mathsf{mod}(f))$. By \Cref{lem:sequencing}, we get:
\[
\vDash\triple{\ok:P \land \vec x = \vec e[X/x]}{f(\vec x)}{\varphi}
\]
Now, we need to show that $\vDash\triple{\ok:P\land\vec x=\vec X}{f(\vec e)}{\varphi}$.
Suppose that $m\vDash\ok:P\land \vec x= \vec X$. Let $m'$ be obtained by taking every state $\inj_\ok(s,h) \in \supp(m)$ and modifying it to be $\inj_\ok(s[\vec x\mapsto \de{\vec e}(s)], h)$. By a similar argument to the $x\coloneqq e$ case, we know that $m'\vDash\ok: \vec x=\vec e[\vec X/\vec x]$. We know $P$ is disjoint from the program variables by the definition of $\mathsf{seq}$, so $m'\vDash \ok:P$ as well, since $m\vDash \ok:P$ and the only difference between $m$ and $m'$ is updates to the program variables.
Let $C$ be the body of $f$ and note that $\dem{f(\vec x)}{m'}{\af} = \dem{C}{m'}{\af}$ since the initial modification of the program state is just updating variable values to themselves.
We know from $\vDash\triple{\ok:P \land \vec x = \vec e[X/x]}{f(\vec x)}{\varphi}$ that $\dem{C}{m'}{\af}\vDash\varphi$ and by definition, $\dem{f(\vec e)}m{\af} = \dem{C}{m'}{\af}$, so $\dem{f(\vec e)}{m}{\af}\vDash\varphi$.
\end{itemize}
In addition, we show that the two refinements for single-path computation and loops invariants are sound too:
\begin{itemize}

\item Single Path. Any element of $\dea{C_1+ C_2}(T)$ has one of two forms. In the first case, we need to show that $\vDash\triple{\ok:P}{C_1+ C_2}{\varphi \oplus\top}$ given that $\vDash\triple{\ok:P}{C_1}{\varphi}$. Suppose that $m\vDash \ok:P$. By our assumption, we know that $\dem{C_1}m{\af}\vDash\varphi$. Now, $\dem{C_1+ C_2}m{\af} = \dem{C_1}m{\af} + \dem{C_2}m{\af}$ and clearly $\dem{C_2}m{\af}\vDash\top$, so $\dem{C_1+ C_2}m{\af}\vDash \varphi\oplus\top$. The second case is symmetrical, using the fact that $+$ is commutative.

\item Loop invariants. We need to show that $\vDash\triple{\ok:I}{\whl eC}{(\ok: I\land\lnot e)\vee(\top)_\zero}$ given that $\vDash\triple{\ok:I\land e}C{\ok:I}$. Note that this case is only valid for deterministic or nondeterministic programs (not probabilistic ones). Suppose $m\vDash \ok:I$, so every state in $\supp(m)$ has the form $\inj_\ok(s,h)$ where $(s,h)\vDash I$. By assumption, we know that every execution of the loop body will preserve the truth of $I$, so either all the states in $\de{\whl eC}_\af(s,h)$ must satisfy $I$ and $\lnot e$, or there are no terminating states. In other words, $\de{\whl eC}_\af(s,h)\vDash (\ok:I\land\lnot e)\vee\wg\top\zero$. In the deterministic case, we are done since there can only be a single start state. In the nondeterministic case, each start state $(s,h)$ leads to a set of end states satisfying $(\ok:I\land\lnot e)\vee\wg\top\zero$, then the union of all these states will also satisfy $(\ok:I\land\lnot e)\vee\wg\top\zero$.
\end{itemize}

\end{proof}
\fi

\end{document}
\endinput